\documentclass[aihp]{imsart}
\pdfoutput=1

\RequirePackage{amsthm,amsmath,amsfonts,amssymb}
\RequirePackage[numbers]{natbib}
\RequirePackage[colorlinks,citecolor=blue,urlcolor=blue]{hyperref}
\RequirePackage{graphicx}

\usepackage[utf8]{inputenc}
\usepackage{wasysym} 
\usepackage{tikz}

\usepackage{csquotes}
\MakeOuterQuote{"}
\usepackage[normalem]{ulem} 

\startlocaldefs
\numberwithin{equation}{section}
\theoremstyle{plain}

\newtheorem{theorem}{Theorem}[section]
\newtheorem{lemma}[theorem]{Lemma}
\newtheorem{prop}[theorem]{Proposition}
\theoremstyle{remark}
\newtheorem{definition}[theorem]{Definition}
\newtheorem{remark}[theorem]{Remark}

\def\eps{\epsilon}
\newcommand{\beq}{\begin{equation}} 
\newcommand{\eeq}{\end{equation}}
 
\def\nn{\nonumber}

\def\ge{\geqslant}
\def\le{\leqslant}

\def\<{\langle}
\def\>{\rangle}

\def\oB{\overline{B}}

\newcommand{\myraise}[1]{\raisebox{0.6ex}{\raisebox{-0.5\height}{#1}}}
\newcommand{\myinclude}[2][]{\raisebox{0.6ex}{\raisebox{-0.5\height}{\includegraphics[#1]{#2}}}}

\renewcommand{\epsilon}{\varepsilon}
\newcommand{\normfactor}{\nu}

\newcommand{\tmop}[1]{\ensuremath{\operatorname{#1}}}

\usepackage{enumitem}

\def\HLT{\widetilde{\mathcal{H}}}
\def\H0{\widetilde{\mathcal{H}}}
\renewcommand{\Re}{\mathop{\rm Re}}
\renewcommand{\Im}{\mathop{\rm Im}}
\renewcommand{\O}{\Omega}

\endlocaldefs

\makeatletter
\def\journal@name{\ }
\def\journal@issn{\ }
\def\journal@url{\ }
\makeatother

\begin{document}

\begin{frontmatter}
\title{Tensor Renormalization Group at Low Temperatures: Discontinuity Fixed Point}

\begin{aug}
\author[A]{\fnms{Tom} \snm{Kennedy}\ead[label=e1]{tgk@math.arizona.edu}\orcid{0000-0002-9194-9961}},
\author[B]{\fnms{Slava} \snm{Rychkov}\ead[label=e2]{slava@ihes.fr}\orcid{0000-0002-5847-1011}} %
\address[A]{Department of Mathematics, University of Arizona,
	Tucson, AZ 85721, USA\printead[presep={,\ }]{e1}}

\address[B]{Institut des Hautes \'Etudes Scientifiques,
	91440 Bures-sur-Yvette, France\printead[presep={,\ }]{e2}}
\end{aug}

\begin{abstract}

\begin{center}
	\it To the memory of Krzysztof Gawędzki, a pioneer of rigorous renormalization group studies	\end{center}

\noindent We continue our study of rigorous renormalization group (RG) maps
for tensor networks that was begun in \cite{paper1}.
In this paper we construct a rigorous RG map for 2D tensor networks whose
domain includes tensors that represent the 2D Ising model at low
temperatures with a magnetic field $h$.
We prove that the RG map has two stable fixed points, corresponding to
the two ground states, and one unstable fixed point which is an example of
a discontinuity fixed point. For the Ising model at low temperatures the
RG map flows to one of the stable fixed points if $h \neq 0$, and to the
discontinuity fixed point if $h=0$. In addition to the nearest neighbor
and magnetic field terms in the Hamiltonian,
we can include small terms that need not be
spin-flip invariant. In this case we prove there is a critical value
$h_c$ of the field (which depends on these additional small interactions and
the temperature) such that the RG map flows to the discontinuity fixed point
if $h=h_c$ and to one of the stable fixed points otherwise.
We use our RG map to give a new proof of previous results on the
first-order transition, namely, that
the free energy is analytic for $h \neq h_c$, and the magnetization is
discontinuous at $h = h_c$.
The construction of our low temperature RG map,
in particular the disentangler, is surprisingly very
similar to the construction of the map in \cite{paper1} for the high
temperature phase.  
We also give a pedagogical discussion of some general rigorous
transformations for infinite dimensional tensor networks and an overview of
the proof of stability of the high temperature fixed point for the RG map in
\cite{paper1}.

\end{abstract}


%

\end{frontmatter}

\tableofcontents

\section{Introduction}

Renormalization group (RG) theory associates to each phase of a lattice
model a fixed point of an RG transformation, which is an attractor for the given phase. 
In a typical case when the model has disordered and ordered phases separated by the critical point, one speaks of the high-temperature (high-T),
low-temperature (low-T), and critical fixed points. In physics, thinking in terms of fixed points and RG flows among them
has become since the 1970s a leading approach to phase transitions. A mathematically rigorous treatment has been achieved near the high-
and low-temperature fixed points. The critical fixed point remains however a
challenge. This is because in a generic lattice model the critical fixed point
is expected to live in an infinite-dimensional coupling space and to have no small parameter.
Constructing such a fixed point rigorously probably requires a
computer-assisted approach. For hierarchical models this was accomplished long
ago, but for physically interesting lattice models with translationally
invariant interactions, such as e.g. the 3D Ising model, this is wide open. To
make progress on this problem is important not only conceptually, but also
practically, as its solution will yield as a by-product the critical exponents
with rigorous error bars.

To achieve this, one needs a nonperturbative RG approach which is both
rigorous and computable. By computable we mean that it should be possible to
evaluate numerically the RG map truncated to a large but finite number of
couplings. The hope is to first identify an approximate fixed point
numerically, and then to prove that an exact fixed point exists nearby.

In this paper we will work with tensor RG {\cite{Levin:2006jai}}, which seems
at present the only RG approach satisfying the requirements of both
computability and rigor. This approach starts by rewriting the lattice model
partition function as a tensor network---a contraction of a periodic arrangement
of tensors (see Fig.~\ref{TNex}). An RG step coarse-grains the network, replacing it
by an equivalent network consisting of a smaller number of tensors. If done
exactly, this step would increase the network bond dimension. In numerical
calculations, one keeps the bond dimensions from growing by truncating the new
tensors. There is significant flexibility in how coarse-graining and subsequent
truncation are performed, and many numerical tensor RG algorithms have been
proposed differing in these details
{\cite{SRG,TEFR,SRG1,HOTRG,Evenbly-Vidal,LoopTNR,GILT,Lyu:2021qlw}}.
When benchmarked on the 2D Ising model, these algorithms give approximate
critical exponents in excellent agreement with the exact values (better than
any other RG method). Can we use one of these algorithms or their modification
to construct an exact critical fixed point, first in 2D and eventually in 3D?

In preparation for this task, one needs to develop a theory of rigorous Tensor
RG maps. Such maps do not involve truncation and preserve the tensor network
value exactly. They naturally operate in the space of infinite-dimensional
tensors. We expect the exact critical fixed point tensors
to be infinite-dimensional. The high-T and low-T fixed point tensors are finite-dimensional, but in an exact treatment the tensor dimension is expected to grow without limit when approaching them (although the weight of all but finitely many tensor components should tend to zero in an appropriate norm).

As a first step in this direction, in {\cite{paper1}} we developed
rigorous 2D tensor RG theory near the high-T fixed point. This fixed point is represented by a
very simple tensor $A$ with a single nonzero component $A_{0000} = 1$.
We considered arbitrary infinite-dimensional perturbations $\delta A$ of this
tensor having small Hilbert-Schmidt norm $\| \delta A \|$, and showed that
after an appropriate tensor RG step the Hilbert-Schmidt norm is reduced: $\|
\delta A' \| = O (\| \delta A \|^{3 / 2})$. In other words, we showed that the
high-T fixed point is stable.

In this work we will continue the study of 2D tensor RG by considering the
vicinity of the low-T fixed point. Low
temperatures are known to be more subtle for rigorous studies than high
temperatures where the cluster expansion can be used. A specific model we have in mind is a small perturbation of the 2D Ising model in a magnetic field. Such models at low temperatures  exhibit a first-order phase transition as a function of the magnetic field. Standard proofs of this are based on the Pirogov-Sinai theory {\cite{Pirogov1975,Pirogov1976,sinai-book}} or a coarse-graining approach formulated in terms of Peierls contours and not in terms of spins by Gawędzki, Koteck\'y and Kupiainen \cite{Gawedzki1987}. In fact, it is known that the block-spin RG
transformations can be ill defined at low temperatures (e.g. {\cite{van_Enter_1993,BRICMONT2001}}) while they are well defined at high temperatures \cite{Griffiths1979,Kashapov1980}.

On the contrary, here we will show that the tensor RG procedure needs no major modifications at low temperatures compared to high temperatures. In the case of $N$ coexisting phases, the low-T fixed point is the direct sum of $N$ high-T fixed points with equal weights: 
\begin{equation}
	A^{(1)} \oplus A^{(2)} \oplus\ldots \oplus A^{(N)}.
\end{equation}
We will work in 2D and will focus on the case of two coexisting phases $N=2$, relevant for the Ising model. In this case, the tensors $A^{(1)}=A^{(+)}$ and $A^{(2)}=A^{(-)}$ describe the states of the Ising model with all spins pointing up and down, respectively.

The Ising model in a magnetic field will be represented by a tensor which is a convex linear combination of $A^{(+)}$ and $A^{(-)}$ up to a correction which is small at low temperatures:
\beq
A= \alpha A^{(+)}+(1-\alpha)A^{(-)} +B,\qquad B=O(e^{-4\beta}),
\eeq
where the zero magnetic field corresponds to $\alpha=1/2$. We will construct a tensor RG map $(\alpha,B)\to (\alpha',B')$ which has the following property:
\begin{equation}
	\label{needed-prop}
	\text{$\textstyle\frac{1}{2} 	A^{(+)} \oplus \frac{1}{2} A^{(2)}$ is a fixed point
	near which the $B$ direction is contracting, while the $\alpha$ direction is expanding}\,.
\end{equation}
By these results, the 2D Ising model at small temperature and at zero magnetic field $h$ converges to the low-T fixed point, while for nonzero $h$ it converges towards $A^{(+)}$ or $A^{(-)}$ depending on the sign of $h$. We will also allow small $\mathbb{Z}_2$-breaking interactions, in which case the critical magnetic field is nonzero. 

In this work we will also use the tensor RG map to study the magnetization near the low-T fixed point and to prove that it is a discontinuous function of the magnetic field, i.e.~that the phase transition is first order. For this reason the low-T fixed point is referred to in the theoretical physics literature as the \emph{discontinuity fixed point}, following Nienhuis and Nauenberg \cite{NienhuisNauenberg}. These authors identified a crucial property that this fixed point should have, to ensure the discontinuity: the eigenvalue of the magnetic field variable should be equal to the volume rescaling factor of the RG map. 

The work by Gawędzki, Koteck\'y and Kupiainen \cite{Gawedzki1987} was the first rigorous implementation of the discontinuity fixed point picture in mathematical physics. Our work provides a new and completely different implementation, based on tensor RG.

The paper is structured as follows. In Sec.~\ref{tensors} we collect the notation and a few simple results about tensors and tensor networks. In Sec.~\ref{models2networks} we recall the tensor network representation of the nearest-neighbor Ising model. More generally, we prove that any lattice spin model with finite range interactions can be transformed to a tensor network, and illustrate this on the Ising model with next-to-nearest-neighbor and plaquette interactions.

Then in Sec.~\ref{HighT} we recall results of our previous paper \cite{paper1} near the high-T fixed point. We reformulate the main result of \cite{paper1} making clear the analytic nature of the RG map. We also give an RG proof that the free energy is analytic near the high-T fixed point.

In Sec.~\ref{LTRG} we move to low temperatures. We formulate the general strategy of recovering the low-temperature phase diagram via a tensor RG map having property \eqref{needed-prop}. Such an RG map is then explicitly constructed in Sec.~\ref{RG_lowT}. 
Sec.~\ref{LTprop} uses this RG map to study the free energy and the magnetization of 
the general tensor network model with two phases at low temperatures (which includes the Ising model and its small perturbations). We prove that the free energy and the magnetization are analytic away from the phase coexistence surface. On this surface, we prove that the free energy is continuous, while the magnetization has a discontinuous jump, i.e.~that the phase transition is first order. The Nienhuis-Nauenberg condition on the magnetic field eigenvalue plays a crucial role in our proof, but the way we derive the discontinuity from this condition is different from theirs, for reasons explained in Sec.~\ref{comparison}. In Section \ref{conclusions} we make final remarks and formulate a few problems for the future.

\section{Tensors}\label{tensors}

In this section we collect the notation about tensors and tensor networks. We also define, following  \cite{paper1}, a few simple equivalence transformations of tensor networks, which will serve later on as building blocks for RG transformations. 

\begin{definition}
	Let $n \geqslant 2$ be an integer and $\mathcal{I}$ a nonempty set which is
	either finite or countable (called an index set). An $n$-leg tensor $A$ with
	legs indexed by $\mathcal{I}_{}$ is a table of real numbers $A_{i_1 i_2
		\ldots i_n}$, $i_k \in \mathcal{I}_{}$, $k = 1 \ldots n$. 
\end{definition}

In other words, such a tensor is a map from $\mathcal{I}^n$ to
$\mathbb{R}$. $n$-leg tensors may also be referred to as ``$n$-tensors'' or
``$n$-index tensors'' or ``tensors with $n$ legs''.

The Hilbert-Schmidt (HS) norm of an $n$-tensor $A^{(n)}$ is defined as
\beq 
\| A^{(n)} \| = \Biggl( \sum_{i_1, \ldots, i_n \in \mathcal{I}} | A_{i_1
	\ldots i_n} |^2 \Biggr)^{1 / 2}\,. 
\eeq
This is the norm we will use most often. A tensor with a finite HS norm will
be called HS. All $n$-tensors with $n \geqslant 3$ considered in this work
will be HS.

For the special case of $n = 2$ (2-tensors) we will also consider a different
norm $\| A \|_{\tmop{op}}$, the operator norm. It is defined as the norm of
the linear operator from $\ell_2 (\mathcal{I})$ to $\ell_2 (\mathcal{I})$
associated with the tensor $A$. We can compute this norm as $\| A
\|_{\tmop{op}} = \sup \left| \sum_{i, j \in \mathcal{I}} A_{i j} v_i w_j
\right|$, the sup being taken over $v, w \in \ell_2 (\mathcal{I})_{}$ having
$\ell_2$ norm 1 and a finite number of nonzero elements. As is well known, $\|
A \|_{\tmop{op}} \leqslant \| A \|_{}$ for 2-tensors.

The above definitions admit obvious generalizations to $n$-tensors whose every
leg is indexed by its own index set $\mathcal{I}_1, \ldots, \mathcal{I}_n$.

For any two $n$-tensors $A$ and $B$, we define their sum as follows:
\begin{itemize}
	\item
	If both
	tensors are indexed by the same index sets $\mathcal{I}_1, \ldots,
	\mathcal{I}_n$, then the sum $A + B$ is the usual element-wise sum. 
	\item
	If the
	two collections of index sets $\mathcal{I}^A_1, \ldots, \mathcal{I}^A_n$ and
	$\mathcal{I}^B_1, \ldots, \mathcal{I}^B_n$ are not the same, we first extend
	each tensor by zeros to a larger tensor defined on $(\mathcal{I}^A_1 \cup
	\mathcal{I}^B_1) \times \ldots \times (\mathcal{I}^A_n \cup \mathcal{I}^B_n)$,
	and then take the usual element-wise sum of such extended tensors.
	\item
	If the index sets for at least one of the $n$ legs do not overlap, say $\mathcal{I}^A_1\cap \mathcal{I}^B_1=\emptyset$, we define the sum as in the previous case, and call it a ``direct sum'' $A \oplus B$. 
\end{itemize}

\subsection{Graphic notation. Contractions}

In graphic notation, an $n$-tensor is drawn as a box with $n$ legs coming out.
E.g. a 4-tensor $A$ and its components $A_{i j k l}$ is denoted by a diagram:
\begin{equation}
\tikzset{every picture/.style={line width=0.75pt}} 
\myraise{\begin{tikzpicture}[x=0.75pt,y=0.75pt,yscale=-1,xscale=1]
	
	\draw   (206,169) -- (224,169) -- (224,187) -- (206,187) -- cycle ;
	\draw    (224,178) -- (239,178) ;
	\draw    (194,178) -- (206,178) ;
	\draw    (215,199) -- (215,187) ;
	\draw    (215,169) -- (215,157) ;
	\draw   (328,169) -- (346,169) -- (346,187) -- (328,187) -- cycle ;
	\draw    (346,178) -- (361,178) ;
	\draw    (316,178) -- (328,178) ;
	\draw    (337,199) -- (337,187) ;
	\draw    (337,169) -- (337,157) ;
	
	\draw (215,184.8) node [anchor=south] [inner sep=0.75pt]    {$A$};
	\draw (337,184.8) node [anchor=south] [inner sep=0.75pt]    {$A$};
	\draw (337.67,156) node [anchor=south] [inner sep=0.75pt]    {$j$};
	\draw (337.67,212) node [anchor=south] [inner sep=0.75pt]    {$l$};
	\draw (180.47,183.1) node [anchor=south] [inner sep=0.75pt]  [xslant=-0.07]  {$A=$};
	\draw (280,187) node [anchor=south] [inner sep=0.75pt]    {$A_{ijkl}=$};
	\draw (367,185) node [anchor=south] [inner sep=0.75pt]    {$i$};
	\draw (310.67,185) node [anchor=south] [inner sep=0.75pt]    {$k$};
\end{tikzpicture}}
	\ .\label{Aind}
\end{equation}
In general, no symmetry will be assumed for $n$-tensors, so one has to pay
attention to the order of the indices. In {\eqref{Aind}}, $A_{i j k l}$ is
denoted by the diagram with $i, j, k, l$ on the right, top, left, bottom legs.

We will also consider contractions of tensors. E.g. for two 4-tensors $A_{}$
and $B$ we may define tensor $C$ contracting the first leg of $A$ with the
third leg of $B$, i.e.
\begin{equation}
	C_{m n j k l p} = \sum_{i \in \mathcal{I}} A_{i j k l} B_{m n i p}\,.
	\label{CAB}
\end{equation}
Contraction is only defined if the two contracted legs have the same index set. In
graphic notation, contraction is denoted by gluing the legs of contracted
tensors. Eq. {\eqref{CAB}} may then be represented as:
\begin{equation}
\tikzset{every picture/.style={line width=0.75pt}} 
\myraise{\begin{tikzpicture}[x=0.75pt,y=0.75pt,yscale=-1,xscale=1]
	
	\draw   (166,129) -- (184,129) -- (184,147) -- (166,147) -- cycle ;
	\draw   (199,129) -- (217,129) -- (217,147) -- (199,147) -- cycle ;
	\draw    (184,138) -- (199,138) ;
	\draw    (217,138) -- (229,138) ;
	\draw    (154,138) -- (166,138) ;
	\draw    (175,159) -- (175,147) ;
	\draw    (208,159) -- (208,147) ;
	\draw    (208,129) -- (208,117) ;
	\draw    (175,129) -- (175,117) ;
	\draw   (88,129) -- (114,129) -- (114,147) -- (88,147) -- cycle ;
	\draw    (114,138) -- (126,138) ;
	\draw    (76,138) -- (88,138) ;
	\draw    (95,159) -- (95,147) ;
	\draw    (107,159) -- (107,147) ;
	\draw    (106,129) -- (106,117) ;
	\draw    (95,129) -- (95,117) ;
	
	\draw (101,144.4) node [anchor=south] [inner sep=0.75pt]    {$C$};
	\draw (175,144.8) node [anchor=south] [inner sep=0.75pt]    {$A$};
	\draw (208,144.4) node [anchor=south] [inner sep=0.75pt]    {$B$};
	\draw (139.5,142.8) node [anchor=south] [inner sep=0.75pt]    {$=$};

\end{tikzpicture}}\ .
\end{equation}
The following proposition follows easily from the Cauchy-Schwarz inequality,
and we omit its proof.

\begin{prop}
	\label{prop-contr}Let $C$ be a tensor
	formed by contracting one leg of a tensor $A$ with one leg of a tensor $B$.
	
	(a) If A and $B$ are HS, then so is $C$, and $\| C \|_{} \leqslant \| A
	\|_{} \| B \|_{}$.
	
	(b) If $A$ is HS, while $B$ is a 2-tensor with a finite operator norm, then $C$
	is HS, and $\| C \|_{} \leqslant \| A \|_{} \| B \|_{\tmop{op}}$.
\end{prop}

The next example will be used frequently. We form a tensor $T$
by contracting 4 copies of an HS tensor $A$, as follows:
\begin{equation}
	\tikzset{every picture/.style={line width=0.75pt}} 
	\myraise{\begin{tikzpicture}[x=0.75pt,y=0.75pt,yscale=-1,xscale=1]
		
		\draw   (186,149) -- (204,149) -- (204,167) -- (186,167) -- cycle ;
		\draw    (204,158) -- (219,158) ;
		\draw    (174,158) -- (186,158) ;
		\draw    (195,179) -- (195,167) ;
		\draw    (195,149) -- (195,137) ;
		\draw   (226,149) -- (244,149) -- (244,167) -- (226,167) -- cycle ;
		\draw    (244,158) -- (259,158) ;
		\draw    (214,158) -- (226,158) ;
		\draw    (235,179) -- (235,167) ;
		\draw    (235,149) -- (235,137) ;
		\draw   (186,189) -- (204,189) -- (204,207) -- (186,207) -- cycle ;
		\draw    (204,198) -- (219,198) ;
		\draw    (174,198) -- (186,198) ;
		\draw    (195,219) -- (195,207) ;
		\draw    (195,189) -- (195,177) ;
		\draw   (226,189) -- (244,189) -- (244,207) -- (226,207) -- cycle ;
		\draw    (244,198) -- (259,198) ;
		\draw    (214,198) -- (226,198) ;
		\draw    (235,219) -- (235,207) ;
		\draw    (235,189) -- (235,177) ;
		\draw   (107,165) -- (133,165) -- (133,189.6) -- (107,189.6) -- cycle ;
		\draw    (133,173) -- (145,173) ;
		\draw    (114,202) -- (114,190) ;
		\draw    (126,202) -- (126,190) ;
		\draw    (125,165) -- (125,153) ;
		\draw    (114,165) -- (114,153) ;
		\draw    (133,183) -- (145,183) ;
		\draw    (95,173) -- (107,173) ;
		\draw    (95,183) -- (107,183) ;
		
		\draw (195,164.8) node [anchor=south] [inner sep=0.75pt]    {$A$};
		\draw (235,164.8) node [anchor=south] [inner sep=0.75pt]    {$A$};
		\draw (195,204.8) node [anchor=south] [inner sep=0.75pt]    {$A$};
		\draw (235,204.8) node [anchor=south] [inner sep=0.75pt]    {$A$};
		\draw (121,183.4) node [anchor=south] [inner sep=0.75pt]    {$T$};
		\draw (159.5,182.8) node [anchor=south] [inner sep=0.75pt]    {$=$};

	\end{tikzpicture}}\ \ .
	\label{Tex}
\end{equation}
The tensor $T$ is well defined: it's easy to see that each component is given by an absolutely
convergent series. Moreover, $T$ is HS and $\| T \| \leqslant \| A \|^4$. This is easy to prove directly,
and is also a consequence of the general Prop. \ref{contr-gen} below.

\subsection{Leg grouping and reindexing}
\label{LGR}
The tensor $T$ in \eqref{Tex} has 8 legs indexed by $\mathcal{I}$
(same index set as $A$). These 8 legs come naturally grouped in 4 pairs (right, up, left,
down). We may consider each of these 4 pairs of legs as a single leg with
an index set $\mathcal{I} \times \mathcal{I}$. Viewed this way, $T$ becomes a
4-tensor indexed by $\mathcal{I} \times \mathcal{I}$. This is an example of
leg grouping, which reshapes the tensor but does not change its components. We
can also group more than two legs.\footnote{In the {\tt python} package {\tt numpy}, often used for numerical tensor manipulations, leg grouping can be performed by the function {\tt reshape()}.}

Another simple operation is reindexing. Let $\mathcal{I}_1$ and
$\mathcal{I}_2$ be two index sets of the same cardinality, and let $\rho$ be a
one-to-one map from $\mathcal{I}_1$ onto $\mathcal{I}_2$ (reindexing map). If
$A$ is a tensor with a leg indexed by $\mathcal{I}_1$, we can use $\rho$ to
transform $A$ to a tensor $A'$ where the same leg is indexed by
$\mathcal{I}_2$. 

Leg grouping and reindexing just reshuffle tensor components but do not change their numerical values. In particular, these operations preserve the norm.

Here is an equivalent view of reindexing. Let $J$ be a 2-tensor with
the only nonzero components $J_{i i'} = 1$ if $i' = \rho (i)$. We call $J$
a reindexing tensor; it has operator norm 1. The reindexed tensor $A'$ is then
obtained by contracting $A$ with $J$. We write this graphically as:
\begin{equation}
	\tikzset{every picture/.style={line width=0.75pt}} 
	\myraise{\begin{tikzpicture}[x=0.75pt,y=0.75pt,yscale=-1,xscale=1]
		
		\draw   (99,59) -- (117,59) -- (117,77) -- (99,77) -- cycle ;
		\draw    (117,68) -- (132,68) ;
		\draw    (87,68) -- (99,68) ;
		\draw    (108,89) -- (108,77) ;
		\draw    (108,59) -- (108,47) ;
		\draw   (189,59) -- (207,59) -- (207,77) -- (189,77) -- cycle ;
		\draw    (207,68) -- (218,68.2) ;
		\draw    (177,68) -- (189,68) ;
		\draw    (198,89) -- (198,77) ;
		\draw    (198,59) -- (198,47) ;
		\draw   (218,59.2) -- (225,59.2) .. controls (228.87,59.2) and (232,63.01) .. (232,67.7) .. controls (232,72.39) and (228.87,76.2) .. (225,76.2) -- (218,76.2) -- cycle ;
		\draw    (232,68) -- (243,68.2) ;
		
		\draw (108,74.8) node [anchor=south] [inner sep=0.75pt]    {$A'$};
		\draw (198,73.8) node [anchor=south] [inner sep=0.75pt]    {$A$};
		\draw (152,64.2) node [anchor=north west][inner sep=0.75pt]    {$=$};
		\draw (219,62.4) node [anchor=north west][inner sep=0.75pt]    {$J$};
		\draw (134,61.2) node [anchor=north west][inner sep=0.75pt]    {$i'$};
		\draw (247,61.2) node [anchor=north west][inner sep=0.75pt]    {$i'$};

	\end{tikzpicture}}\ \ .
	\label{Jex}
\end{equation}

Below we will often apply leg grouping followed by reindexing, as follows.
Consider two legs of a tensor with index set $\mathcal{I}$. We group them,
obtaining a leg with an index set $\mathcal{I} \times \mathcal{I}$. Suppose that either $| \mathcal{I} | = 1$ or $| \mathcal{I} | =
\infty$. Then
$\mathcal{I} \times \mathcal{I}$ has the same cardinality as $\mathcal{I}$.
We can then apply reindexing as above with $\mathcal{I}_1
=\mathcal{I} \times \mathcal{I}$ and $\mathcal{I}_2 =\mathcal{I}$. After
reindexing, the leg is indexed with $\mathcal{I}$.

Let us see how this works for the tensor $T$ defined by {\eqref{Tex}}. As
explained, after leg grouping we can view it as a 4-tensor with legs indexed
by $\mathcal{I} \times \mathcal{I}$. We then apply reindexing on each of this
legs, and obtain a 4-tensor indexed by $\mathcal{I}$, the same index set we
started with.

Our main case of interest will be $| \mathcal{I} | = \infty$. In this case the reindexing map from
$\mathcal{I} \times \mathcal{I}$ to $\mathcal{I}$ is vastly non-unique.
Without loss of generality, we can take $\mathcal{I}=\mathbb{N}$. Choosing the
reindexing map $\rho$ then amounts to enumerating $\mathbb{N} \times
\mathbb{N}$ in some particular order. In our constructions below, we will fix
the first few elements of the enumeration sequence, and the rest of it will be
left arbitrary.

\subsection{Tensor networks}

Consider a finite periodic square lattice of size $L_x \times L_y$. Suppose we
put a 4-tensor $A$ at every vertex $(n, m)$ of the lattice, contracting its
legs with the legs of tensors at neighboring vertices, and taking into account
periodicity, as shown in this figure:
\begin{equation}
	\tikzset{every picture/.style={line width=0.75pt}} 
	\myraise{\begin{tikzpicture}[x=0.75pt,y=0.75pt,yscale=-1,xscale=1]
		
		\draw   (86,74.27) -- (104,74.27) -- (104,92.27) -- (86,92.27) -- cycle ;
		\draw    (104,83.27) -- (119,83.27) ;
		\draw    (74,83.27) -- (86,83.27) ;
		\draw    (95,104.27) -- (95,92.27) ;
		\draw    (95,74.27) -- (95,62.27) ;
		\draw   (126,74.27) -- (144,74.27) -- (144,92.27) -- (126,92.27) -- cycle ;
		\draw    (144,83.27) -- (159,83.27) ;
		\draw    (114,83.27) -- (126,83.27) ;
		\draw    (135,104.27) -- (135,92.27) ;
		\draw    (135,74.27) -- (135,62.27) ;
		\draw   (86,114.27) -- (104,114.27) -- (104,132.27) -- (86,132.27) -- cycle ;
		\draw    (104,123.27) -- (119,123.27) ;
		\draw    (74,123.27) -- (86,123.27) ;
		\draw    (95,144.27) -- (95,132.27) ;
		\draw    (95,114.27) -- (95,102.27) ;
		\draw   (126,114.27) -- (144,114.27) -- (144,132.27) -- (126,132.27) -- cycle ;
		\draw    (144,123.27) -- (159,123.27) ;
		\draw    (114,123.27) -- (126,123.27) ;
		\draw    (135,144.27) -- (135,132.27) ;
		\draw    (135,114.27) -- (135,102.27) ;
		\draw   (196,74.27) -- (214,74.27) -- (214,92.27) -- (196,92.27) -- cycle ;
		\draw    (214,83.27) -- (229,83.27) ;
		\draw    (184,83.27) -- (196,83.27) ;
		\draw    (205,104.27) -- (205,92.27) ;
		\draw    (205,74.27) -- (205,62.27) ;
		\draw   (196,114.27) -- (214,114.27) -- (214,132.27) -- (196,132.27) -- cycle ;
		\draw    (214,123.27) -- (229,123.27) ;
		\draw    (184,123.27) -- (196,123.27) ;
		\draw    (205,144.27) -- (205,132.27) ;
		\draw    (205,114.27) -- (205,102.27) ;
		\draw   (86,181.27) -- (104,181.27) -- (104,199.27) -- (86,199.27) -- cycle ;
		\draw    (104,190.27) -- (119,190.27) ;
		\draw    (74,190.27) -- (86,190.27) ;
		\draw    (95,211.27) -- (95,199.27) ;
		\draw    (95,181.27) -- (95,169.27) ;
		\draw   (126,181.27) -- (144,181.27) -- (144,199.27) -- (126,199.27) -- cycle ;
		\draw    (144,190.27) -- (159,190.27) ;
		\draw    (114,190.27) -- (126,190.27) ;
		\draw    (135,211.27) -- (135,199.27) ;
		\draw    (135,181.27) -- (135,169.27) ;
		\draw   (196,181.27) -- (214,181.27) -- (214,199.27) -- (196,199.27) -- cycle ;
		\draw    (214,190.27) -- (229,190.27) ;
		\draw    (184,190.27) -- (196,190.27) ;
		\draw    (205,211.27) -- (205,199.27) ;
		\draw    (205,181.27) -- (205,169.27) ;
		\draw    (205,211.27) .. controls (204.44,220.07) and (220.44,222.05) .. (219.73,210.25) ;
		\draw    (205,62.27) .. controls (205.22,50.67) and (219.22,51.65) .. (219.73,61.25) ;
		\draw [color={rgb, 255:red, 0; green, 0; blue, 0 }  ,draw opacity=0.15 ]   (219.54,213.74) -- (219.37,58.74) ;
		\draw    (74,190.27) .. controls (65.2,189.7) and (63.2,205.7) .. (75,205) ;
		\draw    (229,190.27) .. controls (240.6,190.5) and (239.6,204.5) .. (230,205) ;
		\draw [color={rgb, 255:red, 0; green, 0; blue, 0 }  ,draw opacity=0.16 ]   (75,205) -- (230,205) ;
		\draw    (74,83.27) .. controls (65.2,82.7) and (63.2,98.7) .. (75,98) ;
		\draw    (229,83.27) .. controls (240.6,83.5) and (239.6,97.5) .. (230,98) ;
		\draw [color={rgb, 255:red, 0; green, 0; blue, 0 }  ,draw opacity=0.15 ]   (75,98) -- (230,98) ;
		\draw    (135,211.27) .. controls (134.44,220.07) and (150.44,222.05) .. (149.73,210.25) ;
		\draw    (135,62.27) .. controls (135.22,50.67) and (149.22,51.65) .. (149.73,61.25) ;
		\draw [color={rgb, 255:red, 0; green, 0; blue, 0 }  ,draw opacity=0.15 ]   (149.54,213.74) -- (149.37,58.74) ;
		\draw    (95,211.27) .. controls (94.44,220.07) and (110.44,222.05) .. (109.73,210.25) ;
		\draw    (95,62.27) .. controls (95.22,50.67) and (109.22,51.65) .. (109.73,61.25) ;
		\draw [color={rgb, 255:red, 0; green, 0; blue, 0 }  ,draw opacity=0.15 ]   (109.54,213.74) -- (109.37,58.74) ;
		\draw    (74,123.27) .. controls (65.2,122.7) and (63.2,138.7) .. (75,138) ;
		\draw    (229,123.27) .. controls (240.6,123.5) and (239.6,137.5) .. (230,138) ;
		\draw [color={rgb, 255:red, 0; green, 0; blue, 0 }  ,draw opacity=0.16 ]   (75,138) -- (230,138) ;
		\draw    (95,234) -- (205,234) ;
		\draw [shift={(207,234)}, rotate = 180] [color={rgb, 255:red, 0; green, 0; blue, 0 }  ][line width=0.75]    (10.93,-3.29) .. controls (6.95,-1.4) and (3.31,-0.3) .. (0,0) .. controls (3.31,0.3) and (6.95,1.4) .. (10.93,3.29)   ;
		\draw [shift={(93,234)}, rotate = 0] [color={rgb, 255:red, 0; green, 0; blue, 0 }  ][line width=0.75]    (10.93,-3.29) .. controls (6.95,-1.4) and (3.31,-0.3) .. (0,0) .. controls (3.31,0.3) and (6.95,1.4) .. (10.93,3.29)   ;
		\draw    (51,190) -- (51,86) ;
		\draw [shift={(51,84)}, rotate = 90] [color={rgb, 255:red, 0; green, 0; blue, 0 }  ][line width=0.75]    (10.93,-3.29) .. controls (6.95,-1.4) and (3.31,-0.3) .. (0,0) .. controls (3.31,0.3) and (6.95,1.4) .. (10.93,3.29)   ;
		\draw [shift={(51,192)}, rotate = 270] [color={rgb, 255:red, 0; green, 0; blue, 0 }  ][line width=0.75]    (10.93,-3.29) .. controls (6.95,-1.4) and (3.31,-0.3) .. (0,0) .. controls (3.31,0.3) and (6.95,1.4) .. (10.93,3.29)   ;
		
		\draw (95,90.07) node [anchor=south] [inner sep=0.75pt]    {$A$};
		\draw (135,90.07) node [anchor=south] [inner sep=0.75pt]    {$A$};
		\draw (95,130.07) node [anchor=south] [inner sep=0.75pt]    {$A$};
		\draw (135,130.07) node [anchor=south] [inner sep=0.75pt]    {$A$};
		\draw (205,90.07) node [anchor=south] [inner sep=0.75pt]    {$A$};
		\draw (205,130.07) node [anchor=south] [inner sep=0.75pt]    {$A$};
		\draw (95,197.07) node [anchor=south] [inner sep=0.75pt]    {$A$};
		\draw (135,197.07) node [anchor=south] [inner sep=0.75pt]    {$A$};
		\draw (205,197.07) node [anchor=south] [inner sep=0.75pt]    {$A$};
		\draw (163,77.86) node [anchor=north west][inner sep=0.75pt]    {$\dotsc $};
		\draw (164,117.86) node [anchor=north west][inner sep=0.75pt]    {$\dotsc $};
		\draw (163,184.86) node [anchor=north west][inner sep=0.75pt]    {$\dotsc $};
		\draw (99.75,147.63) node [anchor=north west][inner sep=0.75pt]  [rotate=-89.67]  {$\dotsc $};
		\draw (139.75,147.63) node [anchor=north west][inner sep=0.75pt]  [rotate=-89.67]  {$\dotsc $};
		\draw (209.75,148.63) node [anchor=north west][inner sep=0.75pt]  [rotate=-89.67]  {$\dotsc $};
		\draw (148,238.2) node [anchor=north west][inner sep=0.75pt]    {$L_{x}$};
		\draw (32,130.2) node [anchor=north west][inner sep=0.75pt]    {$L_{y}$};

	\end{tikzpicture}}\ \ .
	\label{TNex}
\end{equation}
Such an arrangement of tensors is an example of a tensor network (TN). Performing
all contractions in a network, we get a number called its "value" or "partition function". This terminology is natural because partition functions of many lattice models can be represented by tensor networks, see Section \ref{models2networks}.

We will sometimes abuse the language and will use the expression "tensor network" to refer to both the tensor network and its partition function. We denote the partition function of the TN \eqref{TNex} by $Z(A,L_x,L_y)$. 

\begin{prop}
	Suppose that $A$ is HS, and that $L_x, L_y \geqslant 2$. Then the partition function of tensor network {\eqref{TNex}} is finite and satisfies the bound $|Z(A,L_x,L_y) | \leqslant
	\| A \|^{L_x L_y}$.
\end{prop}

This follows from a more general result. 

\begin{prop}
	\label{contr-gen}Let $A^1_{}, \ldots, A^N$ be tensors which are HS. Let $C$
	be a tensor formed by contracting some collection of pairs of their legs.
	Assume that in each pair the two contracted legs do not belong to the same
	tensor (``no self-contraction''). Then $C$ is HS and $\| C \|_{} \leqslant
	\prod_{i = 1}^N \| A^i_{} \|_{}$.
\end{prop}

\begin{remark}
	The no self-contraction condition is important, because finite HS norm does
	not guarantee finiteness of the trace. 
\end{remark}

\begin{proof}
	We use induction on $N$. In the base case $N = 1$ the proposition is trivially
	true, since $C = A^1_{}$ as no self-contractions are allowed.
	
	Now assume the proposition is true for $N$ tensors. Let $C^{N + 1}$ be some
	contraction of $A^1_{}, \ldots, A^{N + 1}$ satisfying the no self-contraction
	condition. Let $C^N$ be the contraction of $A^1_{}, \ldots, A^N$ formed using
	all the contractions in $C^{N + 1}$ that do not involve $A^{N + 1}$. By the
	inductive hypothesis, $C^N$ is HS and $\| C^N \|_{} \leqslant \prod_{i = 1}^N
	\| A^i_{} \|$.
	
	The tensor $C^{N + 1}$ is formed by contracting some of the legs of $C^N$
	with some of the legs of $A^{N + 1}$. We group the legs in $C^N$ being
	contracted and group the legs in $A^{N + 1}$ being contracted. Then, by Prop.
	\ref{prop-contr}(a) we have $\| C^{N + 1} \| \leqslant \| C^N \| \| A^{N + 1}
	\|$ and, by the inductive hypothesis,
	\begin{equation}
		\| C^{N + 1} \| \leqslant \prod_{i = 1}^{N + 1} \| A^{N + 1}_{} \|,
		\label{ind-step}
	\end{equation}
	proving the induction step.
	
	It may happen that there are actually no legs to contract as $A^{N + 1}$ is
	disconnected from $C^N$ (the proposition does not assume that the network is
	connected). In this case $C^{N + 1}$ is the element-wise product of $C^N$ and
	$A^{N + 1}$. Such a $C^{N + 1}$ is trivially HS and satisfies $\| C^{N + 1} \|
	= \| C^N \| \| A^{N + 1} \|$. So Eq. {\eqref{ind-step}} follows as well.
\end{proof}

\subsection{Transformations of tensor networks}

It is possible that two networks are built of different tensors but have the
same value. We call such networks equivalent. 
Here we will describe several equivalence transformations of tensor networks: general and simple RG, gauge, and disentangler transformations.

\subsubsection{Tensor RG transformations}
\label{tensorRG}

Let $b$ be a positive integer. A tensor RG transformation with a lattice rescaling factor $b$ is a map of 4-tensors $A\mapsto A'$ which has the following basic property: for any $L_x,L_y$ divisible by $b$, the partition function of $L_x\times L_y$ tensor network made of $A$ equals the partition function of $(L_x/b)\times (L_y/b)$ tensor network made of $A'$:
\beq
\label{genTRG}
Z(A,L_x,L_y)=Z(A',L_x/b,L_y/b)\,.
\eeq
Note that $A'$ may in general have an index set different from $A$. We are mostly interested in $b\ge 2$ so that the number of tensors is reduced, although gauge transformations (see Section \ref{sec:gauge} below) have $b=1$.

Constructing such tensor RG transformations and using them to understand the phase diagram of lattice models will be the main point of our work. We will construct them by combining judiciously three basic equivalence transformations discussed below.

\subsubsection{Simple RG transformations}
\label{simpleRG}
A simple RG transformation is an example of a tensor RG transformations. It divides tensors $A$ comprising the network into small
groups, and defines new tensors $A'$ by contracting tensors within each group. The
resulting network has fewer tensors than the original one. It has the same
value by construction.

As an example, consider network {\eqref{TNex}} and assume that $L_x$ and
$L_y$ are even. We group tensors $A$ in $2 \times 2$ groups. This gives an
equivalent (same value) TN:
\begin{equation}
	\tikzset{every picture/.style={line width=0.75pt}} 
	\myraise{\begin{tikzpicture}[x=0.75pt,y=0.75pt,yscale=-1,xscale=1]
		
		\draw   (127,155.16) -- (153,155.16) -- (153,179.76) -- (127,179.76) -- cycle ;
		\draw    (153,163.16) -- (165,163.16) ;
		\draw    (135,192.16) -- (135,180.16) ;
		\draw    (146,192.16) -- (146,180.16) ;
		\draw    (145,155.16) -- (145,143.16) ;
		\draw    (134,155.16) -- (134,143.16) ;
		\draw    (153,173.16) -- (165,173.16) ;
		\draw    (115,163.16) -- (127,163.16) ;
		\draw    (115,173.16) -- (127,173.16) ;
		\draw   (172,155.16) -- (198,155.16) -- (198,179.76) -- (172,179.76) -- cycle ;
		\draw    (198,163.16) -- (210,163.16) ;
		\draw    (180,192.16) -- (180,180.16) ;
		\draw    (191,192.16) -- (191,180.16) ;
		\draw    (190,155.16) -- (190,143.16) ;
		\draw    (179,155.16) -- (179,143.16) ;
		\draw    (198,173.16) -- (210,173.16) ;
		\draw    (160,163.16) -- (172,163.16) ;
		\draw    (160,173.16) -- (172,173.16) ;
		\draw   (128,204.16) -- (154,204.16) -- (154,228.76) -- (128,228.76) -- cycle ;
		\draw    (154,212.16) -- (166,212.16) ;
		\draw    (136,241.16) -- (136,229.16) ;
		\draw    (147,241.16) -- (147,229.16) ;
		\draw    (146,204.16) -- (146,192.16) ;
		\draw    (135,204.16) -- (135,192.16) ;
		\draw    (154,222.16) -- (166,222.16) ;
		\draw    (116,212.16) -- (128,212.16) ;
		\draw    (116,222.16) -- (128,222.16) ;
		\draw   (173,204.16) -- (199,204.16) -- (199,228.76) -- (173,228.76) -- cycle ;
		\draw    (199,212.16) -- (211,212.16) ;
		\draw    (181,241.16) -- (181,229.16) ;
		\draw    (192,241.16) -- (192,229.16) ;
		\draw    (191,204.16) -- (191,192.16) ;
		\draw    (180,204.16) -- (180,192.16) ;
		\draw    (199,222.16) -- (211,222.16) ;
		\draw    (161,212.16) -- (173,212.16) ;
		\draw    (161,222.16) -- (173,222.16) ;
		\draw    (127,252.26) -- (200,252.26) ;
		\draw [shift={(202,252.26)}, rotate = 180] [color={rgb, 255:red, 0; green, 0; blue, 0 }  ][line width=0.75]    (10.93,-3.29) .. controls (6.95,-1.4) and (3.31,-0.3) .. (0,0) .. controls (3.31,0.3) and (6.95,1.4) .. (10.93,3.29)   ;
		\draw [shift={(125,252.26)}, rotate = 0] [color={rgb, 255:red, 0; green, 0; blue, 0 }  ][line width=0.75]    (10.93,-3.29) .. controls (6.95,-1.4) and (3.31,-0.3) .. (0,0) .. controls (3.31,0.3) and (6.95,1.4) .. (10.93,3.29)   ;
		\draw    (96.97,225.26) -- (96.03,155.26) ;
		\draw [shift={(96,153.26)}, rotate = 89.23] [color={rgb, 255:red, 0; green, 0; blue, 0 }  ][line width=0.75]    (10.93,-3.29) .. controls (6.95,-1.4) and (3.31,-0.3) .. (0,0) .. controls (3.31,0.3) and (6.95,1.4) .. (10.93,3.29)   ;
		\draw [shift={(97,227.26)}, rotate = 269.23] [color={rgb, 255:red, 0; green, 0; blue, 0 }  ][line width=0.75]    (10.93,-3.29) .. controls (6.95,-1.4) and (3.31,-0.3) .. (0,0) .. controls (3.31,0.3) and (6.95,1.4) .. (10.93,3.29)   ;
		
		\draw (141,173.56) node [anchor=south] [inner sep=0.75pt]    {$T$};
		\draw (186,173.56) node [anchor=south] [inner sep=0.75pt]    {$T$};
		\draw (142,222.56) node [anchor=south] [inner sep=0.75pt]    {$T$};
		\draw (187,222.56) node [anchor=south] [inner sep=0.75pt]    {$T$};
		\draw (156,132.96) node [anchor=north west][inner sep=0.75pt]   [align=left] {...};
		\draw (158,240.96) node [anchor=north west][inner sep=0.75pt]   [align=left] {...};
		\draw (218.71,185.65) node [anchor=north west][inner sep=0.75pt]  [rotate=-89.04] [align=left] {...};
		\draw (111.71,186.65) node [anchor=north west][inner sep=0.75pt]  [rotate=-89.04] [align=left] {...};
		\draw (154,257.46) node [anchor=north west][inner sep=0.75pt]    {$L_{x} /2$};
		\draw (61,185.46) node [anchor=north west][inner sep=0.75pt]    {$L_{y} /2$};
	\end{tikzpicture}}\ \ .
	\label{Trew}
\end{equation}
where $T$ is the contraction of four $A$ tensors defined in {\eqref{Tex}}. The map $A\mapsto T$ is a simple RG transformation. It has lattice rescaling factor $b=2$.

\subsubsection{Gauge transformations}
\label{sec:gauge}
Gauge transformations map a TN to an equivalent one having the same number of tensors. This can be thought as a tensor RG transformation having lattice rescaling factor $b=1$.

Let $G$ be a bounded, invertible linear map on $\ell_2 (\mathcal{I})$ whose inverse
$G^{- 1}$ is also bounded.
We denote the corresponding 2-tensors by
the same symbols $G$, $G^{- 1}$. In graphic notation:
\begin{equation}
	\myinclude{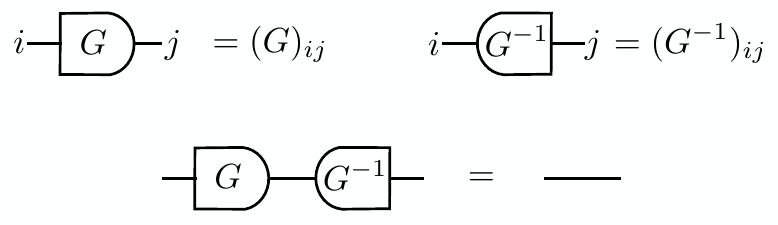}
	\label{GGinv}
\end{equation}

Take two bounded maps $G_x$ and $G_y$ as above and consider network
{\eqref{TNex}} built out of tensors $A$. Gauge transformation of this network
consists of two operations:
\begin{itemize}
	\item insert the product $G_x G_x^{- 1} = I$ on every horizontal link, and
	$G_y G_y^{- 1} = I$ on every vertical link of the network;
	
	\item regroup the tensors, passing to the equivalent network built out of
	tensor $A' $ (which is HS by Prop. \ref{prop-contr}(b)):
	\begin{equation}
\myinclude{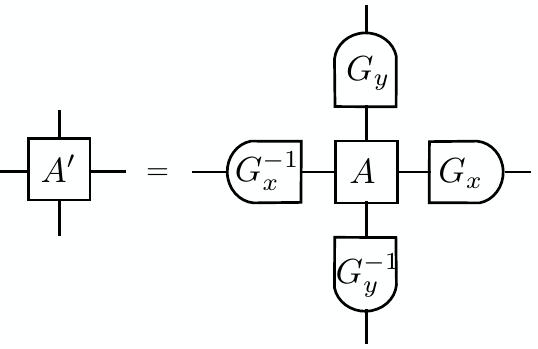}
		\label{gauge-ex}
	\end{equation}
\end{itemize}

\subsubsection{Disentangler transformations}
\label{sec:disentanglers}

Like gauge transformations, disentangler transformations will preserve the number of tensors in the network.

Let us first define ``disentanglers''. Disentanglers are like gauge
transformation tensors but acting on two legs at the same time. Namely,
a disentangler $R$ is a bounded, invertible linear map on $\ell_2 (\mathcal{I} \times
\mathcal{I})$ whose inverse $R^{- 1}$ is also bounded:
\begin{equation}
\myinclude{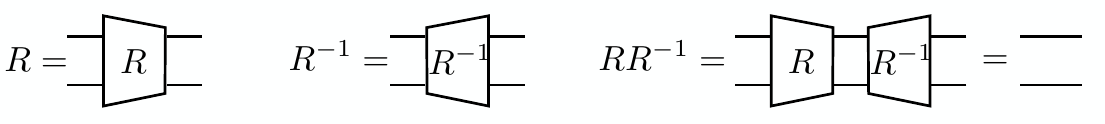}
	\label{RRinv}
\end{equation}
Let us now define disentangler transformations. Let $X$ and $Y$ be two 4-tensors indexed by the same index set. They may be different or identical. Consider a network made of these two tensors, putting them in alternation on the odd and even rows:
\begin{equation}
\myinclude{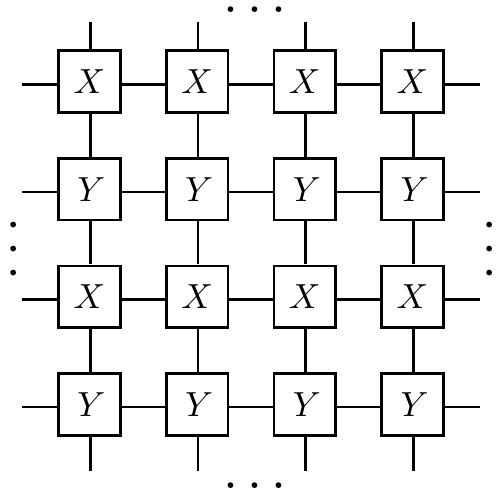} \label{XYtn}
\end{equation}
A disentangler transformation of this network consists of two steps:
\begin{enumerate}
	\item[\bf D1.]\label{D1} Insert $R R^{- 1} = I$, where $R$ is a disentangler, on the pairs of
	horizontal links: 
	\begin{equation}
\myinclude{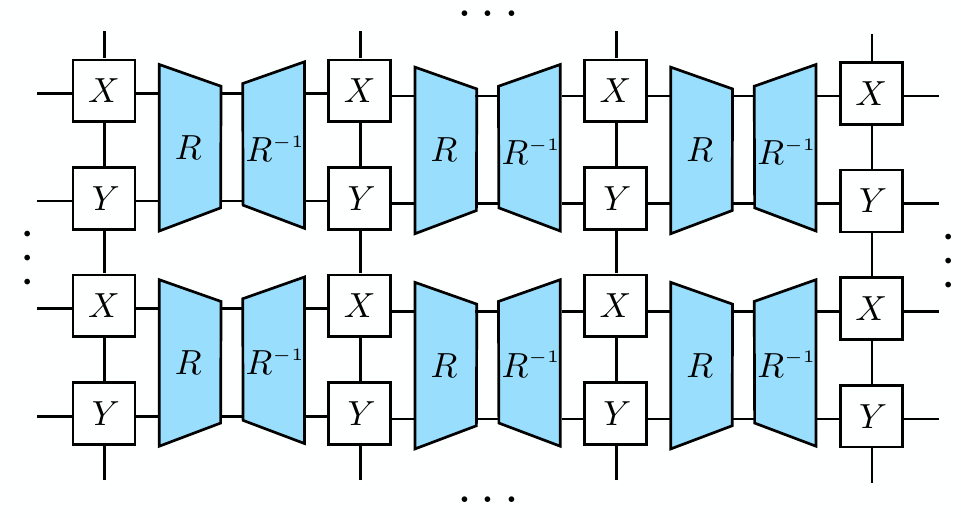}
		\label{XYdis}
	\end{equation}
	The resulting network can be equivalently obtained as a periodic contraction of the following single tensor:
	\begin{equation}
\myinclude[scale=1.3]{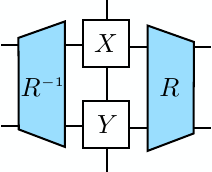}
		\label{XYdis0}
	\end{equation}
	\item[\bf D2.]\label{D2} Represent tensor \eqref{XYdis0} as a contraction of two new tensors
	$X'$ and $Y'$:
	\begin{equation}
\myinclude[scale=1.3]{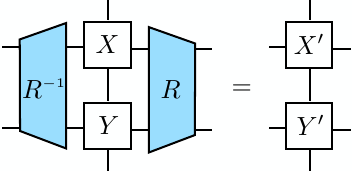}
		\label{XYdis1}
	\end{equation}
\end{enumerate}
Then, the TN built of $X'$ and $Y'$ in place of $X$ and $Y$ is
equivalent to the original network.

This completes the definition of disentangler transformation acting in the horizontal direction. Disentangler transformations acting in the vertical direction could be defined analogously, but we will not need them in this work. 

Note that for disentanglers of tensor product form $R = G_x \otimes G_x$, disentangler transformations reduce to gauge transformations on horizontal legs. 

\begin{remark} \label{disent-origine}The name ``disentangler'' comes from the fact that these transformation can
	reduce ``entanglement'' in the tensor network. We will
	not need to define this rigorously, but we will give an intuitive explanation. Roughly, tensor entanglement measures how many indices effectively contribute to tensor contractions.
	Entanglement is minimized if the contracted index space consists of a single
	element. This is analogous to quantum mechanics, where two subsystems are
	considered not entangled if the wavefunction of the total system is given by
	the product of subsystem wavefunctions. In this work, to reduce entanglement, we will aim to arrange so that the contraction of tensors $X'$ and $Y'$ mostly comes from one or two special values of the contracted index, while the other indices contribute as little as possible.
\end{remark}

Let us next discuss how disentangler transformations are performed in practice. One issue is how to find tensors $X'$ and $Y'$ satisfying {\eqref{XYdis1}}. In the finite-dimensional case $X'$ and $Y'$ always exist, and can be found e.g.~via singular value decomposition (SVD). For infinite-dimensional tensors and for general disentanglers $R$, this may not always be possible. 

In this work we will use infinite-dimensional tensors, and we will not rely on SVD. Instead, we will use disentanglers of a special form: we will assume that
$R$ and $R^{- 1}$ can be represented as the identity plus a contraction of two HS
3-tensors:
\begin{equation}
	\myinclude{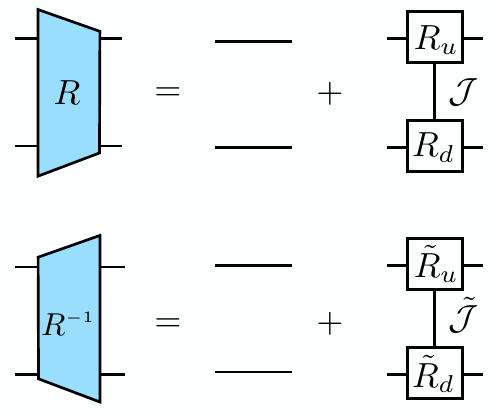},
	\label{Rfact}
\end{equation}
where $\mathcal{J}, \tilde{\mathcal{J}}$ denote the index sets of the contracted
legs in these equations. 

For such disentanglers, $X'$ and $Y'$ satisfying Eq. {\eqref{XYdis1}} can be found explicitly. Namely, we take:
\begin{equation}
\myinclude{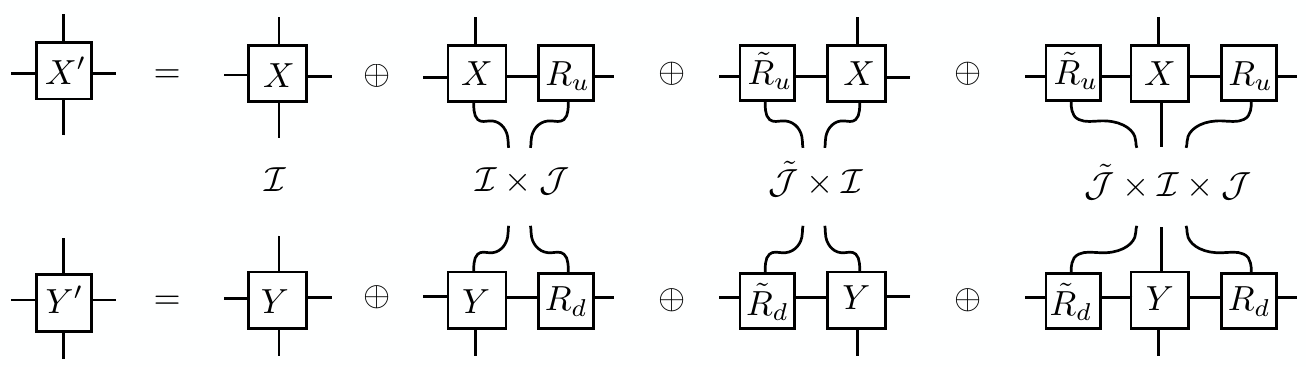}\ \ ,
	\label{XprimYprimSolve}
\end{equation}
where we indicated in the middle the index sets for each of the four terms in
the r.h.s. These four index sets are disjoint, and the total index set for the
common leg of $X'$ and $Y'$ is their union.
The disjointness of these index sets is crucial. It implies that when we contract $X'$ and $Y'$ we only
obtain the four ``diagonal'' terms. 

In turn, disentanglers satisfying {\eqref{Rfact}} will be obtained through the following exponentiation procedure. Consider a 4-tensor
$\rho$ written as a contraction of two HS 3-tensors (contracted leg index space
$\mathcal{K}$):
\begin{equation}
	\myinclude{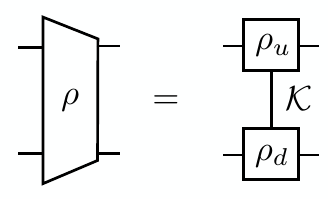}, 
	\label{r}
\end{equation}

\begin{lemma}\label{Rfromr}
	For $\rho$ as in {\eqref{r}}, viewed as a linear map on $\ell^2 (\mathcal{I}\times \mathcal{I})$, the map
	\begin{equation}
		R = \exp (\rho) 
	\end{equation}
	is a disentangler satisfying {\eqref{Rfact}}. More precisely we have 
	\begin{equation}
\myinclude{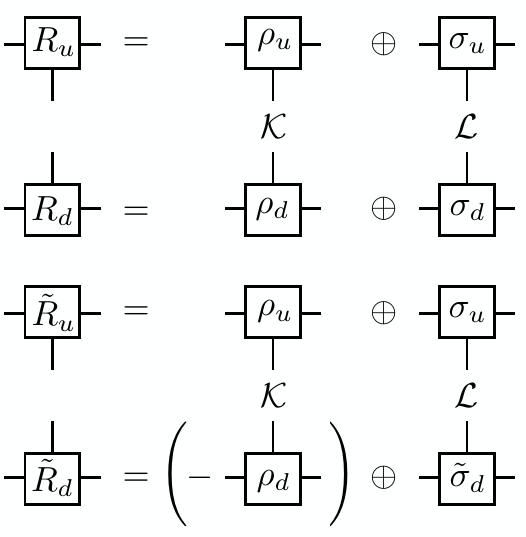}\ \ ,
		\label{LemmaR}
	\end{equation}
	where $\mathcal{K}$ and $\mathcal{L}$ are disjoint index spaces, and
	$\sigma_u, \sigma_d, \tilde{\sigma}_d$ are HS tensors with norms
	\begin{equation}
		\| \sigma_u \|  =O (\| \rho_u \|^2), \qquad \| \sigma_d \|, \| \tilde{\sigma}_d \|
		=O (\| \rho_d \|^2) \label{SuSdnorm},
	\end{equation}
	 which analytically\footnote{See Appendix \ref{abstract} for a brief reminder about analytic functions in Banach spaces.} depend on $\rho_u,\rho_d$.
	  
\end{lemma}	  

\begin{proof}
	(see Lemma 2.1 in \cite{paper1}) For $R$ we have
	\begin{equation}
	\myinclude{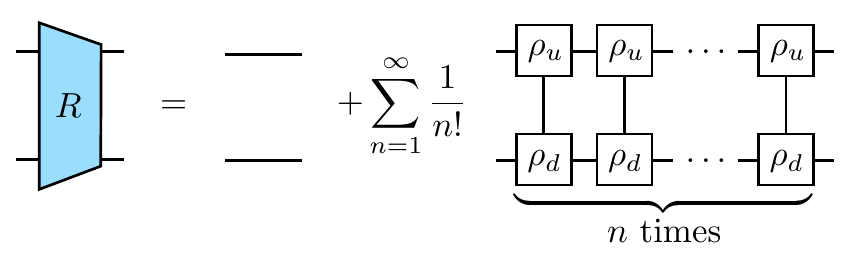},
		\label{Rexp}
	\end{equation}
	from where we see that {\eqref{Rfact}} can be satisfied for $R_u, R_d$
	as in {\eqref{LemmaR}} with $\mathcal{L}= \cup_{n = 2}^{\infty}
	\mathcal{K}^n$.\,
	\begin{equation}
\myinclude{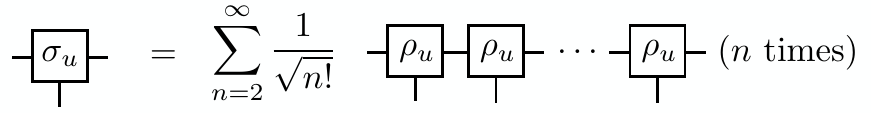}.
		\label{sigmau}
	\end{equation}
	and an analogous series for $\sigma_d$. Estimating the norm of individual summands in these series by Prop.~\ref{prop-contr}, we get that $\sigma_u, \sigma_d$ are HS with norms bounds $\| \sigma_u \| \leqslant f (\| \rho_u \|)$, $\| \sigma_d \|, \| \tilde{\sigma}_d \|
	\leqslant f (\| \rho_d \|)$ where $f (x) = \sum_{n = 2}^{\infty} \frac{1}{\sqrt{n!}} x^n=O(x^2)$.

The case of $R^{- 1}$ is considered similarly using
	$R^{- 1} = \exp (- \rho)$.
\end{proof}

\begin{remark}
  In the representation of $R^{-1}$ as the contraction of $\tilde{R}_u$ and $\tilde{R}_d$ we have made a choice of where to put the minus
  sign in $\exp(- \rho)$. We have arbitrarily chosen to put it in the lower diagrams. Thus this minus sign appears in front of the
  $\rho_d$ term in $\tilde{R}_d$ and appears in front of the odd terms in the series for $\tilde{\sigma_d}$. 
\end{remark}

The tensor $\rho$ in Lemma \ref{Rfromr} will be chosen to reduce the entanglement between tensors $X'$ and $Y'$ in Eq. \eqref{XYdis1}; see Remark \ref{disent-origine}. How to do this and why this is important, will become clear in Section \ref{RG_lowT} where we construct the RG map.

\section{From lattice models to tensor networks}\label{models2networks}

In this section we will discuss how partition functions of lattice spin models 
can be transformed to TN form. In other words, we want to find a TN whose partition function equals the partition function of a lattice model. We will start with the nearest-neighbor (NN) 2D Ising model, and then
consider more complicated models.

\subsection{NN Ising model}\label{nnising}

Consider the 2D square-lattice ferromagnetic NN Ising model in a magnetic field, defined by the partition function\footnote{\label{hh0}$\beta=1/T$ with $T$ the temperature. Note that $h=\beta h_0$, where $h_0$ is the temperature-independent magnetic field. For us it will be more convenient to work in terms of $h$ and not $h_0$.}
\begin{equation}
	Z = \sum_{\{ s \}} e^{H},\qquad H = \beta \sum_{\< i j \>} (s_i s_i-1) + h \sum s_i, \quad s_i = \pm 1\,.
	\label{NNIsingdef}
\end{equation}
We would like to represent this partition function by a
TN. The simplest way to do this is as follows \cite{TEFR}. Split the
spin lattice into black and white squares like a chessboard. For a black square $B$ having spins $s_1,s_2,s_3,s_4$ at its vertices, we define the quantity
\begin{equation}
	H_B(s_1,s_2,s_3,s_4) = \beta (s_1 s_2 + s_2 s_3 + s_3 s_4 + s_4 s_1-4) +
	\frac{h}{2} (s_1 + s_2 + s_3 + s_4)\,.
\end{equation}
We also define a tensor $A=A(\beta,h)$ with four indices $s_1,s_2,s_3,s_4$ corresponding to the spin values at the vertices of a black square and whose elements are given by
\begin{equation}
	A_{s_1 s_2 s_3 s_4} = e^{H_B(s_1,s_2,s_3,s_4)}= e^{\beta (s_1 s_2 + s_2 s_3 + s_3 s_4 + s_4 s_1-4) +
		\frac{1}{2} h (s_1 + s_2 + s_3 + s_4)}. \label{Ntensor}
\end{equation}
The point is that summing $H_B$ over all black squares, we get $H$ for the given spin configuration. This guarantees that the TN contraction of $A$'s reproduces exactly the spin model partition function $Z$. This is illustrated by this figure:
\begin{equation}
\myinclude{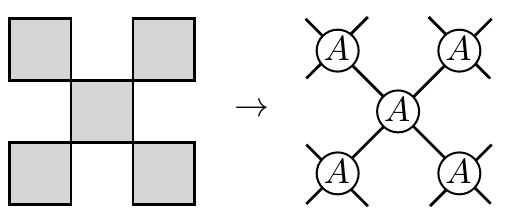} .\label{checkerboard}
\end{equation}

\begin{remark} In this work we will always consider the usual periodic boundary conditions (b.c.) for tensors, as in \eqref{TNex}. Since in \eqref{checkerboard} the tensor lattice is rotated by 45 degrees, the usual periodic b.c.~for tensors correspond to 45-degree-rotated b.c.~for spins, and vice versa. There exist other constructions where the tensor network lattice is not rotated and the usual periodic b.c.~for tensors and spins are mapped to each other. See e.g.~App.~A.2 of \cite{paper1} or Section \ref{GenFinRange} below.
\end{remark}

It is interesting to examine the tensor $A$ given by \eqref{Ntensor} at high and low
temperatures. The high-T fixed point corresponds to setting $\beta =
0$ and $h = 0$. In this case the tensor
\begin{equation}
	A_{s_1 s_2 s_3 s_4} = 1 \quad \forall s_i=\pm 1\,. \label{HT0}
\end{equation}
Let us do an orthogonal gauge transformation corresponding to a change of basis from the states $| \pm \rangle$ to their $\mathbb{Z}_2$-even and odd combinations $| 0 \rangle ,  |1 \rangle=
\frac{1}{\sqrt{2}} (| + \rangle \pm | - \rangle)$. The high-T fixed point tensor \eqref{HT0} has only one nonzero component in this basis.
After rescaling by a factor of $1/16$, it is given by: 
\begin{equation}
	(A_{\ast})_{0000} = 1\ . \label{HTFP}
\end{equation}
If we turn on small $\beta$ or small magnetic field $h$ we will have, in the $0,1$ basis, 
\begin{equation}
	A = A_{\ast} + B,\quad B=B(\beta,h),\quad\| B \| \le O (\beta)+O(h)\ .
\end{equation}
We
can then apply the tensor RG from our previous work {\cite{paper1}} to get
recursively closer to the high-T fixed point. We can also apply
a cluster expansion to study this regime. We will review these considerations in Section \ref{HighT} below.\footnote{\label{whereapplies}We can also consider arbitrary $h$. For $\beta=0$ we have $A_{s_1 s_2 s_3 s_4}=e^{(h/2)s_1} e^{(h/2) s_2} e^{(h/2) s_3} e^{(h/2) s_4}$. We define the basis $|0\rangle \propto e^{h/2}|+\rangle +e^{-h/2}|-\rangle$, $|1\rangle$ an orthogonal state. In the new basis $A$ is proportional to $A_*$. Applying the same rotation at a small nonzero $\beta$ we get $A=A_*+O(\beta)$. Thus our previous work {\cite{paper1}} applies.}

Consider next low temperatures. The low-T fixed point corresponds to taking the limit $\beta
\rightarrow \infty$ in {\eqref{Ntensor}} at $h = 0$. The limiting tensor has two nonzero components:
\begin{equation}
	A_{+ + + +} = A_{- - - -} = 1 . \label{LTFP}
\end{equation}
In this case we don't apply a gauge transformation and we analyze the tensor directly in the $\pm$ basis. We introduce tensors $A^{(\pm)}$ each having just one of the two nonzero components:
\begin{equation}
	(A^{(q)})_{qqqq}=1,\qquad q=\pm \label{Apm}\ .
\end{equation}
Each of these two tensors is by itself a tensor RG fixed point of the same form as
the high-T fixed point \eqref{HTFP}, up to renaming of the state $0$ by $\pm$. In other words, the low-T
fixed point is the direct sum
\begin{equation}
	A_{\rm LT}= A^{(+)} \oplus A^{(-)}\ . \label{LTFP1}
\end{equation}
of two high-T
fixed points, representing the positive and negative
magnetized phases of the Ising model. 

Consider next the vicinity of the low-T fixed point. First let us take a finite large $\beta$ keeping $h = 0$. In this case we have
\begin{equation}
	A = A_{}^{(+)} + A^{(-)} + \delta A\ .
\end{equation}
where $\| \delta A \| = O (e^{- 4 \beta})$, $(\delta A)_{+ + + +} = (\delta
A)_{- - - -} = 0$, and $\delta A$ respects $\mathbb{Z}_2$ invariance. We expect
that such a tensor $A$ should flow to the low-T fixed point {\eqref{LTFP1}} under an
appropriate tensor RG transformation. This will be one of our results in
this paper.

Finally let us take large $\beta$ and nonzero $h$. Then we have
\begin{equation}
	A_{+ + + +} = e^{2 h}, \quad A_{- - - -} = e^{- 2  h},
\end{equation}
while other components are down by $e^{- 4 \beta}$. 
If we take $h$ large then this tensor, after rescaling, will be one of the
fixed points $A^{(+)}$ or $A^{(-)}$ plus a small perturbation.\footnote{\label{whereapplies1}One can see from Eq.~\eqref{Bexpl} below that this is true for $|h|$ large and any $\beta$, large or small. Combining this remark with footnote \ref{whereapplies}, we conclude that our previous work {\cite{paper1}} covers $|h|\gg 1$, $\beta$ any, as well as $\beta\ll 1$, $h$ any.} In this
situation, the tensor RG of our previous work {\cite{paper1}} applies. Also,
a cluster expansion can be used to show that the model is in the corresponding
magnetized phase.

On the other hand, if we take $\beta$ large and $h=O(1)$, the $A_{+ + + +}$ and $A_{- - - -}$ components will be comparable.
So the initial tensor is not one of the two high-T fixed points plus a small
perturbation. Our previous result from {\cite{paper1}} on the stability of the
high-T fixed point cannot be applied here. This will be our main case of interest in the present paper. 
After rescaling by $e^{2h} + e^{- 2 h}$, the initial
tensor $A(\beta,h)$ can be written as
\begin{eqnarray}
	A &=& \alpha A^{(+)} + (1 - \alpha) A^{(-)} + B,\qquad\text{where}\label{nearLT}\\
	\alpha &=& {e^{2 h}}/({e^{2h} + e^{- 2 h}}) \in (0, 1\ ,)\nn\\
	B&=&B(\beta,h),\quad \| B \| = O (e^{- 4 \beta}), \quad B_{+ + + +} = B_{- - - -} = 0\ . \nn
\end{eqnarray}
Explicitly, tensor $B$ has the following three nonzero components (up to cyclic permutations of indices, and the $\mathbb{Z}_2$ transformation $+\leftrightarrow -$ accompanied by $h\to -h$) 
\begin{eqnarray}
 \label{Bexpl}	B_{+++-} &=& e^{-4\beta +h}/(e^{2h} + e^{- 2 h})\ ,\\
	B_{++--} &=& e^{-4\beta}/(e^{2h} + e^{- 2 h})\ ,\nn\\
	B_{+-+-} &=& e^{-8\beta}/(e^{2h} + e^{- 2 h})\ .\nn
\end{eqnarray}
In this paper, we will devise a new form of tensor RG which acting 
a tensor of the form \eqref{nearLT} recursively will converge towards $A^{(+)}$ or $A^{(-)}$ depending on
whether the initial field $h$ was positive or negative, no matter how small.
This will be developed in Sections \ref{LTRG}, \ref{RG_lowT} and \ref{LTprop}.

\subsection{General finite-range model}\label{GenFinRange}

The previous section gave one way to represent the nearest-neighbor 2D Ising model as a TN. Here we will show a general result: an arbitrary spin model with finite range interactions can be reduced to a TN. 

\begin{prop}
	\label{finite-range} The partition function of any spin model with discrete spin space on the 2D square
	lattice with finite-range translation-invariant interactions can be
	represented in tensor network form.
\end{prop}

\begin{proof} It is sufficient to consider a spin model where interactions have range at
	most $\sqrt{2}$. This means that all interactions fit into $2 \times 2$
	plaquettes. A general finite-range interaction can be reduced to
	this case by dividing the lattice into sufficiently large squares and
	introducing new spin variables corresponding to all possible configurations of
	the original spins in each square.
	
	We can write $H$ as $\sum_X V_X$ where $X$ is summed over the 2 by 2 plaquettes, and $V_X$ only depends on the spins in $X$. To get this representation, the NN interactions along an edge have to be equally split between the 2 plaquettes sharing that edge. Analogously, the single-spin interactions localized at a vertex have to be split among the four plaquettes sharing this vertex.
	
	Consider one $2 \times 2$ plaquette with spins $s_1, s_2, s_3, s_4 \in \mathcal{S}$ at
	its vertices, where $\mathcal{S}$ is the spin space of the model, assumed finite.
        Then we define a 4-index tensor $B$ with $s_1, s_2, s_3, s_4$ as indices, and
	components given by
	\begin{eqnarray}
		\label{figBspin}
		\myinclude{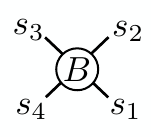} & = & \exp(V_{s_1 s_2 s_3 s_4})\ .
	\end{eqnarray}
	Then the partition function can be written as a tensor network:
	\begin{equation}
		\myinclude{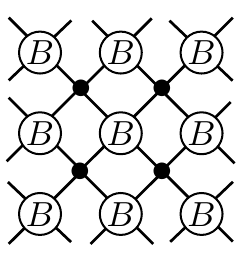}
		\label{BspinTN}
	\end{equation}
	made of the tensor $B$ and the Kronecker 4-tensor $\bullet$ which simply sets the
	indices of four $B$ tensors to the same value, i.e.
	\begin{equation}
\myinclude{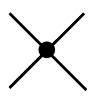} = \delta_{s_1
			s_2 s_3 s_4}\ .
	\end{equation}
	The Kronecker 4-tensor can be written as a contraction of 4 Kronecker
	3-tensors:
	\begin{equation}
\myinclude{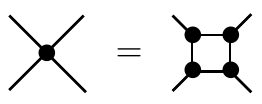}
	\end{equation}
	We see that we can write network \eqref{BspinTN} equivalently in terms of one 4-tensor $A$
	with index space $\mathcal{S} \times \mathcal{S} $, defined as:
	\begin{equation}
\myinclude{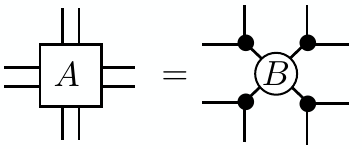}\ .
\label{figAspin}
	\end{equation}
	This proves the proposition. \end{proof}

Let us illustrate this construction on a couple of examples. 

{\it Example 1.} Consider the 2D square-lattice $\mathbb{Z}_2$-invariant Ising
model with nearest-neighbor, next-to-nearest-neighbor, and four-spin plaquette interactions:
\beq
H = J \sum_{\langle i j \rangle} s_i s_j + J' \sum_{\langle \langle
	i j \rangle \rangle} s_i s_j + J_{\Box} \sum_{\langle i j k l \rangle} s_i
s_j s_k s_l . 
\eeq
In this case $\mathcal{S} = \{ +, - \}$. We have 
\beq
H= \sum_X V_X,\quad V_{s_1 s_2 s_3 s_4}=\frac{J}2 (s_1s_2+s_2 s_3+s_3 s_4+s_4 s_1)+ J'(s_1 s_3+s_2 s_4) + J_{\Box} s_1 s_2 s_3 s_4.
\eeq
 Tensors $B$ and $A$ have 4 nonzero components
(up to rotation and $\mathbb{Z}_2$ transformations):
\begin{equation}
\myinclude{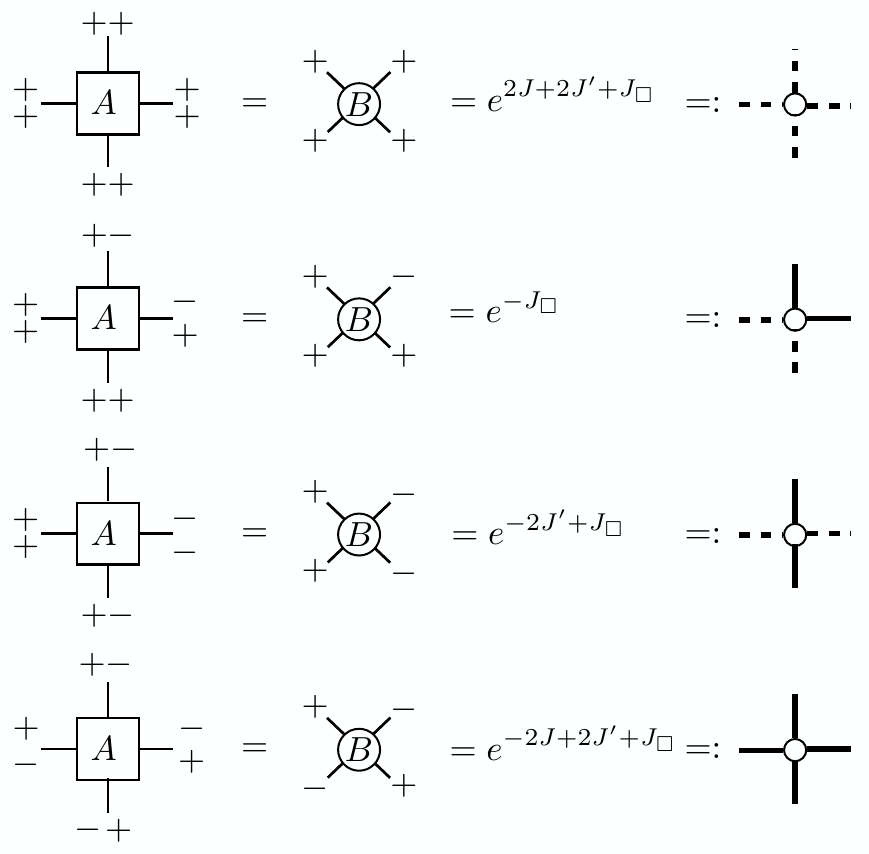}\ . \label{dw}
\end{equation}
The last column of {\eqref{dw}} corresponds to a trick by which one can reduce the bond dimension down to 2 in this special case \cite{Fan-Wu}. Nonzero components of $A$ are in one-to-one
correspondence with configurations of domain walls separating pairs of spins
indexing the legs of $A$, and these define components of a
4-tensor with index space $\{ \text{wall ( $|$ )}, \text{no wall ( $\brokenvert$ )} \}$, given in the last column of {\eqref{dw}}.
The tensor network made of this
``domain wall tensor'' reproduces the partition function up to a factor of 2
(because each domain wall configuration has exactly two spin
configurations giving rise to it). 

{\it Example 2.} Consider next the 2D square-lattice Ising
model with the nearest-neighbor coupling and two $\mathbb{Z}_2$-breaking interactions: magnetic field and the cubic coupling of three nearby spins. The Hamiltonian is
\beq
\label{IsingCubic}
H = \beta \sum_{\langle i j \rangle} s_i s_j + h \sum_{i}
 s_i  + \gamma \sum_{\langle i j k \rangle} s_i
s_j s_k ,
\eeq
where the last term includes the sum over all triangles which fit into a $2\times2$ plaquette.

In this example we have $H= \sum_X V_X$ with
\beq
V_{s_1 s_2 s_3 s_4}=\frac{\beta}2 (s_1s_2+s_2 s_3+s_3 s_4+s_4 s_1)+ \frac{h}{4}(s_1+s_2+s_3+ s_4) + \gamma (s_1 s_2 s_3 + s_2 s_3 s_4+ s_3 s_4 s_1 + s_4 s_1 s_2).
\eeq
Tensor $A$ is easily obtained from the general formulas \eqref{figBspin},\eqref{figAspin}; it is indexed by $\mathcal{S}\times \mathcal{S}$ with $\mathcal{S} = \{ +, - \}$. The all $+$ and all $-$ components are given by:
\beq
\myinclude{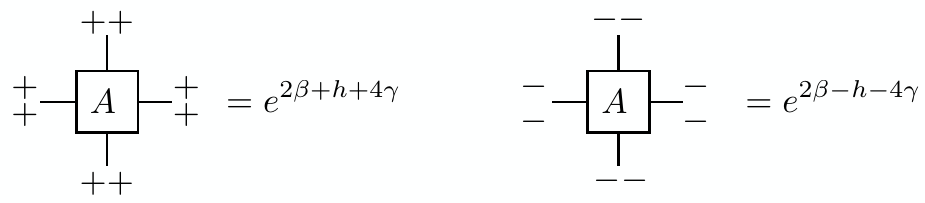}\ \ ,
\eeq
while other nonzero components are (up to rotation and $\mathbb{Z}_2$ transformations which interchange $+$ with $-$ and simultaneously flip the sign of $h$ and $\gamma$):
\begin{equation}
	\myinclude{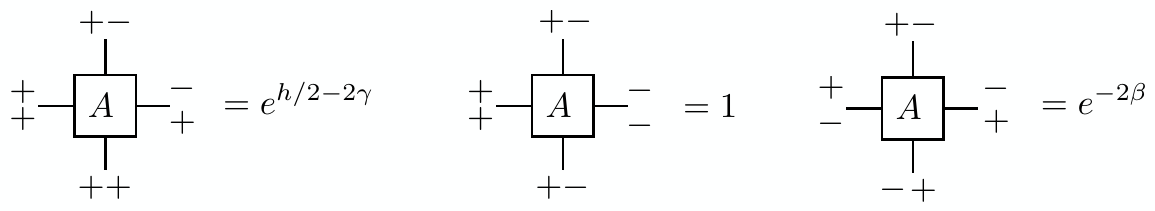}\ .
\end{equation}
We wish to examine the structure of $A$ in the low-T limit. Let us reindex states $\{++,--,+-,-+\}$ by $\{+,-,1,2\}$. Rescaling $A$ by $e^{2\beta}(e^{h+4\gamma}+e^{-h-4\gamma})$, we obtain a tensor of the form
\beq
\alpha A^{(+)}+(1-\alpha) A^{(-)}+B\ ,
\label{ex2}
\eeq
where $\alpha= e^{h+4\gamma}/(e^{h+4\gamma}+e^{-h-4\gamma})$ and $B=O(e^{-2\beta})$.

Compare this to what we had in \eqref{nearLT} for the NN Ising model in magnetic field (but without the cubic coupling). In both cases we have a convex linear combination of $A^{\pm}$ with a correction which is small for $\beta\gg 1$.\footnote{Recall that \eqref{nearLT} was obtained by a different method, which works for the NN Ising but not for more general models, and which results in a tensor network having half as many tensors. This explains the difference in coefficients of $\beta$ and $h$ between \eqref{nearLT} and \eqref{ex2}.} In both cases $\alpha$ is a monotonic function of $h$. These are the similarities.

The main difference is that in \eqref{nearLT} there is a special value $h_0=0$ for which $\alpha=1/2$ and, moreover, $B$ preserves $\mathbb{Z}_2$ invariance. On the other hand, in \eqref{ex2} we have $\alpha=1/2$ at $h_0=-4\gamma$, and the Hamiltonian does not have any special symmetry at this value of $h$. 

This difference will become important when discussing the first-order phase transition at low $T$ as a function of $h$. At $T=0$ the transition happens at $h_0$ in both models. At $T$ slightly above zero, the transition in \eqref{nearLT} is still at $h=0$, while in \eqref{ex2} its location is expected to shift. We will discuss this further in Section \ref{LTRG}.

It is not hard to convince yourself that the observed features of \eqref{nearLT} and \eqref{ex2} are in fact generic. Namely, the following proposition is true.

\begin{prop}
Take any perturbation of the NN Ising model by a magnetic field $h$ and some other couplings:
\beq
\label{smallpert}
H = \beta \sum_{\langle i j \rangle} s_i s_j + h \sum_{i}
s_i  + \ldots ,
\eeq
where $\ldots$ stands for a finite number of translationally invariant interactions, all having finite range, and controlled by couplings $\gamma_1,\ldots,\gamma_N$. We assume that the $\gamma_i$ are fixed, while $h\in \mathbb{R}$ and $\beta>0$ can vary. Reduce this model to a tensor network using the method of this section. The obtained tensor, after rescaling, can then be written in the form \eqref{ex2} where $B=B(\beta,h)= O(\exp(-const.\beta))$ in the large $\beta$ limit (uniformly in $h$) while $\alpha=\alpha(h)$ is a monotonic function of $h$ such that 
\beq
\lim_{h\to- \infty} \alpha=0,\qquad \lim_{h\to+ \infty} \alpha=1\ .
\eeq
If all couplings $\gamma_i$ preserve $\mathbb{Z}_2$, then at $h=0$ we have
$\alpha=1/2$, and $B$ preserves $\mathbb{Z}_2$ invariance. 
\end{prop}

\section{Tensor RG analysis near the high-T fixed point}
\label{HighT}

The tensor RG near the high-T fixed point was studied in {\cite{paper1}}. Here we will review that construction. This will be useful as preparation for the subsequent sections, where we discuss the vicinity of the low-T fixed point.

In {\cite{paper1}} we worked with real tensors, but here we will consider
complex tensors. This will allow us to emphasize the analyticity\footnote{See Appendix \ref{abstract}.} of the RG map. 

In this section $\mathcal{H}$ will denote the complex Hilbert space of complex 4-tensors indexed by $\mathcal{I}=\mathbb{Z}_{\ge 0}$ equipped with the HS norm. $\H0$ will be the complex Hilbert space of such tensors satisfying in addition the condition $B_{0000} = 0$. Let $U_{\varepsilon}$ be the ball $\|B\| < \varepsilon$ in $\H0$. 

We will consider tensors $A$ which are small perturbations of the high-T fixed point tensor $A_{\ast}$ with the only nonzero component $(A_{\ast})_{0000} = 1$:
\begin{equation}
	A = A_{\ast} + B , \quad B \in
	U_{\varepsilon}. 
\label{AstpB}
\end{equation}

We will consider tensor RG transformations which map $A\mapsto A'$ and preserve the partition function as described in Section \ref{tensorRG}. The lattice rescaling factor will be $b=4$. The tensor $A'$ will have the form
\begin{equation}
	A'=\normfactor (B) \cdot [A_{\ast} + R (B)], \quad R (B)\in\H0,
\end{equation}
where $\normfactor (B)$ is a scalar called the normalization factor.\footnote{The normalization factor was denoted by letter $\mathcal{N}$ in \cite{paper1}.} 

Factoring out the normalization factors, we can equivalently write the equation \eqref{genTRG} expressing preservation of the partition function as
\begin{equation}
	Z(A_\ast+B ,L_x,L_y) =\normfactor (B)^{L_x L_y/b^2} Z\Bigl (A_{\ast} + R (B), {L_x}/{b},{L_y}/{b}\Bigr), \label{RGproperty}
\end{equation}
This equation will play a significant role below.

A tensor RG map $\tau=(U_\eps,\nu,R)$ is thus specified by the $U_\epsilon$ neighborhood of the HT fixed point where it is defined, and by two functions $\normfactor :U_\epsilon\to \mathbb{C} $ and $R :U_\epsilon\to \H0$ defined on this neighborhood, so that Eq.~\eqref{RGproperty} holds. 

We say that the high-T fixed point is stable if for a sufficiently small $\epsilon$, there exists a tensor RG map such that $\|R(B)\|<\gamma \|B\|$ ($B\in U_\epsilon$) with $\gamma<1$. The following theorem implies that the high-T fixed point is indeed stable. Later on we will see this fact can be used to establish analyticity of the free energy near the high-T fixed point.

\begin{theorem}[{\cite{paper1}}]
	\label{HTmap}
	There exists an $\bar \epsilon>0$ such that for any $\epsilon\in (0,\bar \epsilon)$ there exists a tensor RG map $\tau_\epsilon=(U_\eps, \normfactor ,R)$ with lattice rescaling factor $b=4$, preserving the partition function as in \eqref{RGproperty}, and having the following
	additional properties:\\[-14pt]
	\begin{subequations}
		\begin{align}
			& \text{$\normfactor$ and $R$ are analytic, taking real values for real $B$\ ;}\\
			& \normfactor  (B) = 1 + O (\| B \|^2)\ ; \label{NBbound}\\
			& \| R (B) \| \leqslant C \epsilon^{1/2}  \| B \|\ . \label{RBbound}
		\end{align}
	\end{subequations}
\end{theorem}

\begin{proof} 
	The proof consists in critically rereading the argument in
	{\cite{paper1}}, paying attention to two main differences. First, Ref.
	{\cite{paper1}} assumed that $B$ was real. However, all maps were given by
	power series expansions in $B$ so they make perfect sense for complex $B$ and
	are analytic when convergent. It's easy to check that there exists an $\bar \epsilon$ so that the series are convergent for $\|B\|<\epsilon$, $\epsilon <\bar \epsilon$. The reality condition holds because the series
	coefficients are all real. This proves (a).
	
	The second difference is that Ref. {\cite{paper1}} stated (b)
	and (c) in a weaker form as $\normfactor  (B) = 1 + O
	(\varepsilon^2)$ and $\| R (B) \| = O (\varepsilon^{3 / 2})$, and one has to
	check that the same argument in fact proves the stronger bounds stated here. We
	will provide here the roadmap to such a check.
	
	The tensor RG map in {\cite{paper1}} is a product of maps from several steps:
	\begin{eqnarray}
		\text{Type 0:} &  & A_{\ast} + B_{} \mapsto \normfactor _1 \cdot [A_{\ast} + B_1]\ ,
		\\
		\text{Type I} : &  & A_{\ast} + B_1 \mapsto \normfactor _2 \cdot [A_{\ast} + B_2]\ ,
	\nn	\\
		\text{Type II:} &  & A_{\ast} + B_2 \mapsto \normfactor _3 \cdot [A_{\ast} + B_3]\ ,
		\nonumber
	\end{eqnarray}
	with $R (B) = B_3$.
	
	Type 0 step does not change the number of tensors, while Type I and Type II
	steps have rescaling factor $b = 2$. The overall rescaling factor is thus $b =
	4$. We also have $\normfactor  (B) =(\normfactor _1)^{16} (\normfactor _2)^4
	\normfactor _3$.
	
	As mentioned one easily checks that $\normfactor _{i + 1}, B_{i + 1}$ are
	holomorphic in $B_i$ for each step ($B_0 \equiv B$). It's also easy to check
	that
	\begin{gather}
	\label{B1B2}	\normfactor _1 = 1 + O (\| B_0 \|^2), \quad \normfactor _2 = 1 + O (\| B_1
		\|^4), \quad \normfactor _3 = 1 + O (\| B_2 \|^2),   \\
		\| B_1 \| \leqslant \tmop{const} \| B_0 \|, \qquad \| B_2 \| \leqslant
		\tmop{const} \| B_1 \| \ .    \nn
	\end{gather}
	The bound on $B_3$ is a bit more subtle, because the Type II step map depends
	explicitly on $\varepsilon$ (it is the only such step). 
	In {\cite{paper1}}, Prop. 2.3, we proved the bound $\| B_3 \| = O
	(\varepsilon^{3 / 2})$. Following the argument in {\cite{paper1}} with a minor modification (see Remark \ref{modification} in the proof of Theorem \ref{LTmap} below), it is possible to prove the following more refined bound:
	\begin{equation}
		\| B_3 \| \leqslant \tmop{const} . (\varepsilon^{1 / 2} \| B_2 \| +
		\varepsilon^{- 1 / 2} \| B_2 \|^2) \leqslant C \varepsilon^{1 / 2} \| B_2 \|\,.
		\label{B3}
	\end{equation}
The first inequality in {\eqref{B3}} comes from the fact that
	some terms which are quadratic or higher order in $B_2$ are rescaled by $\varepsilon^{-1 / 2}$,
        while some other terms which multiply them and are linear in $B_2$ are rescaled by
	$\varepsilon^{1 / 2}$. The tensor network value depends on the product of the
	rescaled terms and stays invariant. This rescaling by $\varepsilon^{\pm 1 / 2}$ was one of the main tricks in {\cite{paper1}}. The second inequality in {\eqref{B3}} then follows from  $\| B_2 \|=O(\eps)$.
	
	Putting together {\eqref{B1B2}} and {\eqref{B3}} we get \eqref{NBbound} and \eqref{RBbound}.
\end{proof}

We will now show how the above result can be used to study the infinite-volume limit and the
analyticity of the free energy for tensor networks. We denote by $f (B, L)$ the free energy per site of the $L \times L$ network made of tensors
{\eqref{AstpB}}:
\begin{equation}
	f (B, L) := \frac{1}{L^2} \ln Z (A_\ast+B, L,L) .
\end{equation}
The analyticity of $f (B, L)$ for small $\| B \|$, uniformly in $L$, can be shown
using a cluster expansion, as discussed in {\cite{paper1}}, App. B. Here we
would like to recover a similar conclusion without a cluster expansion, but via tensor RG.

Note first of all that for any finite $L$ we can obtain analyticity of $f (B, L)$ for $\| B
\| \le \delta_L$ (not uniformly in $L$) via the following simple expansion
argument. We expand the tensor network in powers of $B$, getting $2^{L^2}$
terms. Estimating each term by $\| B \|$ to the corresponding power, we get
the bound
\begin{equation}
	Z (A_\ast+B, L,L) = 1 + z, \qquad | z | \leqslant (1 + \| B \|)^{L^2} - 1 < e^{L^2
		\| B \|} - 1 .
\end{equation}
If we assume that
\begin{equation}
	\|B\| <\delta_L:= 1/ (2 L^2) , \label{deltaL}
\end{equation}
then $| z | < e^{1/2} - 1 < 1$ and $Z (A_\ast+B, L,L)$ is strictly separated from
zero. We obtain the following lemma.

\begin{lemma}
	\label{2x2}$f (B, L)$ is bounded and analytic for $\| B \| < \delta_L$.
\end{lemma}

To reach the infinite volume limit, we will combine this lemma with tensor RG,
as follows.

\begin{prop} \label{HTfree}
	Let $\varepsilon < \min (\bar\epsilon, 1 / C^2, \delta_4)$ where $\bar\epsilon$ and $C$ are the constants in Theorem \ref{HTmap}. Then
	\begin{itemize}
		\item[(a)]
		$f (B, 4^n)$ is a bounded analytic function of $B \in U_{\varepsilon}$, uniformly bounded in $n=1,2,\ldots$
		
		\item[(b)] The limit (infinite-volume free energy)
		\begin{equation}
			f (B) = \lim_{n \rightarrow \infty} f (B, 4^n)
		\end{equation}
		exists and is a bounded analytic function of $B \in U_{\varepsilon}$.
		
		\item[(c)]
		We have $f (B) = O (\| B \|^2)$.
	\end{itemize}
\end{prop}

\begin{proof} We use the RG map
	$\tau_{\varepsilon} = (U_\epsilon, \normfactor , R)$ provided by Theorem \ref{HTmap}. We denote $g (B) = \ln \normfactor  (B)
	= O (\| B \|^2)$. Taking the logarithm of Eq.~\eqref{RGproperty}, we get
	\begin{equation}
		f (B, L) = \frac{1}{16} [g (B) + f \bigl(R (B), L / 4\bigr) ] \label{RGprop}.
	\end{equation}
	This equation expresses the fact that the partition function is preserved under a tensor RG step, while the number of lattice sites is reduced by $16 = b^d$ with $b=4$ the rescaling factor of the RG map and $d = 2$ the number of lattice dimensions.
	
	We iterate the RG map and let $B_j$ be the resulting RG trajectory with $B_0 =
	B$ and $B_{j + 1} = R (B_j)$. Eq. {\eqref{RBbound}} guarantees that all $B_j
	\in U_{\varepsilon}$. Iterating {\eqref{RGprop}}, we have
	\begin{equation}
		f (B, L) = f_n (B) + 16^{- n} f (B_n, L / 4^n) , \qquad f_n (B) :=
		\sum_{j = 1}^n 16^{- j} g (B_{j - 1}) .
	\end{equation}
	We will use this equation for $L = 4^{n + 1}$. We iterate the RG $n$ times, until
	the lattice becomes $4 \times 4$. We have
	\begin{equation}
		f (B, 4^{n + 1}) = f_n (B) + 16^{- n} f (B_n, 4)\ . \label{upto2x2}
	\end{equation}
	The first term in the r.h.s. is an analytic function of $B \in
	U_{\varepsilon}$. Moreover it is uniformly bounded in $n$ and converges
	uniformly to
	\begin{equation}
		f_{\infty} (B) := \sum_{j = 1}^{\infty} 16^{- j} g (B_{j - 1}) .
		\label{finfty}
	\end{equation}
	This follows from the boundedness of $g$. Since $f_n (B)$ are analytic,
	uniformly bounded and converge uniformly, the limiting function $f_{\infty}
	(B)$ is also bounded and analytic.
	
	As to the second term in the r.h.s.~of
	{\eqref{upto2x2}}, it is a bounded analytic function of $B \in U_{\varepsilon}$ by Lemma \ref{2x2}.
	Moreover it tends uniformly to zero as $n \rightarrow \infty$.
	
	Putting these observations together, we see that $f (B, 4^{n + 1})$ is a
	uniformly in $n$ bounded analytic function of $B \in U_{\varepsilon}$ which uniformly
	converges to $f_{\infty} (B)$ as $n \rightarrow \infty$. This proves part (a),
	and part (b) with $f (B) = f_{\infty} (B)$.
	
	Since $g (B) = O (\| B \|^2)$ and $\| B_j \| \leqslant \| B \|$ for all $j$,
	we also have $f_{\infty} (B) = O (\| B \|^2)$, proving part (c).
\end{proof}

\begin{remark}
	More generally, we could start the RG when $L = k 4^n$ and iterate until $L =
	k$. The above proof corresponds to $k = 4$. For general $k$, Eq.
	{\eqref{upto2x2}} becomes $f (B, k 4^n) = f_n (B) + 16^{- n} f (B_n, k) $. We
	would then prove that $f (B, k 4^n)$ is analytic for $\| B \| < \min (\epsilon_0, 1 /
	C^2, \delta_k)$ and has an infinite volume limit as $n \rightarrow \infty$.
	This infinite volume limit is given by the same formula {\eqref{finfty}} and
	does not depend on $k$.
\end{remark}

\begin{remark}
	In the above proof, we think of $\varepsilon$ as small and fixed. The RG map
	$\tau_{\varepsilon}$ adapted to $U_{\varepsilon}$ is chosen at the beginning
	and kept fixed for the rest of the argument. It is possible to do things
	differently and reduce $\varepsilon$ as we iterate the RG map. Such a point
	of view was taken originally in {\cite{paper1}}, where we were choosing
	$\varepsilon$ adapted to the size of $B$, i.e. $\varepsilon \sim \| B \|$. The fixed point is then approached superexponentially fast.
	
	To see an example of this in the present context, consider the sequence of
	RG maps $\tau_{\varepsilon_j} = (U_{\epsilon_j}, \normfactor _j, R_j)$ provided by Theorem
	\ref{HTmap} and adapted to the balls of $U_{\varepsilon_j}$ of radii $\varepsilon_0, \epsilon_1,\ldots $ where 
	\begin{equation}
		\varepsilon_0 < \min (\bar \epsilon, 1 / (2 C^2), \delta_4), \qquad \varepsilon_{j + 1} =
		C \varepsilon_j^{3 / 2} .
	\end{equation}
	Note that $\epsilon_{j}$ form a decreasing sequence tending to zero superexponentially fast. By Eq. {\eqref{RBbound}} we have
	\begin{equation}
		R_j : U_{\varepsilon_j} \rightarrow U_{\varepsilon_{j + 1}} .
	\end{equation}
	So we can run the above proof, applying the RG maps $T_{\varepsilon_j}$ in
	succession. We obtain that $f (B)$ is analytic on $U_{\varepsilon_0}$ and
	has a representation similar to {\eqref{finfty}} but with
	$j$-dependent terms:
	\begin{equation}
		f_{} (B) = \sum_{j = 1}^{\infty} 16^{- j} g_j (B_{j - 1}), \quad g_j : =
		\ln \normfactor _j .
	\end{equation}
	Since the radii $\varepsilon_j$ tend to zero superexponentially fast, this
	representation converges even faster than {\eqref{finfty}}.
\end{remark}

\section{RG analysis at low $T$: general strategy}
\label{LTRG}

We will now move towards discussing the 2D Ising model \eqref{NNIsingdef}, and its small perturbations like \eqref{smallpert}, at low temperatures. Let us state some well-known facts. The phase diagram of the NN Ising model in the $(T, h)$ plane
contains, at $T < T_c$ and $h = 0$, a line of first-order phase transitions:
\begin{equation}
\myinclude{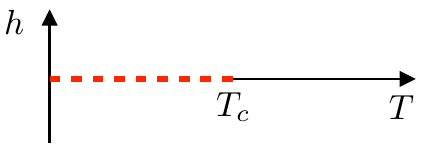} .
	\label{phase}
\end{equation}
We define the infinite-volume free energy per spin and magnetization as
\begin{eqnarray}
	{f_{\tmop{spin}}}  (T, h) = \lim_{L\to \infty} (\ln Z)/L^2\,,\qquad
	m_{\tmop{spin}} (T,h)   =  \partial f_{\tmop{spin}} / \partial h\,.
\end{eqnarray}
Away from the phase transition line, these are continuous (in fact analytic) functions. Approaching the phase transition line from above and below, the free energy is continuous while the magnetization experiences a discontinuous jump. 

The standard proofs of these facts for the NN Ising model (see e.g.~\cite{friedli_velenik_2017}) use the Peierls argument to show that the magnetization is nonzero at $h=0$ with "all $+$" or "all $-$" boundary conditions. (One can then use correlation inequalities to show that magnetization only increases in absolute value as $|h|$ is increased. The Lee-Yang theorem shows that the free energy is analytic at $h\ne 0$.) This crucially uses that $h$ is the only coupling breaking the spin flip symmetry, and that the transition is at $h=0$.

Consider now one of perturbations of the NN Ising model which involve several other couplings breaking the spin flip symmetry $\gamma_i$, like \eqref{IsingCubic} or \eqref{smallpert}. We still expect a first-order phase transition at some value of $h=h_c$ but now $h_c$ will depend on $\gamma_i$ and $T$. (The value of $h_c$ at $T=0$ is easy to determine by comparing the energy of the "all $+$" and the "all $-$" states. E.g.~for the model \eqref{IsingCubic} it happens for $h=-4\gamma$.) The existing proofs of the first-order transition require either the Pirogov-Sinai theory {\cite{Pirogov1975,Pirogov1976,sinai-book}}, or the coarse-graining approach of \cite{Gawedzki1987}. Both these approaches use the contour representation of the partition function, but they are more subtle than the Peierls argument for the Ising model at $h=0$.


In the remainder of the paper, our goal will be to develop a rigorous tensor RG that can be applied in the low temperature regime $T\ll 1$ and used to recover the expected properties of the free energy and magnetization. Our results will apply also to general small perturbations of the Ising model which do not necessarily preserve the spin flip symmetry. 

As discussed in Section \ref{models2networks}, the partition function of such models can be represented as a tensor network built from
\begin{equation}
	A = A(\alpha,B) = \alpha A^{(+)} + (1 - \alpha) A^{(-)} + B  ,
	\label{LTTN1}
\end{equation}
where $A^{(\pm)}$ are the two tensors defined in \eqref{Apm}, $\alpha\in(0,1)$ is a function of $h$ and the other couplings $\gamma_i$, and $B$ is a tensor which is small at low temperatures. For the NN Ising model we have $\alpha = \frac{e^{2 h}}{e^{2 h} + e^{- 2 h}}$ and $B=B(\beta,h)$ is given in Eqs.~\eqref{nearLT}, \eqref{Bexpl},  with $\| B \| = O (e^{- 4 \beta})$.

We denote by $f = f (\alpha, B)$ the free energy per tensor of the network made of tensors \eqref{LTTN1}. For the NN Ising model, it is related to $f_{\rm spin}$ by:\footnote{The factor $1/2$ is because there is one tensor per two spins. The term $\ln (e^{2 h} + e^{- 2 h})$ is added in parentheses because the tensor was rescaled by $e^{2 h} + e^{- 2 h}$, as mentioned before Eq.~\eqref{nearLT}.}
\begin{equation}
	f_{\tmop{spin}} = \frac{1}{2} [ f + \ln (e^{2 h} + e^{- 2 h})]\ .
	\label{twofs}
\end{equation}
In what follows we will study $f $ and the "magnetization" $m = \partial f/\partial \alpha$. By $\eqref{twofs}$, the regularity properties of $f, m$ are the
same as for $f_{\tmop{spin}} {, m_{\tmop{spin}}} $, with $h=0$ mapped to $\alpha=1/2$.

For $T=0$ we have $B=0$. The tensor network partition function is then easy to compute: it equals $\alpha^N + (1 - \alpha)^N$ where $N$ is the number of sites. From here we obtain the $T=0$ free energy:\footnote{Note that substituting \eqref{fa} into \eqref{twofs} gives $f_{\tmop{spin}} = | h |$ which is the right answer for the spin model free energy at zero temperature.}
\begin{equation}
	f (\alpha,0) = \max (\ln (\alpha), \ln (1 - \alpha)). \label{fa}
\end{equation}
This is indeed a continuous function of $\alpha\in [0, 1]$, smooth away from $\alpha=1/2$. The magnetization is given by:
\begin{equation}\label{ma}
	m (\alpha,0) = \begin{cases} -1/(1-\alpha)& (0\le \alpha<1/2)\\
		1/\alpha & (1/2<\alpha\le1),
		\end{cases}
\end{equation} 
and it has a jump at $\alpha=1/2$. 
%

For nonzero $B$ we will of course not be able to compute $f(\alpha,B)$ explicitly, but we will derive a formula for it using a form of tensor RG, and we will use that formula to establish (dis)continuity properties of $f(\alpha,B)$. 

\subsection{RG for the $T=0$ free energy and magnetization}
To get an idea how this will work, let us first derive an RG formula for the $T=0$ free energy $f(\alpha,0)$. We consider the tensor RG map which is like the simple RG map from Section \ref{simpleRG} except that the new tensor $A'$ is defined by contracting a $4\times 4$ group of $A$ tensors, and not a $2\times 2$ group of tensors as in in {\eqref{Tex}}: 
\begin{equation}
\myinclude{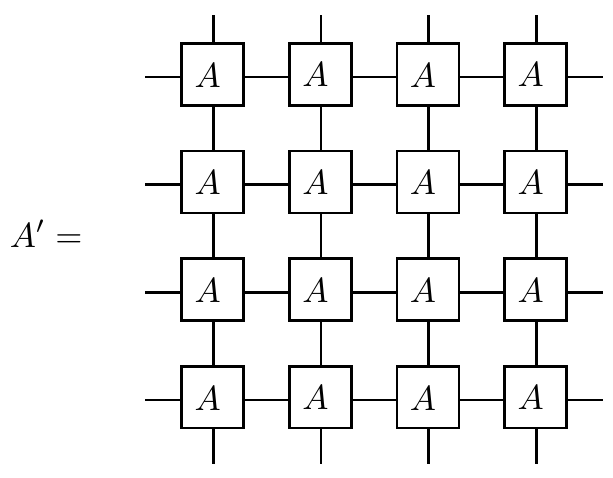} . \label{simple4x4}
\end{equation}
Thus this RG map has $b=4$.\footnote{$b = 2$ would also work here, but we will need $b = 4$ for nonzero $B$, so we do $b=4$ here to simplify the future comparison.}
We follow this RG map by reindexing, mapping the states $qqqq$ to $q$ ($q=\pm$). It's easy to see that the $T=0$ tensor 
\begin{equation}
	A = A(\alpha,0) = \alpha A^{(+)} + (1 - \alpha) A^{(-)} \,.
\end{equation}
is transformed by this RG map into
\begin{gather}
	A' = \alpha^{16} A^{(+)} + (1 - \alpha)^{16} A^{(-)} \equiv \normfactor(\alpha) 
	A(r(\alpha),0) \ ,
	\label{LT0map}\\
	\normfactor(\alpha)  :=  \alpha^{16} + (1 - \alpha)^{16} \,,\nn \\
	\boxed{r (\alpha)  :=  \frac{\alpha^{16}}{\alpha^{16} + (1 - \alpha)^{16}}}\ . \nn
\end{gather}
The map $\alpha\to r(\alpha)$ will play an important role in what follows. It has an unstable fixed point at $\alpha = 1 / 2$ (note that
$r' (1 / 2) = 16$), and two stable fixed points at $\alpha = 0, 1$. 

Consider now the RG trajectory $\alpha_0 = \alpha$, $\alpha_{j + 1} = r
(\alpha_j)$. By arguments similar to the proof of Prop. \ref{HTfree}, we have the following formulas for the $T=0$ free energy:
\begin{align}
\label{RGfT0a}	f (\alpha,0) &= 16^{-1} \ln \nu (r(\alpha))+16^{-1} f (r(\alpha),0)\,,
	\\ 
	\label{RGfT0} 	f (\alpha,0)&= \sum_{j = 1}^{\infty} 16^{- j} \ln \nu (\alpha_{j-1})\,. 
\end{align}
From here one gets an intuitive understanding of the discontinuity of magnetization at $\alpha=1/2$, relating it to the discontinuity of the RG trajectory. Namely, for the initial point on one or the other side of $\alpha=1/2$, the RG trajectory converges to one or the other stable fixed point. The trajectory jumps when crossing $\alpha=1/2$, although it varies continuously on both sides of this point.

A crucial aspect of Eq.~\eqref{RGfT0} is that the exponent $16$ in the geometric factor $16^{-j}$ exactly equals the eigenvalue $r'(1/2)$ by which $\alpha$ expands near the unstable fixed point. This is precisely the discontinuity fixed point condition of Nienhuis and Nauenberg \cite{NienhuisNauenberg}. If this condition were not satisfied, then the magnetization would not have a jump but a stronger or a weaker singularity. 

To see this, differentiate both sides of \eqref{RGfT0a} in $\alpha$, take the limits $\alpha\to 1/2^\pm$, and subtract the two expressions. We get for the discontinuity $\Delta m$ the equation:
\beq
\Delta m = 16^{-1} r'(1/2) \Delta m\,.
\eeq
Thus a nonzero finite discontinuity requires $16^{-1} r'(1/2) =1$. 

It is also instructive to see how the condition $16^{-1} r'(1/2) =1$ plays out in the expression \eqref{RGfT0}. If $\alpha=1/2\pm O(16^{-j_0})$, then the leading contribution to the \emph{derivative} of the r.h.s.~comes from the term with $j\sim j_0$, while terms with $|j-j_0|\gg  1$ are exponentially suppressed. This leading contribution is $O(1)$, independently of $j_0$. It is positive for $\alpha>1/2$ and negative for $\alpha<1/2$, ensuring a jump of magnetization.

Note that the argument of the previous paragraph only establishes a weak version of the jump, namely that 
\beq
\liminf_{\alpha\to 1/2^+} m(\alpha)  - \limsup_{\alpha\to 1/2^-} m(\alpha) >0\,.
\eeq
It does not seem obvious from Eqs.~\eqref{RGfT0a},\eqref{RGfT0} that the magnetization actually has \emph{continuous} limits as $\alpha\to1/2^\pm$. For this we have to appeal to the exact expression \eqref{ma}.

\subsection{A strategy for $T\ne0$}

To recap, at $T=0$ we have an exact formula \eqref{ma} which shows that the magnetization is discontinuous, and an RG formula \eqref{RGfT0} which is nontrivially compatible with the discontinuity due to the condition $r'(1/2)=16$. This suggests the following strategy to prove that the magnetization remains discontinuous for small nonzero $B$.
We will construct a tensor RG map with $b = 4$ which is a small perturbation of the $T = 0$ map \eqref{LT0map} in the $\alpha$ direction, and a contraction in the $B$ direction. More precisely we will have: 
\begin{align}
	\label{LTbasic}	& A = A(\alpha,B) \mapsto A'=\nu(\alpha,B) A(\alpha',B') \ ,
	\\
	& \normfactor  (\alpha, B) = \alpha^{16} + (1 - \alpha)^{16} + 
	O(\|B^2\|)\ ,\nn\\
	& \alpha'(\alpha,B) = r(\alpha) + O(\|B^2\|)\ ,\nn\\
	& \| B' (\alpha, B) \| < l \| B \|,\qquad l<1\ . \nn
\end{align}
Denoting $R(\alpha,B):=(\alpha',B')$, we will consider the RG trajectory $(\alpha_0,B_0)=(\alpha,B)$, $(\alpha_{j+1},B_{j+1})=R(\alpha_j,B_j)$. Then, using the stable manifold theorem, we will show that there exists a critical surface $\alpha=\alpha_c(B)$ such that RG trajectories starting on this surface converge to the low-T fixed point $\alpha=1/2,B=0$, while those starting with a larger or smaller value of $\alpha$ will converge towards the fixed points $\alpha=1,B=0$ and $\alpha=0,B=0$. 

Furthermore, we will have a formula for the free energy $f(\alpha,B)$ in terms of the RG trajectory:
\begin{equation}
	f (\alpha, B) = \sum_{j = 1}^{\infty} 16^{- j} \ln \nu (\alpha_{j-1}, B_{j-1})\ .
	\label{f0form}
\end{equation}
Because $\nu(\alpha,B)$ entering this formula is close to $\nu(\alpha)$ in Eq.~\eqref{RGfT0}, we will be able to compare the magnetization $m(\alpha,B)$ to $m(\alpha,0)$. We will show that $m(\alpha,B)$ has a jump on the stable manifold, of roughly the same size as the jump of $m(\alpha,0)$ at $\alpha=1/2$, which we know is nonzero from the explicit expression \eqref{ma}. 

The rest of the paper will be devoted to carrying out this strategy. In the next Section \ref{RG_lowT} we will construct an RG map with properties like \eqref{LTbasic}. 
Then, in Section \ref{LTprop}, we will derive the properties of the free energy from the existence of the RG map. There will be a minor technical difference with respect to the strategy described above. Namely, the RG map will be defined not for $\alpha\in [0,1]$ but for $\alpha\in (\delta,1-\delta)$ with a small $\delta>0$. This is not a big limitation since we know that for $\alpha$ close to 0,1 the tensor $A(\alpha,B)$ is close to one of the fixed points $A^{(\pm)}$ and the free energy is analytic there by the results of Section \ref{HighT}. Thus, instead of formula \eqref{f0form} with infinitely many terms we use a similar formula with $N$ terms plus an analytic remainder. This modification will be good enough for our purposes. In Section \ref{LTprop} we will show that the magnetization at small nonzero $B$ has a nonzero jump on the stable manifold, while it is analytic away from it. The Nienhuis-Nauenberg eigenvalue condition $r'(1/2)=16$ will play crucial role when showing the jump.

\section{RG map construction}
\label{RG_lowT}
In this section $\mathcal{H}$ will be the complex Hilbert space of tensors indexed by $\mathcal{I} = \{+,-\} \cup \mathbb{N}$, having finite HS norm, and $\HLT$ will be the complex Hilbert space of such tensors satisfying in addition the constraint:
\begin{equation}
	B_{qqqq}=0\quad(q=\pm)\ .
	\label{Breq}
\end{equation}

Continuing the line of reasoning from the previous section, we consider a tensor network built of tensors
\begin{equation}
	A = A(\alpha,B) := \alpha A^{(+)} + (1 - \alpha) A^{(-)} + B\,,\qquad B\in\HLT.
	\label{LTTN10}
\end{equation}
Recall that $A^{(q)}$ ($q=\pm$) are tensors having a single nonzero element $(A^{(q)})_{qqqq} =1$.

{\textbf{The following index conventions will be used extensively:}  
	\beq
	\begin{cases} 
		\text{$i,j,\ldots$ will vary over $\mathbb{N}$};\\
		\text{$q,q'$ will vary over $\{+,-\}$};\\
		\text{$I,J,\ldots$ will vary over the full $\mathcal{I}$.}
	\end{cases}\label{index}
	\eeq
	
	We will construct a tensor RG map $A\mapsto A'$ where
	\begin{equation}
		A' = \normfactor (\alpha, B) \cdot A(\alpha',B'),\qquad B'\in\HLT,
	\end{equation}
	where $\normfactor (\alpha, B)$ is a scalar (normalization factor). The function $R(\alpha,B):= (\alpha',B')$ will also be called the "RG map" (small abuse of terminology here).
	
	The lattice rescaling factor will be $b=4$. Since the partition function is preserved, we have (compare \eqref{RGproperty})
	\begin{equation}
		Z\Bigl(A(\alpha,B),L_x,L_y\Bigr) =\normfactor (\alpha,B)^{L_x L_y/b^2} Z\Bigl (A(\alpha',B'), {L_x}/{b},{L_y}/{b}\Bigr), \label{RGpropertyLT}
	\end{equation}
	
	Eventually we are interested in real $\alpha\in[0,1]$ and in real $B$, but in this section we will assume more generally that $\alpha$ and $B$ are complex. All functions will be analytic (see Appendix \ref{abstract}) in $\alpha$ and $B$, taking real values for $\alpha,B$ real. 
	
	We will assume that $\alpha$ belongs to the open rectangle 
	\begin{gather}
		\Pi=\{\alpha : \delta<\Re \alpha<1-\delta, |\Im \alpha|<w\}\label{rectangle}
		\\
		\myinclude{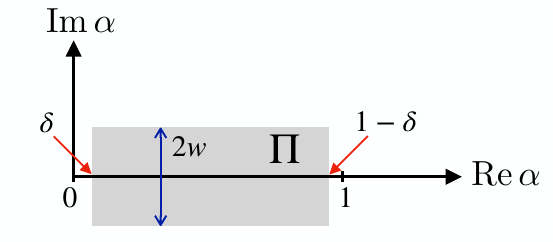}\nn
	\end{gather}
	The parameters $\delta,w>0$ will be fixed now and remain unchanged until the end of the paper:
	\begin{itemize}
		\item
		We choose and fix $w>0$ so that
		\beq
		\inf_{|\Im \alpha|< w} |\alpha^4+(1-\alpha)^4|>0,\qquad \inf_{|\Im \alpha|< w} |\alpha^{16}+(1-\alpha)^{16}|>0\ . \label{polybnd}
		\eeq
		Since both these polynomials do not vanish on the real axis, this can always be achieved by choosing $w>0$ sufficiently small. Numerically this requires $w< 0.049$.
		
		\item
		When $\alpha$ is close to $0$ or $1$ and $B$ is small, the tensor $A(\alpha,B)$ is close to $A^{(+)}$ or $A^{(-)}$. By results of Section \ref{HighT} (or standard cluster expansion methods), the infinite volume limit of the free energy
		then exists and is an analytic function. We choose and fix $\delta>0$ sufficiently small so that:
		\begin{gather}
			\label{deltachoice}
			\text{the free energy, as well as each of its derivatives, are analytic and bounded on  }\\
			\Omega = \{ (\alpha,B):  |\alpha| < 3 \delta \, \, \text{or} \, \, |1-\alpha| < 3 \delta, \|B\|<\delta \}\ .\nn
		\end{gather}
		The factors of $3$ in front of $\delta$
		are chosen to ensure an overlap with $\alpha\in \Pi$.
	\end{itemize}
	
	\begin{theorem}
		\label{LTmap} There exists an $\bar \epsilon = \bar \epsilon(\delta,w)>0$ such that for any $\epsilon\in (0,\bar \epsilon)$ there is a tensor RG map as described above, having rescaling factor $b=4$, defined on
		\begin{gather}
			\{(\alpha,B)\in \mathbb{C}\times \HLT:\quad \alpha\in \Pi ,\quad \|B\|<\varepsilon\}\ ,\label{domainLT}
		\end{gather}
		preserving the partition function as in \eqref{RGproperty}, and having the following
		additional properties: 
		\begin{subequations}\label{LTconds}
			\begin{align}
				&\text{$\normfactor$ and $R$ are analytic functions, taking real values for real $\alpha,B$};
				\label{LTreal}\\
				&\nu 
				=\alpha^{16}+(1-\alpha)^{16}+O(B^2);
				\label{LTnu}\\
				&\alpha' = r(\alpha)  +O(B^2);
				\label{LTalpha}\\
				& B' =\varepsilon^{1/2} O(B),\qquad  B'\in\HLT. \label{LTB}
			\end{align}
		\end{subequations}
	\end{theorem}
	
	\noindent{\bf Notation:} \label{OBn} $O(B^n)$ will denote an analytic function ($\mathbb{C}$ valued or $\mathcal{H}$ valued) of $\alpha,B$ in domain \eqref{domainLT} which is bounded by $C \| B\|^n$ where $C$ is a constant which {\bf does not depend on $\eps$}.
	
	\vspace{1em}\noindent
	{\bf{Proof:}} \label{startproof} As in Theorem \ref{HTmap}, the tensor RG map will be obtained by a composition of three steps:
	\begin{eqnarray}
		\text{Step 1 }(b=1):  &  & A(\alpha,B) \ \ \mapsto \nu_1(\alpha,B) A(\alpha_1,B_1)\ ,
		\nonumber\\
		\text{Step 2 }(b=2):  &  & A(\alpha_1,B_1) \mapsto \nu_2(\alpha_1,B_1) A(\alpha_2,B_2)\ ,
		\\
		\text{Step 3 }(b=2): &  & A(\alpha_2,B_2) \mapsto \nu_3(\alpha_2,B_2) A(\alpha_3,B_3)\ ,
		\nonumber
	\end{eqnarray}
	where we identify $R(\alpha,B)=(\alpha',B')=(\alpha_3,B_3)$ and
	\beq
	\nu(\alpha,B)= [\nu_1(\alpha,B) ]^{16}
	[\nu_2(\alpha_1,B_1)]^4 \nu_3(\alpha_2,B_2)\ .
	\label{nuprodLT}
	\eeq
	We will define the three steps in turn. The proof will conclude on p.~\pageref{endofproofTh61}.
	The following comments are useful to understand the inner workings of the argument:
	\begin{itemize}
		\item Condition $\delta <\Re \alpha < 1-\delta$ is used below Eqs.~\eqref{RhRv} and \eqref{rchoice}. See also Remark \ref{whydelta} below.
		\item Condition \eqref{polybnd} is used below Eqs.~\eqref{A'step2factor} as well as in Eqs.~\eqref{mudef},\eqref{alpha3bound}. 
		\item $\bar\eps$ is set almost at the end of the proof, below Eq.~\eqref{mudef}. The condition $\eps<\bar{\eps}$ is only used after that point.
		\item Condition \eqref{deltachoice} is not used at all here. In fact Theorem \ref{LTmap} holds for an arbitrary $\delta>0$. We prefer to choose and fix $\delta$ early on, so that the meaning of this parameter is clear to the reader. Condition \eqref{deltachoice} will come into play only in Section \ref{LTprop} below. 
	\end{itemize}

	\subsection{Step 1 - gauge transformation}
	\label{sec:LTgauge}
	
	Step 1 will have the properties
	\begin{subequations} \label{B2-B}
		\begin{align}
			&	\nu_1  = 1,\quad \alpha_1  =\alpha,\\ 
			&B_1=O(B), \label{B2-Bfirst} \\
			&(B_1)_{qqqI}\text{ (and rotations)}=O(B^2) \label{B2-Bsecond}\,.\footnotemark
		\end{align}
	\end{subequations} 
	\footnotetext{\label{restriction}This notation means that the whole tensor obtained by restricting $B_1$ to index combinations $qqqI$ with arbitrary $q,I$ is $O(B^2)$.}
	This step will be a gauge transformation (Section \ref{sec:gauge}). We have $b=1$ since gauge transformations do not change the number of tensors in the network. $A'$ is related to $A$ by Eq.~\eqref{gauge-ex}, in which
	we will choose $G_x$ and $G_y$ of the form
	\begin{equation}
		G_x = \exp(g_x),\qquad G_y=\exp(g_y),
	\end{equation}
	where tensors $g_x,g_y$ will be specified below in Eq.~\eqref{RhRv} and they will be $O(B)$. We thus have
	\begin{equation}
		G_x = I+g_x+O(B^2),\qquad G_y=I+g_y+O(B^2).
	\end{equation}
	We see that 
	\beq
	A'=A+g(A)+O(B^2),
	\label{GA}
	\eeq
	where $g(A)$ denotes the operator mapping $A$ to the sum of the four terms which are first-order in $g_x,g_y$:
	\begin{equation}
		\myinclude{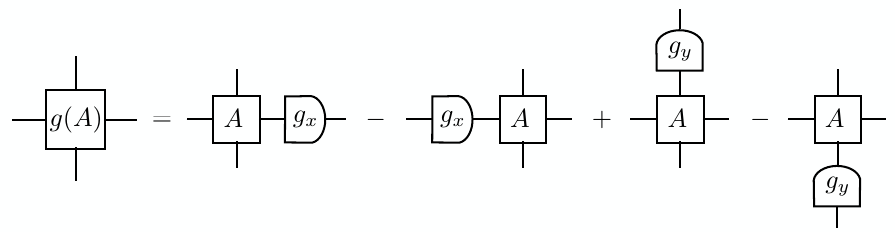}\ .
		\label{prop1-SR-2}
	\end{equation}
	
	We let $\alpha_+=\alpha$, $\alpha_-=1-\alpha$ and substitute $A=\sum_q \alpha_q A^{(q)}+B$
	into \eqref{GA} to get 
	\begin{gather}
		A'=\sum_q \alpha_q A^{(q)}+ B_1\ ,
		\label{GA2}\\
		B_1 =  B+ g\bigl(\sum_q \alpha_q  A^{(q)}\bigr)+O(B^2)\,.
		\label{dAp}
	\end{gather}
	
	Let us finally specify the tensors $g_x, g_y$. We set 
	\begin{align}
		\label{RhRv}		(g_x)_{q I}&= -\frac 1{\alpha_q} \myinclude{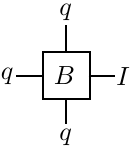}\,,\qquad
		(g_x)_{I q}= \frac 1{\alpha_q} \myinclude{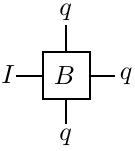}\\
		(g_y)_{q I} &= -\frac 1{\alpha_q} \myinclude{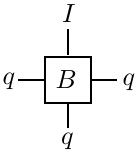}\,,\qquad
		(g_y)_{I q} = \frac 1{\alpha_q} \myinclude{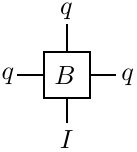}\nn
	\end{align}
	where we are using the index notation \eqref{index}. All other components of $g_x,g_y$ are zero.
	Since $|\alpha_{q}^{-1}|\le 1/\delta$ ($\alpha\in \Pi $), we have that $g_x,g_y=O(B)$.
	
	\begin{remark} 	\label{whydelta}The appearance of $\alpha_q^{-1}$ factors in Eq.~\eqref{RhRv}, and of similar $\alpha_q^{-2}$ factors in Eq.~\eqref{rchoice} below, is the technical reason why we cannot allow $\alpha$ to get close to $0,1$.
	\end{remark}
	
	We can now see that $B_1$ in \eqref{dAp} satisfies the requirements in \eqref{B2-B}:
	\begin{itemize}
		\item
		it is $O(B)$ since all terms in \eqref{dAp} are such;
		\item
		by the choice of $g_x,g_y$, the $O(B)$ part of the $B_1$ components $qqqI$ (and rotations) cancels exactly between $g\bigl(\sum_q \alpha_q  A^{(q)}\bigr)$ and $B$. So those components are $O(B^2)$ as needed.
	\end{itemize}
	
	\begin{remark} If the perturbation $B$ preserves parity, then $g_x,g_y$ are antisymmetric matrices and $G_x,G_y$ are orthogonal. 
	\end{remark}
	
	\begin{remark} In general, $B_1$ does not satisfy $(B_1)_{qqqq}=0$. This is unlike in the construction near the high-T fixed point \cite{paper1}, where we restored this condition after each of the three smaller steps comprising the full tensor RG step. Here instead we will restore this condition only at the very end. This eliminates some unnecessary steps from the argument.
	\end{remark}

		\subsection{Step 2 - simple RG step}
		
		The quantities $\nu_2, \alpha_2, B_2$ defining Step 2 will be analytic in $\alpha_1, B_1$ and hence in $\alpha, B$. They will satisfy a list of conditions, which we state in terms of $\alpha,B$:
		\begin{subequations}\label{Simple-B}
			\begin{align}
				&\nu_2={\alpha^4+(1-\alpha)^4},\qquad  
				\alpha_2 = \frac{\alpha^4}{\alpha^4+(1-\alpha)^4},
				\label{Simple-a2}\\
				& (B_2)_{qqij} \text{ (and rotations)}=O(B) \text{ for } q\in\{\pm\},
				i,j\in\mathbb{N}\,. \label{Simple-B-O(B)} \footnotemark\\
				&	(B_2)_{\text{all other components}}=O(B^2) . \label{Simple-B-O(B2)}
			\end{align}
		\end{subequations}
		\footnotetext{Like in footnote \ref{restriction}, this condition and the next refer to the whole tensor obtained by restricting $B_2$ to the mentioned combinations of indices, not to individual components which would be a weaker statement.}
		The point of Step 2 is that $B_2$ compared to $B_1$ has more of its components which are $O(B^2)$. 
		
		Step 2 will be a simple RG transformation followed by leg grouping and reindexing (see Sections \ref{simpleRG} and \ref{LGR}): 
		\begin{equation}
			\myinclude{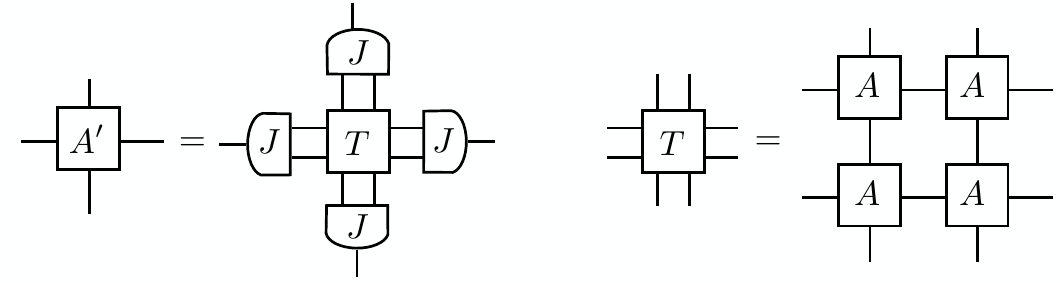}
			\label{simpleRGstep}
		\end{equation}
		The reindexing map $J$ will satisfy
		\beq
		J(+,+)=+,\qquad J(-,-)=-,
		\label{condJ}
		\eeq
		while it will map the set of all other index pairs $(I_1,I_2)$ onto $\mathbb{N}$ in an arbitrary order.
		
		Denote
		\begin{equation}
			A^{(q)}=\myinclude{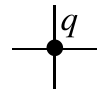}\,.
		\end{equation}
		We substitute $A(\alpha_1,B_1)$ into the definition of $T$ and expand the result in $B_1$.  The only zeroth-order terms in $T$ are
		\begin{equation}
			\myinclude{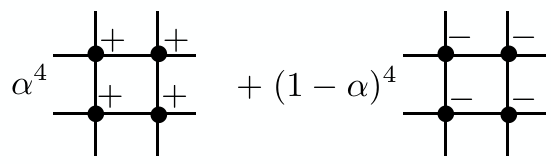}
			\label{0term}
		\end{equation}
		This is because whenever $A^{(+)}$ and $A^{(-)}$ share a leg the contraction
		is zero. After reindexing by $J$, tensor \eqref{0term} maps to 
		\begin{equation}
			\alpha^4 A^{(+)}+(1-\alpha)^4 A^{(-)}.
			\label{0term1}
		\end{equation}
		We thus have
		\beq
		A'=\alpha^4 A^{(+)} + (1-\alpha)^4 A^{(-)} + \tilde{B},
		\label{A'step2}
		\eeq
		where $\tilde{B}=O(B)$ comes from terms in $T$ linear or higher-order in $B_1$. 
		
		We claim that $\tilde{B}$ satisfies \eqref{Simple-B-O(B)},\eqref{Simple-B-O(B2)}. It suffices to check this for the terms first-order in $B_1$. These are: 
		\begin{equation}
			\myinclude{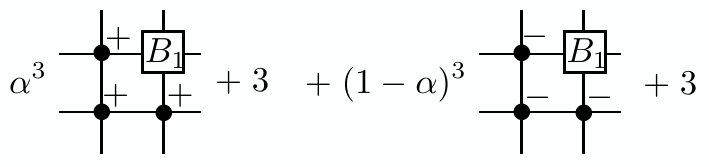}\ \ ,
			\label{fig-one-circle}
		\end{equation}
		where "$+3$" stands for the 3 other similar diagrams involving a single $B_1$. The nonzero contractions here involve components of $B_1$ with two internal lines set to the corresponding $q=\pm$:
		\begin{equation}
			\myinclude{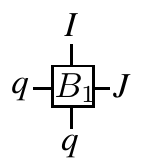}\ \ .
		\end{equation}
		If one or both of the indices $I,J$ is $q$, then by \eqref{B2-B} these components are $O(B^2)$, so not problematic. Let us examine then what happens for $I,J\ne q$, when we have a diagram $O(B)$:
		\begin{equation}
			\myinclude{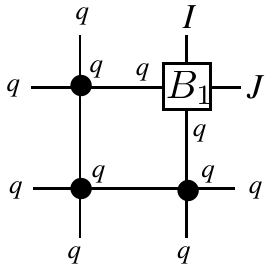}\qquad(I,J\ne q)
		\end{equation}
		After reindexing, this diagram gives components of the form $ijqq$ with $i,j\in \mathbb{N}$, i.e.~precisely the index requirements in \eqref{Simple-B-O(B)}. This completes the check that $\tilde{B}$ satisfies \eqref{Simple-B-O(B)},\eqref{Simple-B-O(B2)}. 
		
		Finally, we factor out $\nu_2=\alpha^4+(1-\alpha)^4$ in Eq.~\eqref{A'step2}, writing
		\beq
		A'=\nu_2\cdot A(\alpha_2,B_2),\qquad \alpha_2= \nu_2^{-1} \alpha,\quad B_2=\nu_2^{-1} \tilde{B}\ .
		\label{A'step2factor}
		\eeq
		By Eq.~\eqref{polybnd}$, \nu_2^{-1}$ is bounded for $\alpha\in \Pi $. Therefore $B_2$
		satisfies \eqref{Simple-B-O(B)},\eqref{Simple-B-O(B2)} since $\tilde B$ satisfies them.
		All the conditions \eqref{Simple-B} have been verified.

		
		
		\subsection{Step 3 - main RG step }
		\label{sec:mainLTRG}
		
		This step will be analogous to Step II of the RG map in \cite{paper1}. As there, it will be a composition of several substeps:
		\begin{enumerate}
			\item[3.1] (Simple RG) Form tensor $T$ by contracting four $A_2$ tensors as in Eq.~\eqref{Tex}.
			
			\item[3.2] (Disentangling 1) Act on the horizontal indices of $T$ by disentanglers $R$, $R^{-1}$. As discussed in Section \ref{sec:disentanglers}, a disentangler is a bounded invertible operator $R$ which acts on two indices. We
			insert $R R^{-1}=I$ into the two horizontal lines between $T$'s. We then combine $T$ with the $R^{-1}$ on its
			left and the $R$ on its right, and call the resulting tensor $S$: 
			\begin{equation}
				\myinclude{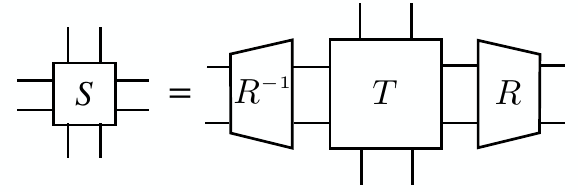}\,. 
				\label{def_S}
			\end{equation}
			
			\item[3.3] (Disentangling 2) Represent the tensor $S$ as a contraction of two tensors $S^u$ and $S^d$:
			\beq
			\myinclude{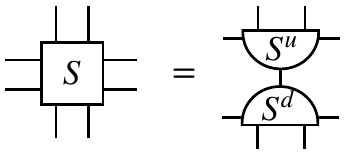} \,.
			\label{II3}
			\eeq
			\item[3.4] (Reconnection) Define the tensor $U$ by contracting tensors $S^u$ and $S^d$ in the opposite order, and then define $A'$ via leg grouping and reindexing on the horizontal pairs of legs of $U$, with a reindexing map $J$ satisfying \eqref{condJ}:
			\beq
			\myinclude{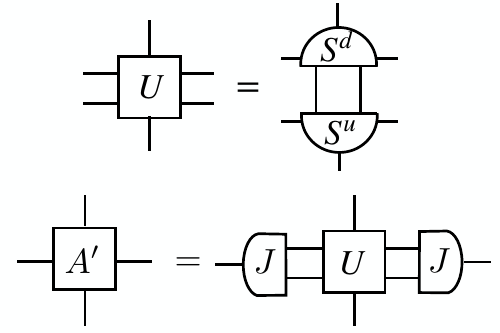} \label{Udef}
			\eeq
			
		\end{enumerate}
		\begin{remark}
			Disentangling 1\&2 are like steps D1,D2 from Section \ref{sec:disentanglers}, with $X=Y$ equal to the horizontal contraction of two out of four $A$ tensors of the diagram defining tensor $T$, $X'=S^u$ and $Y'=S^d$. In Reconnection, rewriting the network in terms of $U$ is an example of a simple RG transformation.
		\end{remark}
		
		\begin{remark}
			The first rigorous use of disentanglers was in the study of tensor network RG in the high temperature phase \cite{paper1}. Section 2.2.1 of \cite{paper1} shows the necessity of introducing them in this setting. The importance of introducing disentanglers was recognized early in the numerical studies of tensor network RG. See appendix C in \cite{paper1} for a review.
		\end{remark}
		
		We will show that with an appropriate choice of $R$, $S^u$, $S^d$, we will have $A'=\nu_3\cdot A(\alpha_3,B_3)$ where $\nu_3, \alpha_3, B_3$ obey the following conditions in terms of $\alpha,B$:
		\begin{subequations}\label{step3}
			\begin{align}
				&\nu_3=\frac{\alpha^{16}+(1-\alpha)^{16}}{[\alpha^4+(1-\alpha)^4]^4}+ O(B^2),\label{nu3}\\
				&\alpha_3 = r(\alpha)+ O(B^2) ,
				\label{alpha3}\\
				& B_3 =\varepsilon^{1/2} O(B), \label{step3bound}\\
				&	(B_3)_{qqqq}=0\quad \text{i.e.}\ B_3\in\HLT\label{step3HLT}.
			\end{align}
		\end{subequations}

		{\it Step 3.1.} 
		We substitute $A=A(\alpha_2,B_2)$ into the definition of $T$, and expand in $B_2$. 
		It will be convenient to write, in agreement with \eqref{Simple-B-O(B)},\eqref{Simple-B-O(B2)},
		\begin{align}
			&B_2=\myinclude{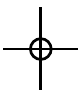}+\myinclude{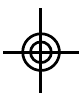},
			\qquad \myinclude{fig-B21.pdf}=O(B),
			\qquad \myinclude{fig-B22.pdf}=O(B^2) \\
			& \myinclude{fig-B21.pdf} \text{  has only nonzero components }
			\myinclude{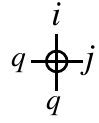}\text{ (and rotations) with } q\in\pm,i,j\in \mathbb{N}.
			\label{sc-conds}
		\end{align}
		Then the first few terms in the expansion of $T$ are:
		\begin{align}
			T & = T_0 + T_1 + T_2 + O(B^3), \label{Texp}\\
			T_0 & =\sum_q \alpha_q^4 \myinclude{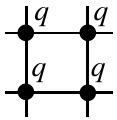}\ \qquad(\alpha_+=\alpha_2, \alpha_-=1-\alpha_2),
			\label{T0}\\
			T_1 & = \sum_{q} \alpha_q^3 \myinclude{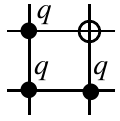}+3\ ,\\
			T_2 & 
			= \sum_{q} \alpha_q^3 
			\myinclude{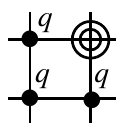}+3 \ ,
			\\
			&+  \sum_{q} \alpha_q^2 
			\Biggl(\myinclude{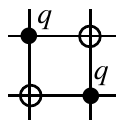}+1 \quad 
			+ \myinclude{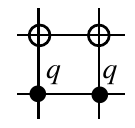}+1\quad 
			+ \myinclude{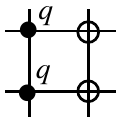}+1 \Biggr)\ .\nn
		\end{align}
		Other potential $T_2$ diagrams vanish identically by \eqref{sc-conds}, such as e.g.
		$
		\myinclude{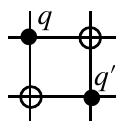}
		$ for $q\ne q'$. As in \eqref{fig-one-circle}, the notation "$+n$" after a diagram refers to the other $n$ diagrams related by obvious lattice symmetries, which we don't write explicitly. 
		
		We further decompose the last diagram of $T_2$ into two terms according to whether the contracted index between the two single-circle tensors equals $q$ or not:
		\begin{equation}
			\myinclude{fig-T2d.pdf} = 		 \myinclude{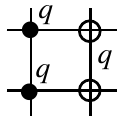}\ +\ \myinclude{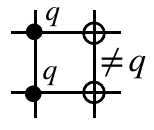}\ \ .
		\end{equation}
		The last diagram here is the "dangerous diagram". It's the only diagram in $T_0+T_1+T_2$ which doesn't have a pair of identical $q$ indices on the vertical contracted bonds. Using \eqref{sc-conds} the non-zero components of the dangerous diagram are 
		\begin{equation}
			\myinclude{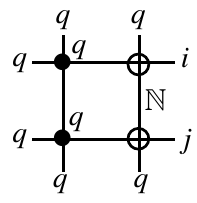},\qquad i,j \in \mathbb{N}\ \ .
		\end{equation}
		
		{\it Step 3.2.} The goal of this substep will be to remove the dangerous diagram via a disentangling transformation. Removing this diagram is an example of reducing the "entanglement" between the upper and lower groups of indices of $T_0+T_1+T_2$, in the sense of Remark \ref{disent-origine}. Why this is useful will become clear in Step 3.4.
		
		We will use the general disentangler theory from Section \ref{sec:disentanglers}. Define $S=R^{-1}T R$ by Eq.~\eqref{def_S}. The disentangler $R$ will be chosen as $R=\exp(\rho)$ with $\rho$ having the following nonzero components:
		\beq
		\myinclude{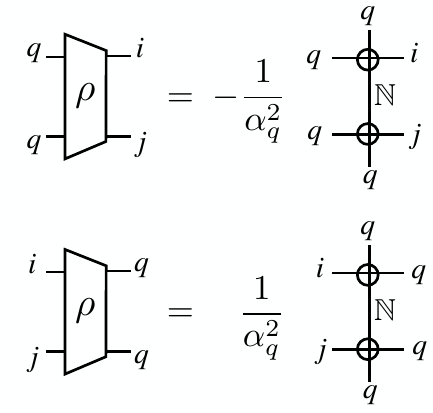}
		\label{rchoice}
		\eeq
		We will see that this is the right choice to remove the dangerous diagram.
		
		Splitting the diagrams in \eqref{rchoice} in their upper and lower halves, we see that $\rho$ can be written as in Eq.~\eqref{r}, i.e.~as a contraction of two tensors $\rho_u,\rho_d$, which are both $O(B)$.\footnote{Note that this requires doubling of the vertical index space; see the proof of Lemma 2.1 in \cite{paper1} (arXiv version).}   Here we are using that the prefactors in \eqref{rchoice} are bounded: $|1/\alpha^2_q|<1/\delta^2$. These prefactors should be split somehow between $\rho_u$ and $\rho_d$, e.g.~as $\pm 1/\alpha_q^2=(\pm 1/\alpha_q)\times (1/\alpha_q)$.
		
		Since $\rho$ is a contraction of $\rho_u$ and $\rho_d$, we are in a position to apply Lemma \ref{Rfromr}. By that lemma, we can represent $R$, $R^{-1}$ in factorized form (see Eqs.~\eqref{Rfact},\eqref{LemmaR}): 
		\begin{equation}
			\myinclude{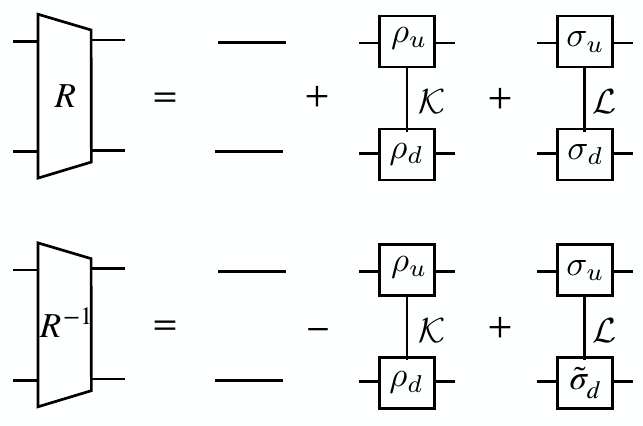},
			\label{LemmaR1}
		\end{equation}
		with $\sigma_u,\sigma_d,\tilde \sigma_d=O(B^2)$. 
		
		We substitute Eqs.~\eqref{Texp} and \eqref{LemmaR1} into the definition of $S$. We have
		\beq
		S=T_0+T_1+T_2+(T_0 \rho - \rho T_0)+O(B^3),
		\label{SOB3}
		\eeq
		By \eqref{rchoice}, the term $T_0 \rho - \rho T_0$ exactly cancels the dangerous diagram terms in $T_2$:
		\beq
		\sum_q \alpha_q^2\left( \myinclude{fig-dang1.pdf}\ +\ \myinclude{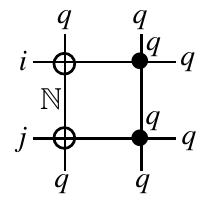}\right)
		\eeq
		So the dangerous diagram is removed as promised.
		
		{\it Step 3.3.}
		We have to decompose $S$ as a contraction of $S^{u}$ and $S^{d}$. We will first define the leading parts of these two tensors, $\hat{S}^u$ and $\hat{S}^d$. The contraction of $\hat{S}^u$ and $\hat{S}^d$ will be equal to $S$ up to error terms $O(B^3)$. We will see later how to correct $\hat{S}^u$ and $\hat{S}^d$ so that the full $S$ is obtained.
		
		Tensors $\hat{S}^u$ and $\hat{S}^d$are defined as follows:
		\begin{equation}
			\label{Su2def}
			\myinclude{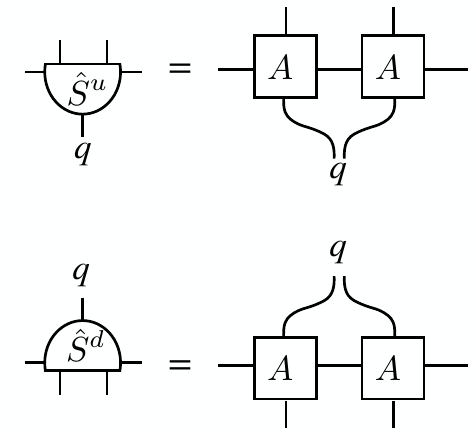}
		\end{equation}
		Thus the only non-zero components of $\hat{S}^u$ (resp.~$\hat{S}^d$) have bottom (resp.~top) index $q=\pm$. According to this equation, $\hat{S}^u$ with bottom index $q$ is obtained by contracting of two copies of $A=A(\alpha_2,B_2)$ and setting both bottom indices of this contraction equal to $q$. Similarly for $\hat{S}^d$.
		
		
		Denote by $\hat{S}$ the contraction of $\hat{S}^u$ and $\hat{S}^d$ with the former on top
		and the latter on the bottom. $\hat{S}$ contains a great many diagrams from $S$, Eq.~\eqref{SOB3}, namely:
		\begin{itemize}
			\item
			all diagrams in $T_0+T_1$
			\item all diagrams in $T_2$ \emph{except} the dangerous diagram which is, by Step 3.2, all diagrams in $T_2+(T_0 \rho - \rho T_0)$.
		\end{itemize}
		$\hat{S}$ also contains a part of the $O(B^3)$ error term in \eqref{SOB3}.
		In particular, it contains the diagrams of $T$ which have the upper or lower row occupied by two $A^{(q)}$'s, i.e.~of the following form:
		\beq
		\label{figB2B2}
		\myinclude{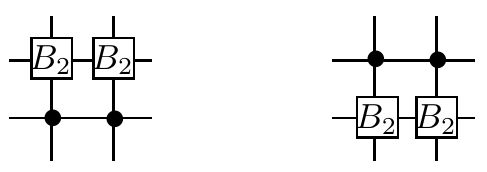}
		\eeq
		Here $B_2$ stands for the whole of $B_2$, i.e.~the sum of the single-circle and the double-circle parts, see Eq.~\eqref{sc-conds}.
		
		However, some $O(B^3)$ diagrams in $S$ are \emph{not} contained in $\hat{S}$. These are:
		\begin{itemize}
			\item
			diagrams from $T$ which have at least one $B_2$ in the upper row and in the lower row, such as
			\beq
			\myinclude{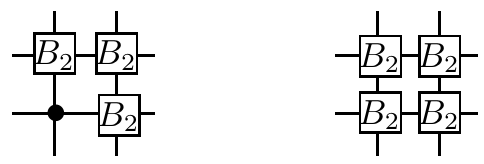}
			\eeq
			More precisely, we are talking only about the part of these diagrams where not both vertical indices are set to $q$.
			\item 
			diagrams obtained by multiplying diagrams from $T$ by the first- and second-order pieces of the expansions of $R$ and $R^{-1}$ from Eq.~\eqref{LemmaR1}. (We exclude the diagram $T_0\rho-\rho T_0$, which canceled the dangerous diagram.) Examples of such diagrams are 
			\beq
			\myinclude{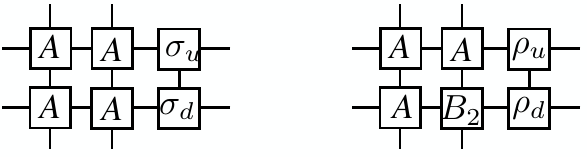}
			\eeq
		\end{itemize}
		Let us number all these "missing" diagrams in an arbitrary order, as $D_\alpha$, $\alpha=1,\ldots,N$.
		We will need below their properties:
		\begin{itemize}
			\item $\text{the contraction of }S^u_{0}\text{ and }S^d_{0}\text{ equals }S-\sum_{\alpha=1}^N D^\alpha$;
			\item each $D^\alpha=O(B^3)$;
			\item each $D^\alpha$ naturally comes as a contraction of its upper part $D^{u,\alpha}$ and its lower part $D^{d,\alpha}$, with one of the following three possibilities:
			\begin{enumerate}
				\item
				$D^{u,\alpha}=O(B)$, $D^{d,\alpha}=O(B^2)$;
				\item
				$D^{u,\alpha}=O(B^2)$, $D^{d,\alpha}=O(B)$;
				\item
				$D^{u,\alpha}=O(B^2)$, $D^{d,\alpha}=O(B^2)$. 
			\end{enumerate}
			Note that $D^{u,\alpha}=O(1)$, $D^{d,\alpha}=O(B^3)$ and the reverse are excluded, because such $O(B^3)$ diagrams are contained in $\hat{S}$, see Eq.~\eqref{figB2B2}.
		\end{itemize}
		
		Let us now complete the definition of $S^u$ and $S^d$ .
		The index set for the bottom leg in $S^u$ and the top leg in $S^d$ 
		will be the usual $\mathcal{I} = \{+,-\} \cup \mathbb{N}$.
		The components of $S^u$ when the index on the bottom leg is $+$ or $-$ are equal to the corresponding components
		in $\hat{S}^u$ with a similar statement for $S^d$. 
		
		We will now define $S^u$ when the index on the bottom leg is in $\mathbb{N}$
		and $S^d$ when the index on the top leg is in $\mathbb{N}$ in such a way that the contraction of $S^u$ and
		$S^d$ over $\mathbb{N}$ will be equal the sum of the missing diagrams. 
		
		As we said, each missing diagram $D^\alpha$ naturally comes as a contraction of its upper part $D^{u,\alpha}$ and its lower part $D^{d,\alpha}$
		where the contraction is over some infinite index set. 
		We now partition the index set $\mathbb{N}$ into $N$ infinite, disjoint sets:
		$\mathbb{N} = \sqcup_{\alpha=1}^N \, \mathcal{I}_\alpha$.
		Since $\mathcal{I}_\alpha$ is infinite, we can
		use it for the contraction of $D^{u,\alpha}$ and  $D^{d,\alpha}$.
		Note that $D^{u,\alpha}$, respectively $D^{d,\alpha}$, vanishes if the index on the bottom, respectively top,
		leg is not in $\mathcal{I}_\alpha$. 
		
		As mentioned above, some of $D^{u,\alpha}$ and $D^{d,\alpha}$ are $O(B)$. We need to improve these terms by a factor of $\epsilon^{1/2}$. 
		We will define
		\beq
		S^{u,\alpha}=k_\alpha D^{u,\alpha},\quad S^{d,\alpha}=k_\alpha ^{-1}D_{d,\alpha}\ .
		\label{kalpha}
		\eeq
		with an appropriate $k_\alpha>0$. Whatever $k_\alpha$ is, the contraction of $S^{u,\alpha}$ and $S^{d,\alpha}$ gives $D^\alpha$.
		
		According to the three possibilities mentioned above, we set:
		\begin{itemize}
			\item
			If both $D^{u,\alpha}, D^{d,\alpha}=O(B^2)$, then we set $k_\alpha=1$;
			\item
			If $D^{u,\alpha}=O(B)$, $D^{d,\alpha}=O(B^2)$, then we set $k_\alpha=\varepsilon^{1/2}$;
			where $\varepsilon$ is from Eq.~\eqref{domainLT}. 
			\item Finally, if $D^{u,\alpha}=O(B^2)$, $D^{d,\alpha}=O(B)$, then we set $k_\alpha=\varepsilon^{-1/2}$.
		\end{itemize}
		
		We now define 
		\begin{gather}
			S^u= \hat{S}^{u}\oplus\Delta S^u,\quad  \Delta S^u =\bigoplus_{\alpha}
			S^{u,\alpha},\nn \\
			S^d= \hat{S}^{d} \oplus \Delta S^d, \quad\Delta S^d=\bigoplus_{\alpha} S^{d,\alpha}.
			\label{SuSd}
		\end{gather}
		The crucial point in this construction is that when we contract $S^u$ and $S^d$ there are no ``cross terms.''
		$\hat{S}^u$ can only contract with $\hat{S}^d$ and $S^{u,\alpha}$ can only contract with $S^{d,\alpha}$
		because the $\mathcal{I}_\alpha$ are disjoint from $\{+,-\}$ and from each other. 
		(The direct sum notation in Eq.~\eqref{SuSd} is meant
		to emphasize the disjointness of the various index sets.)
		By the choice of $k_\alpha$ in \eqref{kalpha}, we have
		\beq
		\Delta S^u, \Delta S^d =O(B^2)+\varepsilon^{1/2} O(B) +\varepsilon^{-1/2} O(B^2) =\varepsilon^{1/2} O(B)\ ,
		\label{Sudbound}
		\eeq
		where in the last equality we used $\| B \|<\varepsilon$.  
		

		The factor $\eps^{1/2}$ is very important. In Step 3.4 it will lead to Eq.~\eqref{step3bound}, which implies that the RG map contracts in the $B$ direction for sufficiently small $\eps$. Note that we would not get the $\eps^{1/2}$ factor without first removing the dangerous diagram, which is $O(B^2)$ and would require $\Delta S^u, \Delta S^d=O(B)$. This shows the importance of disentangling.
		
		\begin{remark}\label{modification}
			It is instructive to compare this step with the proof of Proposition 2.3 in \cite{paper1}. There is a minor difference in the definition of tensors $\hat{S}^u$ and $\hat{S}^d$. Translating Eq.~(2.46) of \cite{paper1} to the present situation, this would correspond to taking instead of $\hat{S}^u$, $\hat{S}^d$ the tensor $\tilde{S}^u$ defined by 
			\beq
			\myinclude{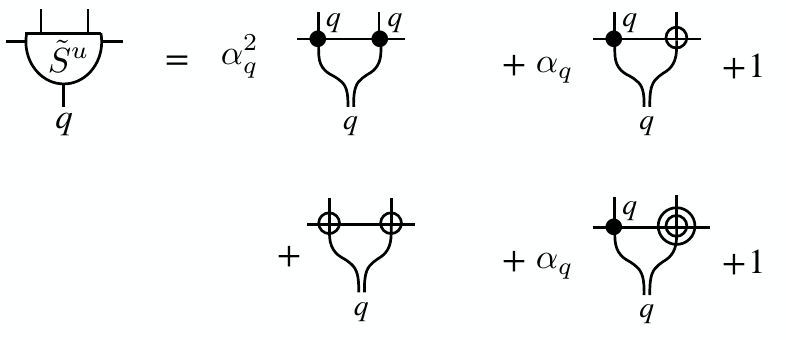}
			\eeq
			and the tensor $\tilde{S}^d$ defined by flipping this equation upside down.
			
			We can see that $\tilde{S}^u$, $\tilde{S}^d$ are nothing but the expansions of our $\hat{S}^u$, $\hat{S}^d$ up to $O(B^2)$. $\tilde{S}^u$ and $\tilde{S}^d$ share the first good property of $\hat{S}^u$, $\hat{S}^d$, namely that their contraction reproduces all diagrams $T_0+T_1+T_2$ except the dangerous diagram. However, they do not share the second good property: their contraction does not contain diagrams \eqref{figB2B2}. Because of this, with $\tilde{S}^u$ and $\tilde{S}^d$ we would not be able to get $\Delta S^u,\Delta S^d=\eps^{1/2} O(B)$ whose contraction reproduces all missing diagrams. We could only get $\Delta S^u,\Delta S^d=O(\eps^{3/2})$, which was sufficient for the purposes of \cite{paper1} but not here.
			
			In the proof of Theorem \ref{HTmap} we mentioned that a minor modification is needed in the proof of Proposition 2.3 in \cite{paper1} to get Eq.~\eqref{B3}. This modification consists precisely in replacing Eq.~(2.46) of \cite{paper1} by Eq.~\eqref{Su2def}. 
		\end{remark}
		
		\textit{Step 3.4.} The reconnected tensor $U$ defined in \eqref{Udef} can be expanded as 
		\begin{align}
			&U=U_0+\Delta U \\
			&U_0 = \myinclude{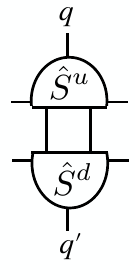}\quad (q,q'=\pm,\text{ may be equal or not}) \\
			&\Delta U = \myinclude{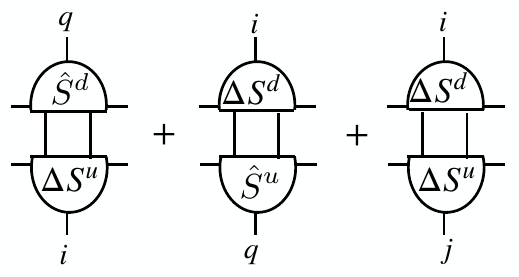}
		\end{align}
		For $\Delta U$ the estimate $\Delta U = \varepsilon^{1/2} O( B)$ by \eqref{Sudbound} will suffice. Let us examine $U_0$ more closely. Substituting $\hat{S}^u$ and $\hat{S}^d$ from Eq.~\eqref{Su2def} into the definition of $U_0$, we have
		\beq
		U_0=\sum_q \alpha_q^4 \myinclude{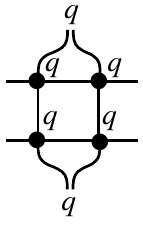} + O(B^2)\label{U0corr}
		\eeq
		To see that the correction term is indeed $O(B^2)$, note that the potential $O(B)$ terms vanish by Eq.~\eqref{sc-conds}, such as
		\beq
		\myinclude{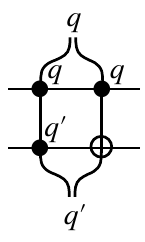} 
		\eeq
		
		Finally, we define $A'$ via reindexing by $J$ on the horizontal pairs of legs of $U$, Eq.~\eqref{Udef}. The leading term in \eqref{U0corr} maps to $\sum_q \alpha_q^4\, A^{(q)}$. Let $\tilde B$ be the sum of $\Delta U$ and the correction term in \eqref{U0corr} after reindexing. We have
		\begin{align}
			&A'=\sum_q \alpha_q^4\, A^{(q)}+ \tilde B, \label{A'step30}
			\\
			&\tilde B = \varepsilon^{1/2} O(B) + O(B^2) = \varepsilon^{1/2} O(B)\,.
			\label{tildeBbound}
			\\
			&\tilde B_{qqqq}=O(B^2)\,.\label{B2qqqq}
		\end{align}
		For the last estimate, note that $\tilde B_{qqqq}$ components do not come from $\Delta U$ but only from the correction term in \eqref{U0corr}.
		
		We are almost done, since $\tilde B$ satisfies the bound \eqref{step3bound} we are aiming for. We only need to rewrite $A'$ as $\nu_3\cdot A(\alpha_3,B_3)$ and arrange that $B_3$ has zero $qqqq$ components.   
		To this end, let us remove the $qqqq$ components from $\tilde B$ and reabsorb them into the leading term in \eqref{A'step30}, adjusting the coefficients of $A^{(q)}$ accordingly. We obtain:
		\begin{gather}
			A'=\sum_q \beta_q A^{(q)}+\hat{B},\qquad \beta_q=\alpha_q^4+\tilde B_{qqqq}\,,
			\label{A'step3}
		\end{gather}
		where tensor $\hat{B}$ equals $\tilde B$ with the $qqqq$ components set to zero. 
		
		We then define
		\begin{equation}
			\nu_3=\beta_+ +\beta_-\ .
		\end{equation}
		Using \eqref{B2qqqq}, $\alpha_+=\alpha_2, \alpha_-=1-\alpha_2$, $\alpha_2=\frac{\alpha^4}{\alpha^4+(1-\alpha)^4}$, we obtain
		\beq
		\nu_3 = \alpha_+^4+\alpha_-^4+O(B^2)= \frac{\alpha^{16}+(1-\alpha)^{16}}{[\alpha^4+(1-\alpha)^4]^4}+O(B^2),\label{nu31}
		\eeq
		i.e.~precisely \eqref{nu3}. 
		
		Factoring out $\nu_3$ from \eqref{A'step3}, we write $A'$ in \eqref{A'step3} as
		\begin{align}
			&A'=\nu_3 \cdot A(\alpha_3,B_3)\ ,\\
			& \alpha_3 = \nu_3^{-1} \beta_+\ ,\\
			& B_3 = \nu_3^{-1} \hat{B}\ . \label{B3def}
		\end{align}
		
		We have to make sure that $\nu_3^{-1}$ is regular on our domain \eqref{domainLT}. This will be the case if the leading term in \eqref{nu31} is nonzero and if the correction term does not overwhelm the leading term. The leading term is nonzero by our initial assumption Eq.~\eqref{polybnd}. The correction term will not overwhelm the leading term if we choose $\bar{\eps}$ sufficiently small. {\bf (That's where $\bar\eps$ is fixed.)} 
		
		Specifically, let
		\beq
		\mu = \inf_{\alpha\in \Pi } \left|\frac{\alpha^{16}+(1-\alpha)^{16}}{[\alpha^4+(1-\alpha)^4]^4}\right|,
		\label{mudef}
		\eeq
		which is positive by Eq.~\eqref{polybnd}. Choose $\bar{\varepsilon}>0$ so small that $ C \bar{\varepsilon}^{2}<\mu/2$ where $C$ is the constant in the $O(B^2)$ in \eqref{nu31}. 
		
		Then $\nu_3^{-1}$ is bounded and analytic for $\alpha\in \Pi $ and $\|B\|<\eps<\bar\eps$, specifically
		\begin{align}
			&|\nu_3^{-1}|<2/\mu,\\ 
			&\nu_3^{-1}=\frac{[\alpha^4+(1-\alpha)^4]^4}{\alpha^{16}+(1-\alpha)^{16}}
			+O(B^2). \label{nu32}
		\end{align}
		
		We can now see that $\alpha_3$ and $B_3$ satisfy \eqref{alpha3},\eqref{step3bound}. For $B_3$ this is obvious. For $\alpha_3$ we have
		\begin{align}
			\alpha_3 =\left( \frac{[\alpha^4+(1-\alpha)^4]^4}{\alpha^{16}+(1-\alpha)^{16}} + O(B^2)\right)\left(\frac{\alpha^{16}}{[\alpha^{4}+(1-\alpha)^4]^4}+O(B^2)\right).
			\label{alpha3bound}
		\end{align}
		Taking into the account Eq.~\eqref{polybnd}, this implies \eqref{alpha3}. All conditions \eqref{step3} are verified.
		
		\begin{proof}[End of proof of Theorem \ref{LTmap} started on p.~\pageref{startproof}]
			As stated there we have $\alpha'=\alpha_3$, $B'=B_3$. Conditions \eqref{LTalpha} and \eqref{LTB} coincide with conditions satisfied by $\alpha_3,B_3$ at the end of Step 3. Condition \eqref{LTnu} follows from \eqref{nuprodLT} using the conditions satisfied by $\nu_1,\nu_2,\nu_3$. 
			\label{endofproofTh61}
		\end{proof}
	
\section{Properties of free energy and magnetization at low $T$}
\label{LTprop}

\subsection{Properties of the RG map} 
In this subsection we use the properties of the RG map established in
Theorem \ref{LTmap} to prove properties of the derivative of the RG map.
Later on we will use these properties and the stable
manifold theorem to discuss the RG flow. Recall that we denote the RG map
by $R(\alpha,B)=(\alpha^\prime,B^\prime)$. In the previous section it was defined
for $\alpha \in \Pi=\{\alpha : \delta<\Re \alpha<1-\delta, |\Im \alpha|<w\}$
and $\|B\|<\epsilon$. Recall that $R$ is an analytic function on this domain (see Appendix \ref{abstract}). We will use this analyticity to get bounds on
the derivative in a smaller domain. The smaller domain consists of
real $\alpha$ with $2 \delta < \alpha < 1-2\delta$ and real $B$ with
$\|B\|<\epsilon/2$. We let $DR(\alpha,B)$ denote the derivative of $R$.

{\bf Notation:} We will use $O(B^n)$ notation from Section \ref{RG_lowT} (p. \pageref{OBn}) only for analytic functions. We will use $O(\|B\|^n)$ notation for bounds on functions restricted to the real domain. The constants in $O(\|B\|^n)$, like in $O(B^n)$, are independent of $\eps$ but may depend on $\delta$ and $w$. 

\begin{theorem}
	\label{DRthm}
	The derivative $DR(\alpha,B)$ has the following structure.
	Divide $DR$ into four blocks corresponding to the splitting of
	the space into $\alpha$ and $B$.
	We denote the restrictions of  $DR(\alpha,B)$ to these four blocks by
	$\partial_{\alpha} \alpha'$, $\partial_\alpha B'$,
	$\partial_B \alpha'$, and $\partial_B B'$.
	Then, for any $\eps$ sufficiently small, there is a constant  $l<1$ such that for real $\alpha$
	with $2 \delta < \alpha < 1-2\delta$ and real $B$ with
	$\|B\|<\epsilon/2$ we have
	\begin{equation}
		\| \partial_B B'(\alpha,B) \| \le  l<1\ .
		\label{dra}
	\end{equation}
	Furthermore
	\begin{eqnarray}
		&&\partial_\alpha \alpha'(\alpha,B) = r^\prime(\alpha) + O(\|B\|^2)\ ,
		\label{drb}   \\
		&&  \|\partial_B \alpha'(\alpha,B) \| = O(\|B\|)\ ,
		\label{drc}   \\
		&& \| \partial_\alpha B'(\alpha,B) \| = \epsilon^{1/2} O(\|B\|) \ .
		\label{drd} 
		\label{dr}
	\end{eqnarray}
\end{theorem}
\begin{remark}
	Recall that $r(\alpha)$ was defined in \eqref{LT0map}. Note that $r(1/2)=1/2$. 
	On $[0,1]$ the derivative $r^\prime(\alpha)$ is non-negative and has
	its maximum at $r^\prime(1/2)=16$. 
\end{remark} 

We will derive bounds on $DR$ from bounds on $R$ by using the following basic property of analytic maps on Banach spaces, which we restate here for convenience:

\begin{lemma}[Appendix \ref{abstract}, Property \ref{Taylor-prop}, Eq.~\eqref{Taylor-eq}] \label{Dfbound} Let $X,Y$ be two complex Banach spaces, $\mathcal{U}$ an open subset of $X$, and $f:\mathcal{U}\to Y$ an analytic function. If $f$ is bounded by $M$ in a ball of radius $\varrho$, then the operator norm $\|D f\|_{\mathcal{B}(X,Y)}$ at the center of the ball is bounded by $M/\varrho$. 
\end{lemma}

\begin{proof}[Proof of Theorem \ref{DRthm}]
	We use the properties of the RG map given in Theorem \ref{LTmap}.
	For $\alpha\in \Pi$, $\|B\|<\eps$ we have
	\begin{eqnarray}
		\alpha' &=& r(\alpha) +O(B^2),
		\label{drpropa}
		\\
		B' &=& \varepsilon^{1/2} O(B).
		\label{drpropb}
	\end{eqnarray}
	We will apply Lemma \ref{Dfbound} to each of
	$\partial_{\alpha} \alpha'$, $\partial_\alpha B'$,
	$\partial_B \alpha'$, and $\partial_B B'$ separately. 
	
	Consider first $\partial_\alpha \alpha'$. Note that the $O(B^2)$ term in Eq.~\eqref{drpropa}
	depends on $\alpha$. Since we are studying the partial derivative, $B$ is fixed and we are studying the ordinary analytic function $\alpha'(\cdot,B): \Pi \to \mathbb{C}$.
	To bound its derivative with respect to $\alpha$ we choose a disk of radius $\varrho=\min(\delta,w)$ in $\mathbb{C}$ centered at the point $\alpha$ of interest, Since we are assuming $2\delta<\alpha<1-2\delta$, $\delta\le w$, this disk is in $\Pi$ and we can use \eqref{drpropa}. By Lemma \ref{Dfbound}, the $\alpha$-derivative of the error term is bounded by $O(\|B\|^2)/\varrho$, i.e.
	\begin{equation}
		\partial_\alpha \alpha' = r'(\alpha)+ \varrho^{-1}{O(\|B\|^2)}.
	\end{equation}
	This proves Eq.~\eqref{drb}.
	
	Next we consider $\partial_\alpha B'$ and use Eq.~\eqref{drpropb}. We are studying the analytic function $B'(\cdot,B):\mathbb{C}\to \HLT$. We take the same disk of radius $\varrho=\min(\delta,w)$ in $\mathbb{C}$ and $M= \varepsilon^{1/2} O(\|B\|)$.
	So Lemma \ref{Dfbound} implies the 
	norm of this derivative is bounded by $\varepsilon^{1/2} O(\|B\|)/\varrho$,
	which proves Eq.~\eqref{drd}. 
	
	Now consider $\partial_B \alpha'$.
	$r(\alpha)$ does not depends on $B$, so we
	only need to bound the derivative of the $O(B^2)$ term in Eq.~\eqref{drpropa}
	with respect to $B$. We are studying the analytic function $\alpha'(\alpha,\cdot):\HLT\to \mathbb{C}$.
	Since $\|B\|<\epsilon/2$ we can take $\varrho=\|B\|$ 
	and the ball about $B$ in $\HLT$ of radius $\varrho$ will lie inside the domain of analyticity.
	We take $M=O(\|B\|^2)$, and Lemma \ref{Dfbound} then implies 
	the norm of this derivative is $O(\|B\|)$, which proves Eq.~\eqref{drc}. 
	
	Finally consider $\partial_B B' $. We are studying the analytic function $B'(\alpha,\cdot):\HLT\to \HLT$. We take the same ball of radius $\varrho=\|B\|$ in $\HLT$. We take $M= \varepsilon^{1/2} O(\|B\|)$, so Lemma \ref{Dfbound} implies that 
	the norm of this derivative is $O(\epsilon^{1/2})$. For $\eps$ sufficiently small this proves \eqref{dra}.
\end{proof}

\begin{theorem}
	\label{gthm}
	As in Section \ref{RG_lowT},
	let $\normfactor $ be the normalization factor in the RG map.
	Let $g(\alpha,B)=\ln(\normfactor )$.
	Then, for any $\eps$ sufficiently small, $g$ is an analytic function on the open domain \eqref{domainLT}, $g$ is bounded, and the derivative of $g$ is bounded on subsets of this domain having positive distance from its complement.
	Let $g^\prime(\alpha,B)$ be the derivative of $g(\alpha,B)$ with respect
	to $\alpha$ and let $\nabla g(\alpha,B)$ be the gradient with respect to
	the variables in $B$. 
	For $\alpha \in (\delta,1-\delta)$ we have
	\begin{eqnarray}
		\nabla g(\alpha,0) = 0.
		\label{gradB}
	\end{eqnarray}
	Furthermore, $g^\prime(1/2,0)=0$.
\end{theorem}

\begin{proof}
	Recall Eq.~\eqref{LTnu}. For $\alpha \in \Pi$ the polynomial 
	$\alpha^{16} + (1-\alpha)^{16}$ is bounded away from zero. Hence, for $\|B\|<\eps$ with $\eps$ sufficiently small, 
	$g$ is analytic and bounded and 
	\begin{equation}
		g(\alpha,B) = \ln(\alpha^{16} + (1-\alpha)^{16}) + O(B^2).
	\end{equation}
	Eq.~\eqref{gradB} follows. $g(\alpha,0)$ is symmetric about $\alpha=1/2$, so
	$g^\prime(1/2,0)=0$. Since $g$ is bounded on the domain, the derivative is bounded by Lemma \ref{Dfbound} on subsets having positive distance from the complement.
\end{proof}

The point $(\alpha,B)=(1/2,0)$ is a fixed point of the RG map called the low-T fixed point. Since $r^\prime(1/2)=16$, it is an unstable fixed point. We will now apply the stable manifold theorem (the statement and details on how it applies here can be found in Appendix \ref{stable}).
Using the properties on $DR$ established in Theorem \ref{DRthm}, the stable manifold theorem implies the following.
By reducing $\epsilon$ if needed, 
there is a $C^\infty$ real-valued function $\alpha_c(B)$
defined for $\|B\|<\epsilon/2$ such that the graph
$\{(\alpha_c(B),B): \|B\|<\epsilon/2 \}$
is the local stable manifold, i.e., the part of the stable manifold
in a neighborhood of $(1/2,0)$. (The stable manifold is the set of
$(\alpha,B)$ which flow to the unstable fixed point under iteration of
the RG map.) 
Proposition \ref{rgflow} in the appendix shows the following.
If we start at $(\alpha,B)$ with $\alpha>\alpha_c(B)$
then the RG trajectory will flow into the region $\alpha>1 - 2\delta$, 
and if $\alpha<\alpha_c(B)$ it will flow into the region $\alpha< 2\delta$. 
If we start on the stable manifold, i.e., at some $(\alpha_c(B),B)$,
then the RG trajectory flows to the $(1/2,0)$ fixed point. 

\begin{remark} The $(1/2,0)$ fixed point describes the direct mixture of the $+$ and $-$ phases. Since any point on the stable manifold flows to the $(1/2,0)$ fixed point, there should be two translationally invariant Gibbs states at any point on the stable manifold. In this work we won't prove this, but we may still intuitively refer to the stable manifold as the phase coexistence surface.
\end{remark}

\subsection{The free energy}

Recall that $\delta$ was fixed early on so that the free energy is analytic for $\alpha$ close to 0 or 1 and $B$ small, specifically on the set $\Omega$ defined in Eq.~\eqref{deltachoice}.

We choose $\eps$ sufficiently small so that all the results about the RG map from the previous subsection apply. We also require $\eps \le \delta$. This
insures that when we iterate the RG map with $\alpha \neq \alpha_c(B)$
we will end up in $\O$.

Let $f(\alpha,B,L)$ be the free energy per site for the tensor network
with $L^2$ sites and the tensor
\begin{equation}
	\alpha A^{(+)} + (1-\alpha) A^{(-)} +B.
	\label{initial_tensor}
\end{equation}

\begin{theorem}
	For $\epsilon$ sufficiently small, if $\|B\|<\epsilon/2$ and
	$\alpha \in [0,1]$ then the infinite volume limit of the free energy
	\begin{eqnarray}
		f(\alpha,B) = \lim_{n \to \infty} f(\alpha,B,4^n)
	\end{eqnarray} 
	exists and is continuous.
	It is a real-analytic function for $\alpha \neq \alpha_c(B)$. 
	\label{freethm}
\end{theorem}

\begin{proof}
	After applying the RG map from Theorem \ref{LTmap} to the tensor in Eq.~\eqref{initial_tensor}
	we get a new tensor 
	\begin{eqnarray}
		\alpha^\prime A^{(+)} + (1-\alpha^\prime) A^{(-)} +B^\prime
		\label{tensor}
	\end{eqnarray}
	and the normalization factor $\normfactor  = \exp(g(\alpha,B))$.
	The partition function is preserved and the number of sites is reduced
	by a factor of $4^2=16$, leading to (see Eq.~\eqref{RGpropertyLT})
	\begin{equation}
		f(\alpha,B,L) = \frac{1}{16} f(\alpha^\prime,B^\prime,L/4)
		+ \frac{1}{16} g(\alpha,B) .
	\end{equation}
	Note that $g(\alpha,B)$ does not depend on $L$. 
	We iterate the RG map and let $(\alpha_j,B_j)$ be the resulting trajectory with
	$(\alpha_0,B_0)=(\alpha,B)$. By \eqref{dra}, the RG map is contracting in the $B$ direction, in particular $\|B_j\|\le l^j \|B_0\|<\eps/2$.
	\footnote{Since $R(\alpha,0)=(r(\alpha),0)$,
		$B'=\int_0^B \, \partial_{\overline{B}} B'(\alpha,\overline{B}) \, d \overline{B}$. So $\|B'\| \le l \|B\|$.}
	We only iterate the map as long as
	$2 \delta < \alpha < 1 - 2 \delta$. So for $n$ such that
	$2 \delta < \alpha_{n-1} < 1- 2\delta$ we have
	\begin{eqnarray}
		f(\alpha,B,L) = f(\alpha_n,B_n,L4^{-n}) 16^{-n}
		+ \sum_{j=1}^n g(\alpha_{j-1},B_{j-1})  16^{-j}.
		\label{iterate}
	\end{eqnarray}
	We define $N_3(\alpha,B)$ to be the first $n$ such that
	$\alpha_n < 3 \delta$ or $\alpha_n > 1-3 \delta$, i.e., the first $n$ such that
	$(\alpha_n,B_n) \in \O$. $N_3(\alpha,B)$ is finite for $\alpha,B$ not on the stable manifold.
	Since $N_3(\alpha,B)$ is integer valued, it will be discontinuous at some $(\alpha,B)$.
	For this reason we need the region where Theorem \ref{DRthm} holds and the region
	where we know the free energy is analytic to overlap.
	This is the reason for the factors of $3$ in the definition of $\O$. 
	
	For every $n_3=0,1,2,\ldots$ consider the set of those real $(\alpha,B)$ in the domain of the RG map for which $(\alpha_{n_3},B_{n_3})\in\Omega$ while $\alpha_{n_3-1}\in (2\delta,1-2\delta)$ (we only impose the latter condition if $n_3>0$):
	\beq
	W_{n_3}=\{ (\alpha,B): \alpha\in(2\delta,1-2\delta), \|B\|<\eps, \alpha_{n_3}\in (0,3\delta)\cup (1,1-3\delta), \alpha_{n_3-1}\in (2\delta,1-2\delta) \}.
	\eeq
	Since the RG map is continuous, each $W_{n_3}$ is an open set. Also we have that the union of each $W_{n_3}$ covers everything but the stable manifold:
	\beq
	\label{cover}
	\bigcup_{n_3=0}^\infty W_{n_3} = \{ (\alpha,B): \alpha\in(2\delta,1-2\delta), \|B\|<\eps, \alpha\ne\alpha_c(B)\}.
	\eeq
	
	For $(\alpha,B) \in W_{n_3}$ we can let $n=n_3$ in Eq.~\eqref{iterate}.
	Since $(\alpha_{n_3},B_{n_3}) \in \O$, we know that
	the limit as $L \rightarrow \infty$ of $f(\alpha_{n_3},B_{n_3},L4^{-n_3})$ exists.
	Denote it by $f(\alpha_{n_3},B_{n_3})$. Hence Eq.~\eqref{iterate} implies that
	the limit as $L \rightarrow \infty$ of $f(\alpha,B,L)$ also exists.
	Denote it by $f(\alpha,B)$.  We have
	\begin{eqnarray}
		f(\alpha,B) = f(\alpha_{n_3},B_{n_3}) 16^{-n_3}
		+ \sum_{j=1}^{n_3} g(\alpha_{j-1},B_{j-1})  16^{-j}\qquad\text{for }(\alpha,B)\in W_{n_3}.
		\label{n3}
	\end{eqnarray}
	$f(\alpha_{n_3},B_{n_3})$ and the $g(\alpha_{j-1},B_{j-1})$ in the finite sum
	are real-analytic functions of $(\alpha,B)$ in $W_{n_3}$. 
	So this equation proves the free energy is real-analytic in $W_{n_3}$. Since the union of $W_{n_3}$ covers everything but the stable manifold, Eq.~\eqref{cover}, the free energy is real-analytic off of the stable manifold. 
	
	If our initial tensor is on the stable manifold, i.e., $\alpha = \alpha_c(B)$, then Eq.~\eqref{iterate} holds for
	all $n$. Setting $L=4^{n+1}$ in Eq.~\eqref{iterate} we have
	\begin{eqnarray}
		f(\alpha,B,4^{n+1}) = f(\alpha_n,B_n,4) 16^{-n}
		+ \sum_{j=1}^n g(\alpha_{j-1},B_{j-1})  16^{-j}.
	\end{eqnarray}
	Analogously to Lemma \ref{2x2}, it is easy to see that $f(\alpha_n,B_n,4)$ is bounded for sufficiently small $\eps$, because $(\alpha_n,B_n)$ is very close to $(1/2,0)$). Also $g$ is bounded (Theorem \ref{gthm}).
	This shows that the limit of $f(\alpha,B,4^{n+1})$ as $n \rightarrow \infty$ exists and is equal to 
	\begin{eqnarray}
		f(\alpha,B) = \sum_{j=1}^\infty g(\alpha_{j-1},B_{j-1})  16^{-j}.
		\label{stableinf}
	\end{eqnarray}
	Note that this also shows that $f(\alpha,B)$ is bounded on the stable manifold. 
	We define $N_3=\infty$ for $\alpha=\alpha_c(B)$. So Eq.~\eqref{n3} holds
	both on and off the stable manifold with the understanding that when
	$N_3=\infty$ we take $16^{-N_3}$ to be $0$. 
	
	Finally we must prove that the free energy is
	continuous at points on the stable manifold. Fix a point $(\alpha_c(B^\prime),B^\prime)$ on the stable manifold.
	Let $\gamma > 0 $. We must show there is an $\eta>0$ such that
	\begin{eqnarray}
		\|(\alpha,B)-(\alpha_c(B^\prime),B^\prime)\|<\eta \, \implies \,
		|f(\alpha,B)-f(\alpha_c(B^\prime),B^\prime)|<\gamma .
	\end{eqnarray}
	Pick $N$ such that
	\begin{eqnarray}
		16^{-N} \max \{\|f\|_\infty,\|g\|_\infty\}  < \frac{\gamma}{4}
		\label{fbound}
	\end{eqnarray}
	and make $\eta$ small enough that $N_3(\alpha,B) \ge N$.  
	We write $f(\alpha,B)$ as 
	\begin{eqnarray}
		f(\alpha,B) = f(\alpha_{N_3},B_{N_3}) 16^{-N_3}
		+ \sum_{j=1}^{N} g(\alpha_{j-1},B_{j-1})  16^{-j}
		+ \sum_{j=N+1}^{N_3} g(\alpha_{j-1},B_{j-1}) 16^{-j},
	\end{eqnarray}
	and write $f(\alpha_c(B^\prime),B^\prime)$ as 
	\begin{eqnarray}
		f(\alpha_c(B^\prime),B^\prime) = f(\alpha^\prime_N,B^\prime_N) 16^{-N}
		+ \sum_{j=1}^{N} g(\alpha^\prime_{j-1},B^\prime_{j-1})  16^{-j}.
	\end{eqnarray}
	Here $(\alpha^\prime_j,B^\prime_j)$ denotes the trajectory starting from $(\alpha_c(B^\prime),B^\prime)$.
	Using Eq.~\eqref{fbound} we have
	\begin{align}
		| f(\alpha,B) - f(\alpha_c(B^\prime),B^\prime)| &\le 2 \frac{\gamma}{4} 
		+ \Bigl| \sum_{j=1}^N [g(\alpha_{j-1},B_{j-1}) - g(\alpha^\prime_{j-1},B^\prime_{j-1})] 16^{-j}\Bigr| \\
		&\hspace{1cm}+ \Bigl| \sum_{j=N+1}^{N_3} g(\alpha^\prime_{j-1},B^\prime_{j-1}) 16^{-j}\Bigr|. \nn
	\end{align}
	Each $g(\alpha_{j-1},B_{j-1})$ is a continuous function of $(\alpha,B)$. So since $N$ is fixed we can
	make the first sum $<\gamma/4$ by making $\eta$ sufficiently small. 
	The second sum is bounded by $16^{-N} \|g\|_\infty$ and so is $<\gamma/4$.
\end{proof}

\begin{remark} \label{bounded}
	Eq.~\eqref{n3} can be used to show that any derivative of the free energy is bounded on any domain having positive distance from the stable manifold. This will be useful in the proof of Lemma \ref{lemma3} below. To see this:
	\begin{enumerate}
		\item
		Note that \eqref{n3} implies that on $W_{n_3}$ any derivative of the free energy is bounded by a constant depending only on $n_3$. 
		\item Let us understand the shape of $W_{n_3}$. Let us focus on the case $\alpha>\alpha_c(B)$. The number of steps it takes to reach $\Omega$ depends on $\alpha - \alpha_c(B)$. For $\alpha-\alpha_c(B)$ small, this vertical distance from the stable manifold increases by factor $\approx16$ during each RG step (see the argument leading to Eq.~\eqref{geomgrowth} in Appendix \ref{stable}). It follows that $W_{n_3}$ consists from all points $(\alpha,B)$ where
		\beq
		\alpha - \alpha_c(B)\in ( u_{n_3}(B), v_{n_3}(B) ),
		\eeq
		where $u_{n_3}(B)$, $v_{n_3}(B)$ are smooth functions of order $16^{-n_3}$ (for large $n_3$). 
	\end{enumerate}
	Combining 1 and 2 we get the statement at the beginning of the remark. 
	
	It follows from the Pirogov-Sinai theory that much more is true: the free energy allows $C^\infty$ continuation to the closed regions $\alpha\le \alpha_c(B)$ and $\alpha\ge \alpha_c(B)$. In other words, any  derivative of the free energy is bounded for $\alpha \ne \alpha_c(B)$ and has finite limits as $\alpha \to \alpha_c(B)^\pm$. In this paper we will only establish this for one particular derivative of the free energy, the magnetization $m(\alpha,B)$, see Theorem \ref{magthm}.
	
	We note in passing that free energy does not allow analytic continuation from $\alpha\le \alpha_c(B)$ and $\alpha\ge \alpha_c(B)$ across the first-order phase transition surface, see \cite{Isakov1984,Friedli2004}.
\end{remark}

\subsection{The magnetization}

Let
\begin{eqnarray}
	m(\alpha,B,L)= \partial_\alpha f(\alpha,B,L)\,.
\end{eqnarray}
We will refer to this quantity as the magnetization. If our tensor
network comes from the Ising model, then this quantity is 
not the usual magnetization for the Ising model. However, as discussed in
section \ref{LTRG} (see Eq.~\eqref{twofs}), it has the same continuity
and analyticity properties as the Ising model magnetization. 

From the explicit $m(\alpha,0)$ expression, Eq.~\eqref{ma}, we know that
this magnetization for $B=0$  has a jump at $\alpha=1/2$, namely $m(1/2^\pm,0)=\pm 2$.
We will now show for $\|B\|<\eps/2$ that an infinite volume limit of the magnetization exists and that this infinite volume magnetization is a continuous function of
$(\alpha,B)$ if we are not on the stable manifold and still has a jump as we cross
the stable manifold, close in magnitude to the $B=0$ jump.

\begin{theorem}
	\label{magthm}
	For $\|B\|<\epsilon/2$ and $\alpha \in [0,1]$ with
	$\alpha \neq \alpha_c(B)$, the infinite volume limit
	\begin{eqnarray}
		m(\alpha,B) = 
		\lim_{L \to \infty} m(\alpha,B,L) 
	\end{eqnarray} 
	exists. It is real-analytic when $\alpha \neq \alpha_c(B)$. 
	The one-sided limits of $m(\alpha,B)$ as $\alpha \to \alpha_c(B)$
	from above or below exist and are not equal. They are given by 
	\begin{eqnarray}
		m(\alpha_c(B)^+,B) &=& d(B) m(1/2^+,0) + \delta m(B)\ ,  \\
		m(\alpha_c(B)^-,B) &=& d(B) m(1/2^-,0) + \delta m(B)\ , 
	\end{eqnarray}
	where $\delta m(B)=O(\|B\|)$ and $d(B)=1+O(\|B\|)$.
	($d(B)$ and $\delta m(B)$ will be defined explicitly later in
	Eqs.~\eqref{dbdef} and \eqref{deltamag}.)
\end{theorem}

\begin{remark}
	One might expect that we should have $\delta m(B)=O(\|B\|^2)$ and $d(B)=1+O(\|B\|^2)$ since any natural way
	to obtain a scalar from the tensor $B$ will be $O(\|B\|^2)$. With significantly more work these stronger
	estimates can in fact be proven, but we will not pursue this here. 
\end{remark}

\begin{proof}
	Recall that $N_3(\alpha,B)$ is the smallest $n$ such that
	$(\alpha_n,B_n) \in \O$ if  $\alpha \neq \alpha_c(B)$, and 
	$N_3(\alpha,B) = \infty$ if $\alpha = \alpha_c(B)$.
	
	Suppose that $(\alpha,B)$ is not on the stable manifold. 
	We take $n=N_3(\alpha,B)$ in  Eq.~\eqref{iterate} and differentiate this equation
	with respect to $\alpha$, keeping $n$ fixed:
	\begin{eqnarray}
		\frac{d}{d \alpha} f(\alpha,B,L) =\frac{d}{d \alpha} f(\alpha_n,B_n,L4^{-n}) 16^{-n}
		+ \sum_{j=1}^n   \frac{d}{d \alpha} g(\alpha_{j-1},B_{j-1})  16^{-j}\ .
		\label{diterate}
	\end{eqnarray}
	We use a total derivative notation $d/d\alpha$ to emphasize that both $B_k$ and $\alpha_k$ depend on
	$\alpha$. The same argument that proved
	the $L \rightarrow \infty$ limit of the free energy exists and is real-analytic
	also proves that the limit of the magnetization exists and is real-analytic
	off the stable manifold.
	Since $(\alpha_n,B_n)$ is in $\O$ the expansion methods that can be applied in $\O$ show that
	$d/d\alpha$ and the infinite volume limit commute. 
	So the infinite volume limit of the magnetization is given by 
	\begin{eqnarray}
		m(\alpha,B) = \frac{d}{d \alpha} f(\alpha_n,B_n)
		+ \sum_{j=1}^n   \frac{d}{d \alpha} g(\alpha_{j-1},B_{j-1})  16^{-j}\ .
	\end{eqnarray}
	
	The rest of the proof is devoted to the one-sided
	limits of the magnetization as we approach the stable manifold. 
	
	If we start close to the stable manifold but not on it,
	then the trajectory starting at the initial $(\alpha,B)$ will pass close to the unstable fixed point and eventually
	move away from it towards either $(\alpha,B)=(0,0)$ or $(1,0)$. 
	We will consider the latter case which happens when the initial
	$(\alpha,B)$ has $\alpha > \alpha_c(B)$. The treatment of the
	other case is essentially identical.
	
	We will split the trajectory into two segments.
	We start by defining a small neighborhood $V$ of the unstable fixed point:
	\begin{eqnarray}
		V= \{ (\alpha,B): \|B\| < \tau_B, \quad |\alpha-1/2| < \tau_\alpha \}\ .
	\end{eqnarray}
	The quantity $\tau_B$ should not be confused with $\epsilon$. The initial $B$ 
	satisfies $\|B\|<\epsilon/2$, but it need not satisfy $\|B\|<\tau_B$.
	In the proof of Theorem \ref{magthm} we will use arbitrarily small $\tau_B$ and $\tau_\alpha$. 
	We will refer to $V$ as the unstable neighborhood. 
	
	Now consider $\alpha$ which are close enough to $\alpha_c(B)$ that the trajectory passes
	through $V$.
	Given $B$ we let $N_1=N_1(B)$ be the number of steps until the trajectory on
	the stable manifold which starts at $(\alpha_c(B),B)$ enters  $V$.
	Since $V$ is open, if we start at 
	any $(\alpha',B')$ which is close to $(\alpha_c(B),B)$, then by step $N_1$
	the trajectory will be in $V$.
	For $\alpha>\alpha_c(B_0)$, 
	let $N_2=N_2(\alpha,B)>N_1$ be the first term after the trajectory has left
	$V$. We stop the RG at this point. The $L \rightarrow \infty$ limit of Eq.~\eqref{iterate} with $n=N_2$
	gives
	\begin{eqnarray}
		f(\alpha,B) = f(\alpha_{N_2},B_{N_2}) 16^{-N_2}
		+ \sum_{j=1}^{N_2} g(\alpha_{j-1},B_{j-1})  16^{-j}\ .
		\label{n2}
	\end{eqnarray}
	
	For a trajectory starting at $(\alpha,B)$ with $\alpha>\alpha_c(B)$ we split the sum above into two sums corresponding to $(\alpha_{j-1},B_{j-1})$ belonging to the following two segments of the trajectory.
	The two segments are (see Figure \ref{fig_RG_flow}) 
	\begin{itemize}
		\item Segment 1: $j$ runs from $1$ to $N_1$.
		This corresponds to the part of the trajectory $(\alpha_{j-1},B_{j-1})$ from the starting point $(\alpha,B)$ which has not yet entered the unstable neighborhood $V$.
		\item Segment 2: $j$ runs from $N_1+1$ to $N_2$.
		This is the part of the trajectory with $(\alpha_{j-1},B_{j-1})\in V$.
	\end{itemize}
	\begin{figure}[!h]
		\includegraphics[scale=1.2]{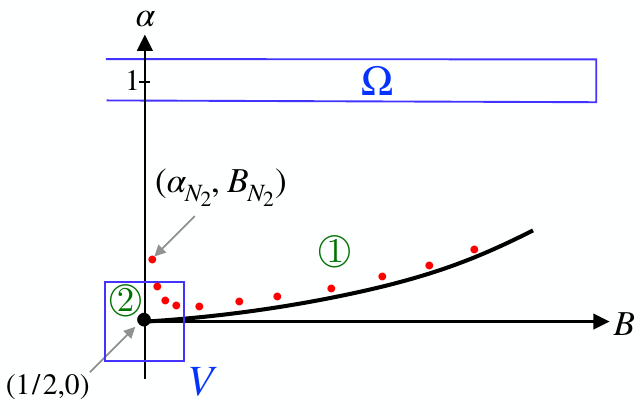}
		\caption{\leftskip=25 pt \rightskip= 25 pt 
			This illustrates the unstable neighborhood $V$,
			the region $\O$ and the two segments of the RG trajectory. The red dots
			are the RG trajectory which starts at the rightmost red dot
			slightly above the stable manifold.
			The two segments are labelled 1 and 2. 
		}
		\label{fig_RG_flow}
	\end{figure}

	Consider the number of steps in the trajectory in each of the
	two regions.
	The number of steps in segment 1 depends on $B_0$ and the size of $V$,
	but it does not depend on $\alpha$ if $\alpha$ is restricted to a small
	neighborhood of $\alpha_c(B_0)$.
	The number of steps in segment 2 is unbounded as $\alpha \to \alpha_c(B_0)$.

	We now differentiate Eq.~\eqref{n2} with respect to $\alpha$ and use the chain rule.
	To keep the notation under control, for a function $f(\tilde \alpha,\tilde B)$ we let $f^\prime$ denote the partial
	derivative with respect to the first scalar argument $\tilde \alpha$ and let $\nabla$ denote the
	partial derivatives with respect to the tensor argument $\tilde B$. This notation will be used for functions whose direct arguments are not $\alpha,B$ but instead functions of $\alpha,B$. 
	So we have
	\begin{eqnarray}
		\label{magsum}	m(\alpha,B)  & = & \sum_{j=1}^{N_2} \left[ g^\prime(\alpha_{j-1},B_{j-1}) \,
		\partial_\alpha \alpha_{j-1}(\alpha,B) 
		+ \nabla g(\alpha_{j-1},B_{j-1}) \cdot
		\partial_\alpha B_{j-1}(\alpha,B) \right] 16^{-j}
		\\
		&+& \frac{d}{d \alpha} f(\alpha_{N_2},B_{N_2}) 16^{-N_2}\ . \nn
	\end{eqnarray}
	(In the last term we have used a total derivative $d/d\alpha$ to emphasize that in
	$f(\alpha_{N_2},B_{N_2})$ both $\alpha_{N_2}$ and $B_{N_2}$ depend on $\alpha$.) 
	We split the sum in \eqref{magsum} into three pieces:
	\begin{eqnarray}
		m(\alpha,B) = m_1(\alpha,B) + m_2(\alpha,B) + m_3(\alpha,B)\ .
		\label{mthreesum}
	\end{eqnarray}
	The first two terms are the restrictions of the sum in Eq.~\eqref{magsum}
	to segments 1 and 2. The last term is
	\begin{eqnarray}
		m_3(\alpha,B)= \frac{d}{d \alpha} f(\alpha_{N_2},B_{N_2}) 16^{-N_2} \ .
	\end{eqnarray}
	Note that if we start at $(\alpha_c(B),B)$ on the stable manifold, then
	$m_1$ is still defined and has $N_1(B)$ terms, $m_2$ has infinitely many terms
	and $m_3$ does not appear.
	
	Next we state lemmas for each of the three terms $m_1,m_2,m_3$, proved below.
	
	\begin{lemma}[$m_1$] \label{lemma1} If $\|B\|<\eps/2$ then we have,
		for a fixed $\tau_B$,
		\begin{eqnarray}
			\lim_{\alpha \to \alpha_c(B)^+} m_1(\alpha,B) = m_1(\alpha_c(B),B)\ .
		\end{eqnarray}
	\end{lemma}

	\begin{lemma}[$m_2$] \label{lemma2}
		
		(a) 
		If $\|B\|<\eps/2$ and $\gamma>0$ then for all sufficiently small $\tau_\alpha>0$ and $\tau_B>0$ 
		we have 
		\beq
		|m_2(\alpha,B)|\le \gamma
		\label{m2aB}
		\eeq
		for any $\alpha$ close enough to $\alpha_c(B)$ that the RG trajectory passes through $V$,
		including $\alpha=\alpha_c(B)$.
		(If $\alpha=\alpha_c(B)$  the number of terms in segment 2 is
		infinite since the trajectory will converge to the unstable fixed point.)
		
		(b) For $\|B\|<\eps/2$ let $(\alpha_j,B_j)$ be the trajectory on the stable
		manifold which starts at $(\alpha_c(B),B)$.
		Since the trajectory starts on the stable manifold, it 
		stays on the stable manifold, and so $\alpha_j=\alpha_c(B_j)$.
		Define 
		\begin{align}
			\delta m(B) &= \sum_{j=1}^\infty [ g^\prime(\alpha_c(B_{j-1}),B_{j-1}) \,
			\partial_\alpha \alpha_{j-1}(\alpha_c(B),B) 
			+ \nabla g(\alpha_c(B_{j-1}),B_{j-1}) \cdot
			\partial_\alpha B_{j-1}(\alpha_c(B),B)] 16^{-j}
			\label{deltamag}\\
			&\equiv m_1(\alpha_c(B),B)+ m_2(\alpha_c(B),B)\ .\nn
		\end{align}
		Then the series converges and $\delta m(B)$ is $O(\|B\|)$.
		
	\end{lemma}
	
	The term $m_3(\alpha,B)$ is where the bulk of the magnetization will appear. The trajectory at $N_2$ is very close to $B=0$, and one may naively expect that $m_3(\alpha,B)\approx m({1/2}^+,0)$. Instead we will get $m_3(\alpha,B)\approx d(B) m({1/2}^+,0)$ with $d(B)=1+O(\|B\|)$. To motivate this, consider the derivative $16^{-N_2} \partial_\alpha \alpha_{N_2-1}$.
	The derivative involves a product of factors of $\frac{1}{16}DR$ along the trajectory.
	This product will get most of its contribution from segment 1 where these factors will differ from $P_\alpha$
	by $O(\|B_j\|)$ (while in segment 2 these derivatives are very close to $P_\alpha$).
	The resulting multiplicative factor is captured by the quantity
	$d(B)$. This rough argument will be made precise in the following lemma:
	
	\begin{lemma}[$m_3$]	\label{lemma3} Let $\|B\|<\eps/2$.
		
		(a) The following limit exists:
		\begin{eqnarray}
			d(B) = \lim_{n \to \infty} 16^{-n} \partial_\alpha \alpha_{n}(\alpha_c(B),B)\ .
			\label{dbdef}
		\end{eqnarray}
		Moreover, $d(B)=1+O(\|B\|)$. 
		
		(b) If $\gamma>0$ then for all sufficiently small $\tau_\alpha>0$ and $\tau_B>0$ we have
		\begin{eqnarray}
			\label{goal3}
			\limsup_{\alpha \to \alpha_c(B)^+} |m_3(\alpha,B) - d(B) m(1/2^+,0)| < \gamma\ .
		\end{eqnarray}
	\end{lemma}
	Note that since $(\alpha_c(B),B)$ is on the stable manifold, the
	trajectory in the definition of $d(B)$ converges to the unstable fixed point.
	
	Assuming these lemmas, the proof of Theorem \ref{magthm} is concluded as follows.
	The magnetization $m(\alpha_c(B)^+,B)$ will be a sum of two terms.
	One term comes from the limit of $m_3$ as $\tau_\alpha,\tau_B$ go to $0$.
	This will equal $d(B) m(1/2^+,0)$.
	The other term comes from $m_1$.
	Its limit as $\tau_\alpha,\tau_B$ go to $0$ will be $\delta m(B)$.
	Since $\delta m(B)$ and $d(B)$ are defined using trajectories on the stable manifold,
	they are the same when we consider the limit as $\alpha \to \alpha_c(B)^-$.
	$\delta m(B)$ is not obviously positive or negative.
	If the original model $(\alpha,B)$ is spin flip symmetric, then
	$\delta m(B)$ must be $0$.

	We split $m(\alpha,B)$ into three terms as in \eqref{mthreesum}.
	Let $\gamma>0$ and chose $\tau_\alpha,\tau_B$ as in the lemmas for $m_2$ and $m_3$.
	Then we use the triangle inequality as follows.
	\begin{align}
		|m(\alpha,B) - d(B) m(1/2^+,0) - \delta m(B)| &\le
		|m_1(\alpha,B)- \delta m(B)| \\
		&+ |m_2(\alpha,B)| + |m_3(\alpha,B) -d(B) m(1/2^+,0)| \ .\nn
	\end{align}
	The lemmas for $m_2$ and $m_3$ imply that the $\limsup$ as $\alpha \to \alpha_c(B)^+$
	of the last two terms in the right side are each bounded by $\gamma$. 
	For the first term in the right side we use the triangle inequality again.
	\begin{eqnarray}
		|m_1(\alpha,B)- \delta m(B)| \le
		|m_1(\alpha,B)-m_1(\alpha_c(B),B)| + |m_1(\alpha_c(B),B) - \delta m(B)|\ .
	\end{eqnarray}
	The lemma for $m_1$ implies that the $\limsup$ as
	$\alpha \to \alpha_c(B)^+$ of the first term in the right side is $0$.
	The second term equals $|m_2(\alpha_c(B),B)|$, and by the lemma for $m_2$ it is $<\gamma$. 
	Putting this all together we have shown that 
	\begin{eqnarray}
		\limsup_{\alpha \to \alpha_c(B)^+}
		|m(\alpha,B) - d(B) m(1/2^+,0) - \delta m(B) | \le 3 \gamma\ .
	\end{eqnarray}
	Since this holds for all $\gamma$ the $\limsup$ must be $0$. 
\end{proof}

It remains to prove the lemmas for segments $m_1,m_2,m_3$.

Let $R^n$ be the $n$-fold composition of $R$ and let $D R^n$ be its derivative. The following lemma will provide control of the derivatives
$\partial_\alpha \alpha_n$  and $\partial_\alpha B_n$. These derivatives are
contained in $D R^n$. Roughly speaking, this lemma says that $DR$ and its compositions are dominated by $r'(1/2)=16$.

\begin{lemma}
	\label{derivlemma}  \begin{itemize}
		\item[(a)] There is a constant $M$ such that 
		\begin{eqnarray}
			\|DR^n(\alpha,B)\| 16^{-n} \le M
		\end{eqnarray}  
		for all $n \ge 1$, all $\alpha \in (2\delta,1-2\delta)$ and
		all $B$ with $\|B\|<\epsilon/2$. 
		
		\item[(b)] Let $P_\alpha$ to be the one-dimensional projection $(\alpha,B) \rightarrow (\alpha,0)$.
		For any trajectory $(\alpha_k,B_k)$, $k=0,\ldots, n-1$ which starts at $(\alpha_0,B_0)$
		with $\alpha_0 \in (2\delta,1-2\delta)$ and $\|B_0\|<\epsilon/2$, we have
		\begin{eqnarray}
			\Bigl	\| 16^{-n} \prod_{k=0}^{n-1} DR(\alpha_k,B_k) - p_n P_\alpha\Bigr\| \le O(\|B_0\|) + l^n\ ,
		\end{eqnarray}
		where $p_n$ is defined by 
		\begin{eqnarray}
			p_n= 16^{-n} \prod_{j=1}^n r^\prime(\alpha_{j-1})\ .
			\label{pdef}
		\end{eqnarray}
		(Recall that $l<1$ is the constant such that
		$\| \partial_B B'(\alpha,B)\| \le  l$.)
		
		\item[(c)] If the trajectory $(\hat{\alpha}_k,\hat{B}_k)$, $k=0,\ldots, n-1$ lies entirely in $V$, then we have
		\begin{eqnarray}
			\Bigl	\| 16^{-n} \prod_{k=0}^{n-1} DR(\hat{\alpha}_k,\hat{B}_k) - P_\alpha\Bigr\| \le O(\tau_\alpha)+O(\tau_B) + l^n\ .
		\end{eqnarray}	
		
	\end{itemize}
\end{lemma}

\begin{remark}The Banach space $\mathbb{C}\times \HLT$ on which $R$ acts (see Eq.~\eqref{domainLT}) has the norm $\max\{|\alpha|,\|B\|\}$. Norms in the lemma are operator norms on $\mathbb{C}\times \HLT$.
\end{remark}

To prove this, we will need the following easy lemma.

\begin{lemma}
	Let $A_i,C_i$ be bounded linear operators on a Banach space for
	$i=1,2,\cdots,n$ with $\|C_i\| \le 1, \, \forall i$ and
	$\sum_{i=1}^n \|A_i\| \le 1$. Then 
	\begin{eqnarray}
		\Bigl\| \prod_{i=1}^n (C_i+A_i) - \prod_{i=1}^n C_i\Bigr\|
		\le 2 \sum_{i=1}^n \|A_i\|\ .
		\label{ACbound1}
	\end{eqnarray}
	\label{AClemma}
\end{lemma}

\begin{remark}
	The operators $A_i$ and $C_i$ need not commute so we must specify the
	order in the above products.
	We make the convention that in any product of operators 
	we take the product going left to right in decreasing order of the index. 
	The reason for this somewhat unusual convention will become clear in
	Eq.~\eqref{chainrule}.
\end{remark}

\begin{proof} We expand $\prod_{i=1}^n (C_i+A_i)$ as a sum of $2^n$ terms and
	cancel the term being subtracted in Eq.~\eqref{ACbound1}.
	We then bound the norm of each term using $\|C_i\| \le 1$ 
	and obtain the upper bound
	\begin{equation}
		\prod_{i=1}^n (1+\|A_i\|) -1 \le \exp(\sum_{i=1}^n \|A_i\|) - 1
		\le 2 \sum_{i=1}^n \|A_i\|\ ,
	\end{equation}
	where the last inequality follows from $e^x-1 \le 2 x$ for $0 \le x \le 1$.
\end{proof}

\begin{proof}[Proof of Lemma \ref{derivlemma}]
	
	Recall that 
	$R^n$ is the $n$-fold composition of $R$, and $D R^n$ is its derivative.
	And $(\alpha_j,B_j)$ is the RG trajectory with $(\alpha_0,B_0)=(\alpha,B)$.
	So $(\alpha_j,B_j)=R(\alpha_{j-1},B_{j-1})$. The chain rule says
	\begin{eqnarray}
		DR^n(\alpha,B) = DR(\alpha_{n-1},B_{n-1}) DR(\alpha_{n-2},B_{n-2})  \cdots
		DR(\alpha_0,B_0) = \prod_{j=0}^{n-1} DR(\alpha_j,B_j)\ .
		\label{chainrule}
	\end{eqnarray}
	As in Lemma \ref{AClemma}, in any product of
	operators like this we take the product going left to right in decreasing order
	of the index. 
	
	We will apply Lemma \ref{AClemma}. We split $\frac{1}{16}DR(\alpha_{j-1},B_{j-1})$ into $A_j+C_j$ as follows.
	\begin{eqnarray}
		C_j &=& \frac{1}{16} 
		\begin{bmatrix}r^\prime(\alpha_{j-1}) &0\\
			0 & \partial_B B'(\alpha_{j-1},B_{j-1}) \end{bmatrix} \\
		A_j &=& \frac{1}{16} DR(\alpha_{j-1},B_{j-1}) -C_j  \ .
		\label{ACdef}
	\end{eqnarray}
	We have included a factor of $1/16$ here since $DR^j$ will appear with a factor of $16^{-j}$.

	The norm of $C_j$ is bounded by $\frac{1}{16} \max \{r^\prime(\alpha_{j-1}),l \}$ which is $\le 1$.
	By Theorem \ref{DRthm} we know that $\|A_j\|=O(\|B_{j-1}\|)$. 
	Since $\|B_{j-1}\|$ decreases geometrically we have
	\begin{eqnarray}
		\sum_{i=1}^n \|A_i\|=O(\|B\|)\ ,
		\label{sumA}
	\end{eqnarray}
	which is $\le 1$ if $\epsilon$ is sufficiently small.
	So Lemma \ref{AClemma} implies
	\begin{equation}
		\|16^{-n} DR^n(\alpha,B) - \prod_{i=1}^n C_i\|  = O(\|B\|)\ .
		\label{iieq}
	\end{equation}

	Since $C_i$ is block diagonal, the entry in the one-dimensional
	$\alpha$ block for $\prod_{i=1}^n C_i$ is $p_n$ 
	with $p_n$ defined by Eq.~\eqref{pdef},
	and the $B$ diagonal block has norm bounded by $l^n$.
	\footnote{
		In fact this norm is bounded by $(l/16)^n$. Since $l<1$ the
		weaker bound of $l^n$ will be sufficient and results in slightly cleaner
		equations.}
	Hence
	\begin{eqnarray}
		\Bigl\| \prod_{j=0}^{n-1} C_j - p_n P_\alpha \Bigr\| \le l^n\ .
		\label{iieq2}
	\end{eqnarray}
	Eqs.~\eqref{iieq} and \eqref{iieq2} imply
	\begin{eqnarray}
		\|16^{-n} DR^n(\alpha,B) - p_n P_\alpha\|  \le  O(\|B\|) + l^n\ .
		\label{prodbound}
	\end{eqnarray}
	Thus $\|16^{-n} DR^n(\alpha,B)\|$ is also bounded as $n \to \infty$.
	This proves part (a) of the lemma.
	Eq.~\eqref{prodbound} also proves part (b).
	
	We will use part (b) to prove (c).
	We let $\hat{p}_n$ denote $p_n$ in Eq.~\eqref{pdef} with $\alpha_{j-1}$
	replaced by $\hat{\alpha}_{j-1}$. 
	We need to bound $|\hat{p}_n-1|$. We have\footnote{The Nienhuis-Nauenberg discontinuity fixed point condition $r'(1/2)=16$ comes into play here.}
	\begin{equation}
		\frac{1}{16} r^\prime(\alpha)=1+O(|\alpha-1/2|)
	\end{equation}
	(in fact $O((\alpha-1/2)^2)$ but this will suffice).
	Because of this, we have
	\beq
	\label{sumtomakesmall}
	|\hat{p}_n-1| =O(\Sigma),\qquad \Sigma:=\sum_{k=1}^{n} |\hat{\alpha}_{k-1}-1/2| \ .
	\eeq
	Note that the number of steps $n$ in $\Sigma$ can be arbitrarily large.
	
	We use
	\beq
	|\hat{\alpha}_{k-1}-1/2| \le  |\alpha_c(\hat{B}_{k-1})-1/2| + |\hat{\alpha}_{k-1}-\alpha_c(\hat{B}_{k-1})| \ .
	\eeq
	It follows from the stable manifold theorem that $\partial_B \alpha_c(B) = O(\|B\|)$.
	Hence
	\begin{equation}
		\alpha_c(B)= 1/2 + O(\|B\|^2) \ .
		\label{acslope}
	\end{equation}
	So $|\alpha_c(\hat{B}_{k-1})-1/2|=O(\|\hat{B}_{k-1}\|^2)$. Since $\|\hat{B}_k\|$ is decreasing geometrically, the sum of $O(\|\hat{B}_{k-1}\|^2)$ over $k$ is $O(\|\hat{B}_{0}\|^2)=O(\tau_B^2)$.
	Furthermore, $|\hat{\alpha}_{k-1}-\alpha_c(\hat{B}_{k-1})|$ essentially grows geometrically with a factor $\approx 16$ (see the argument leading to Eq.~\eqref{geomgrowth} in Appendix \ref{stable}). Therefore the sum of this difference over $k$ is bounded by a constant times the last term $|\hat{\alpha}_{n-1}-\alpha_c(\hat{B}_{n-1})|$. This in turn is bounded by $|\hat{\alpha}_{n-1}-1/2|+|1/2-\alpha_c(\hat{B}_{n-1})|\le \tau_\alpha+O(\tau_B^2)$.
	
	This proves that $\Sigma=O(\tau_\alpha)+O(\tau^2_B)$, and hence so is $|\hat p_n-1|$. Part (c) follows.
\end{proof}

We now prove the lemmas for $m_1,m_2,m_3$.

\begin{proof}[Proof of Lemma \ref{lemma1}]
	For a fixed $\tau_B$, there is a fixed and finite number of terms in the sum for $m_1$ and
	each term is a continuous function of $\alpha$. 
\end{proof}

\begin{proof}[Proof of Lemma \ref{lemma2}]
	We need to bound
	\begin{eqnarray}
		m_2(\alpha,B)  = \sum_{j=n_1+1}^{n_2}
		\left[ g^\prime(\alpha_{j-1},B_{j-1})
		\partial_\alpha \alpha_{j-1}(\alpha,B) 16^{-j} 
		+ \nabla g(\alpha_{j-1},B_{j-1}) \cdot
		\partial_\alpha B_{j-1}(\alpha,B) 16^{-j} \right]\ ,
	\end{eqnarray}
	where to keep the notation simple we have let
	$n_1=N_1(B)$ and $n_2=N_2(\alpha,B)$.
	The lemma will follow from Theorem \ref{gthm} and the fact that in segment 2 the trajectory
	is close to $(\alpha,B)=(1/2,0)$. 
	By part (a) of Lemma \ref{derivlemma}  the derivatives
	$|\partial_\alpha \alpha_{j-1}  16^{-j}|$ and 
	$\| \partial_\alpha B_{j-1} 16^{-j} \|$ are bounded by some constant $M$.
	By Theorem \ref{gthm} we have $\nabla g(\alpha,0) = 0$ and $g^\prime(1/2,0)=0$.
	Since $g$ is analytic there is a constant $c$ such that
	\begin{eqnarray}
		|g^\prime(\alpha_{j-1},B_{j-1})| & \le & c \bigl(|\alpha_{j-1}-1/2| + \|B_{j-1}\|\bigr) \\
		\| \nabla g(\alpha_{j-1},B_{j-1}) \| & \le & c \|B_{j-1}\| \ .
	\end{eqnarray}
	We use
	\begin{equation}
		|\alpha_{j-1}-1/2| \le |\alpha_c(B_{j-1})-1/2| + |\alpha_{j-1}-\alpha_c(B_{j-1})|\ .
	\end{equation}
	By Eq.~\eqref{acslope}, $|\alpha_c(B_{j-1})-1/2|= O(\|B_{j-1}\|^2)$, so we have shown
	\begin{equation}
		|m_2(\alpha,B)| \le cM \sum_{j=n_1+1}^{n_2} [O(\|B_{j-1}\|) + |\alpha_{j-1}-\alpha_c(B_{j-1})|]\ .
		\label{eq2mod}
	\end{equation}
	The rest is identical to the end of the proof of Lemma \ref{derivlemma}. 
	The sum of $O(\| B_{j-1} \|)$  over $j$ is $O(\|B_{n_1}\|)$.
	Since the trajectory
	is already in $V$ at step $n_1$, this is of order $\tau_B$. 
	Inside $V$,  if $\alpha \neq \alpha_c(B)$ then $|\alpha_{j-1}-\alpha_c(B_{j-1})|$ essentially grows geometrically with a
	factor of $\approx16$. So the sum of this difference over $j$ is bounded by a constant times the largest term $|\alpha_{n_2-1}-\alpha_c(B_{n_2-1})|$. We can bound this by 
	$|\alpha_{n_2-1}-\frac{1}{2}| + |\frac{1}{2}-\alpha_c(B_{n_2-1})|$.
	Since the trajectory is still in $V$ at step $n_2-1$, the first term is bounded by $\tau_\alpha$.
	The second term is $O(\|B_{n_2-1}\|^2)$ which is of order $\tau_B^2$. 
	If $\alpha=\alpha_c(B)$, then the $|\alpha_{j-1}-\alpha_c(B_{j-1})|$ term is $0$ for all $j$. 
	This proves part (a).
	
	To prove the convergence of the series in part (b) we can bound the tail of
	the series using the above bounds. We can also use the estimates above
	to bound $\delta m(B)$. The trajectory does not lie entirely in $V$, so we
	do not obtain a bound that is $O(\tau_B)$, but we do obtain a bound that
	is $O(\|B\|)$.
\end{proof}

\begin{proof}[Proof of Lemma \ref{lemma3}]
	For part (a) we start by proving that for a trajectory on the stable manifold,
	$p_n$ converges as $n \rightarrow \infty$. ($p_n$ is defined in Eq.~\eqref{pdef}.)
	Using Eq.~\eqref{acslope} and that for any RG trajectory $\|B_j\|$ converges
	to $0$ exponentially fast, we see that for a trajectory on the stable manifold
	$\alpha_j$ converges to $1/2$ exponentially fast.
	We have $r'(\alpha)=16 + O(|\alpha-1/2|)$, so $p_n$ converges. 
	We have also shown that 
	\begin{equation}
		p_n = 1+O(\|B\|)\ .
		\label{pn}
	\end{equation}
	for our trajectory on the stable manifold starting at $(\alpha_c(B),B)$. 
	
	Next we wish to show that $16^{-n} DR^n(\alpha_c(B),B)$ converges as $n\to\infty$. Let $N,n,m$ be positive integers with $n,m>N$. 
	Using Eq.~\eqref{chainrule} we have
	\begin{equation}
		DR^n(\alpha,B)=DR^{n-N}(\alpha_N,B_N) DR^N(\alpha,B)\ .
	\end{equation}
	with a similar equation for $DR^m(\alpha,B)$. Thus
	\begin{multline}
		16^{-n} DR^n(\alpha,B) - 16^{-m} DR^m(\alpha,B)\\
		= [16^{-n+N} DR^{n-N}(\alpha_N,B_N) - 16^{-m+N} DR^{m-N}(\alpha_N,B_N)]
		16^{-N} DR^N(\alpha,B)\ .
		\label{DRdif}
	\end{multline}
	We use part (a) to bound $\|16^{-N} D^N(\alpha,B)\|$ by $M$.
	Part (b) implies
	\begin{equation}
		\Bigl \| 16^{-n+N} DR^{n-N}(\alpha_N,B_N) - \frac{p_{n}}{p_N} P_\alpha\Bigr\| \le O(\|B_N\|) + l^{-n+N}\ .
	\end{equation}
	with a similar equation with $n$ replaced by $m$. 
	Thus the norm of Eq.~\eqref{DRdif} is bounded by
	\begin{equation}
		M \Bigl[ \frac{|p_{n}-p_{m}|}{p_N}  + O(\|B_N\|) + l^{-n+N} + l^{-m+N}  \Bigr]\ .
	\end{equation}
	Since $p_n$ converges as $n \rightarrow \infty$, this implies
	\begin{equation}
		\limsup_{n,m \rightarrow \infty}
		\| 16^{-n} DR^n(\alpha,B) - 16^{-m} DR^m(\alpha,B)\| \le O(\|B_N\|)\ .
	\end{equation}
	This holds for all $N$, so the $\limsup$ must be zero. This proves the existence of the limit of
	$16^{-n} DR^n(\alpha_c(B),B)$ as $n \rightarrow \infty$.
	The quantity $16^{-n} \partial_\alpha \alpha_{n}(\alpha_c(B),B)$  is the diagonal entry in
	$16^{-n} DR(\alpha_c(B),B)$ corresponding to the $\alpha$ subspace, so this proves the existence of the limit in Eq.~\eqref{dbdef}.
	By part (b) of Lemma \ref{derivlemma} this diagonal entry is 
	$p_j + O(\|B\|) + l^{-j}$. We have shown (Eq.~\eqref{pn}) 
	that $p_j$ is $1+O(\|B\|)$ uniformly in $j$. Thus $d(B)-1=O(\|B\|)$. Part (a) is proved.

	Verifying Eq.~\eqref{goal3} in Part (b) will be rather subtle for the following reason.
	As $\alpha \to \alpha_c(B)^+$ the number of steps in segment 2 diverges.
	So the $\|B_{N_2}\|$ in $m_3$ is going to zero as
	$\alpha \to \alpha_c(B)^+$. So we would like to replace $B_{N_2}$ with $0$.
	The problem is that $\alpha_{N_2}$ depends sensitively on the initial $(\alpha,B)$.
	As $\alpha \to \alpha_c(B)^+$ it does not converge. This problem will be dealt
	with as follows. We will relate $m_3(\alpha,B)$ to the magnetization $m(\alpha_{N_2},B_{N_2})$ .
	The point $(\alpha_{N_2},B_{N_2})$ does depend on $(\alpha,B)$, but this variation is sufficiently
	small that we will be able to argue that $m(\alpha_{N_2},B_{N_2})\approx m(\alpha_{N_2},0)\approx m(1/2^+,0)$.
	
	We have
	\begin{align}
		m_3(\alpha,B) &= 16^{-N_2} \frac{d}{d \alpha} f(\alpha_{N_2},B_{N_2}) \\ \nn
		&= f'(\alpha_{N_2},B_{N_2}) 16^{-N_2} \partial_\alpha \alpha_{N_2}(\alpha,B) +
		\nabla f(\alpha_{N_2},B_{N_2}) \cdot 16^{-N_2} \partial_B \alpha_{N_2}(\alpha,B) \ .
	\end{align}
	Note that $f'(\alpha_{N_2},B_{N_2})=m(\alpha_{N_2},B_{N_2})$. 
	So we can decompose the quantity we want to bound in part (b) as follows.
	\begin{equation}
		m_3(\alpha,B) - d(B) m(1/2^+,0) = E_1 + E_2 + E_3 + E_4 + E_5 + E_6 + E_7\ ,
	\end{equation}
	where
	\begin{align}
		E_1 &= \nabla f(\alpha_{N_2},0) \cdot 16^{-N_2} \partial_B \alpha_{N_2}(\alpha,B) , \\
		E_2 &= [\nabla f(\alpha_{N_2},B_{N_2}) -\nabla f(\alpha_{N_2},0)] \cdot 16^{-N_2} \partial_B \alpha_{N_2}(\alpha,B),  \\
		E_3 &= m(\alpha_{N_2},B_{N_2})
		[16^{-N_2} \partial_\alpha \alpha_{N_2}(\alpha,B) - 16^{-N_1} \partial_\alpha \alpha_{N_1}(\alpha,B)], \\
		E_4 &= m(\alpha_{N_2},B_{N_2}) 16^{-N_1}
		[\partial_\alpha \alpha_{N_1}(\alpha,B) - \partial_\alpha \alpha_{N_1}(\alpha_c(B),B)] , \\
		E_5 &= m(\alpha_{N_2},B_{N_2}) [16^{-N_1} \partial_\alpha \alpha_{N_1}(\alpha_c(B),B) - d(B)] ,\\
		E_6 &= d(B) [m(\alpha_{N_2},B_{N_2}) - m(\alpha_{N_2},0)], \\
		E_7 &= d(B) [m(\alpha_{N_2},0) - m(1/2^+,0)] .
	\end{align}
	
	We start by showing that $m(\alpha_{N_2},B_{N_2})$ is bounded for 
	any $\|B\|<\eps/2$ and any $\alpha$ close enough to $\alpha_c(B)$. 
	By the definition of $N_2$, $\alpha_{N_2} \ge 1/2+\tau_\alpha$. We can
	choose $\eta>0$ such that the region
	\beq
	\label{Aregion}
	A=\{(\alpha,B): 1/2+\tau_\alpha\le \alpha \le 1, \|B\|<\eta \}
	\eeq
	is bounded away from the stable manifold. So, by Remark \ref{bounded}, $f$ and its
	derivatives are bounded on this set. Since $N_2 \rightarrow \infty$ as
	$\alpha  \rightarrow \alpha_c(B)$, for $\alpha$ sufficiently close
	to $\alpha_c(B)$ we will have $(\alpha_{N_2},B_{N_2}) \in A$. In particular $m(\alpha_{N_2},B_{N_2})$ is bounded.
	
	Now we bound $E_1$.
	The derivative $16^{-N_2} \partial_B \alpha_{N_2}(\alpha,B)$ is bounded.
	We will argue that the $\nabla f$ term is small. 
	Let $\hat{B}$ be small and let $\hat{\alpha}=\alpha_{N_2}$. We iterate the RG map starting from
	$(\hat{\alpha},\hat{B})$ until the trajectory enters $\O$.
	Let $(\hat{\alpha}_j,\hat{B}_j)$, $j=0,1,\cdots,N$ be this trajectory.
	Applying $\nabla$ and then setting $B=0$, this yields (compare Eq.~\eqref{n3}) 
	\begin{eqnarray}
		\nabla f(\alpha_{N_2},0) = \nabla f(\hat{\alpha}_N,0) 16^{-N}
		+ \sum_{j=1}^N \nabla g(\hat{\alpha}_{j-1},0)  16^{-j} .
	\end{eqnarray}
	Since $(\hat{\alpha}_N,0) \in \O$, $|\nabla f(\hat{\alpha}_N,0)|$ is bounded. By 
	Eq.~\eqref{gradB}, $\nabla g(\hat{\alpha}_{j-1},0)=0$. So
	$E_1=O(16^{-N})$. As $\tau_\alpha \rightarrow 0$, we have $N \rightarrow \infty$ with $16^{-N}\sim \tau_\alpha$. So
	$E_1=O(\tau_\alpha)$. \footnote{Another argument to control $\nabla f(\hat{\alpha}_N,0) 16^{-N}$
		is the following. In $\O$ we have a convergent expansion for $f(\hat{\alpha}_N,B)$ which
		shows that $f$ is second order in $B$. So $\nabla f(\hat{\alpha}_N,0)=0$. This argument proves that $\nabla f(\alpha_{N_2},0)=0$ and hence $E_1=0$.}
	
	Next we bound $E_3$. We have
	\begin{eqnarray}
		\partial_\alpha \alpha_{N_2}(\alpha,B) 16^{-N_2}
		= P_\alpha DR^{N_2}(\alpha,B) P_\alpha 16^{-N_2}.
	\end{eqnarray} 
	(Recall that $P_\alpha$ is the one-dimensional projection onto the subspace
	corresponding to $\alpha$.)
	We break $DR^{N_2}$ into two factors:
	\begin{eqnarray}
		DR^{N_2}(\alpha,B) 16^{-N_2} = F_2 F_1,
	\end{eqnarray}
	where
	\begin{eqnarray}
		F_1 &=& 16^{-N_1} \prod_{k=0}^{N_1-1} DR(\alpha_k,B_k) ,\\
		F_2 &=& 16^{-(N_2-N_1)} \prod_{k=N_1}^{N_2-1} DR(\alpha_k,B_k) .
	\end{eqnarray}
	Parts (a),(c) of Lemma \ref{derivlemma} imply $F_1=O(1)$ and $\|F_2 -P_\alpha\|= O(\tau_\alpha)+ O(\tau_B) + O(l^{N_2-N_1})$. Hence
	\begin{equation}
		\|P_\alpha F_2 F_1 P_\alpha  - P_\alpha F_1 P_\alpha  \| = O(\tau_\alpha)+ O(\tau_B) + O(l^{N_2-N_1}) .
	\end{equation}
	Since 
	\begin{eqnarray}
		P_\alpha F_1 P_\alpha &=& 16^{-N_1} \partial_{\alpha} \alpha_{N_1}(\alpha,B) ,
	\end{eqnarray}
	we have (using the boundedness of $m(\alpha_{N_2},B_{N_2})$)
	\begin{equation}
		|E_3| \le O(\tau_\alpha)+ O(\tau_B) + O(l^{N_2-N_1}) .
	\end{equation}
	
	Next consider $E_5$. Since $m(\alpha_{N_2},B_{N_2})$ is bounded, part (a)
	of the present lemma says that $E_5$ converges to zero as $N_1 \rightarrow \infty$.
	As $\tau_B \rightarrow 0$, $N_1 \rightarrow \infty$. So $|E_5|=o(1)$ as $\tau_B\to 0$. 
	
	Bounding $E_7$ is trivial since we have an explicit expression
	for $m(\alpha_{N_2},0)$, Eq~\eqref{ma}. So we have
	\beq
	|m(\alpha_{N_2},0) - m(1/2^+,0)| = O(\alpha_{N_2}-1/2) = O(\tau_\alpha).
	\eeq
	So $E_7=O(\tau_\alpha)$.
	
	As $\alpha \rightarrow \alpha_c(B)$, $N_2-N_1 \rightarrow \infty$, and so
	$O(l^{N_2-N_1})$ goes to $0$.
	So the above estimates show that for all sufficiently small $\tau_\alpha,\tau_B>0$
	and all $\alpha$ sufficiently close to $\alpha_c(B)$, 
	we have $|E_1| + |E_3|+|E_5|+|E_7|<\gamma$.
	For the rest of the proof $\tau_\alpha$ and $\tau_B$ are fixed. 
	Since $\tau_B$ is fixed, $N_1$ is fixed. 
	Since $\partial_\alpha \alpha_{N_1}(\alpha,B)$
	is continuous in $\alpha$, for fixed $N_1$, 
	\begin{equation}
		\limsup_{\alpha \to \alpha_c(B)^+} |E_4| =0.
	\end{equation}
	
	It remains to bound $E_2$ and $E_6$.  We already have seen that for $\alpha$ sufficiently close
	to $\alpha_c(B)$ we will have $(\alpha_{N_2},B_{N_2})$ belong to the region $A$, Eq.~\eqref{Aregion}, where $f$ and its
	derivatives are bounded. Hence
	\begin{align}
		\label{eq:high-level}
		\| \nabla f(\alpha_{N_2},B_{N_2}) -\nabla f(\alpha_{N_2},0)] \| &= O(\|B_{N_2}\|), \\
		|m(\alpha_{N_2},B_{N_2}) - m(\alpha_{N_2},0)| &= O(\|B_{N_2}\|).
	\end{align}
	So $E_2$ and $E_6$ are $O(\|B_{N_2}\|)$.
	Since $N_2 \rightarrow \infty$ as $\alpha  \rightarrow \alpha_c(B)$,
	we have
	\begin{equation}
		\limsup_{\alpha \to \alpha_c(B)^+} (|E_2| + |E_6|) =0.
	\end{equation}
	
	This completes the proof of part (b). \end{proof}

\subsection{Comparison with the argument of Nienhuis and Nauenberg}
\label{comparison}
The discontinuity fixed point condition $r'(1/2)=b^d=16$ played a crucial role in the above proof. Apart from identifying that condition, Nienhuis and Nauenberg \cite{NienhuisNauenberg} also gave an intuitive argument of how it implies the discontinuity of magnetization. Here we wish to compare their argument and our proof. Although their argument looks simpler, there are difficulties in implementing it rigorously. 

Phrasing their argument in our language, the free energy satisfies the equation
\beq
f(\alpha,B)=16^{-1} g(\alpha,B)+16^{-1} f(\alpha',B').
\eeq
Differentiating this equation in $\alpha$ one gets
\beq
m(\alpha,B)=16^{-1} g'(\alpha,B)+16^{-1} [m(\alpha',B')\partial_\alpha \alpha'(\alpha,B)+\nabla f(\alpha',B')\cdot \partial_\alpha B'(\alpha,B)] \label{NNmag}.
\eeq
One then takes the two one-sided limits $\alpha\to \alpha_c(B)^\pm$ in this equation. Ref.~\cite{NienhuisNauenberg} only considers the case when $B$ preserves the spin flip symmetry, so that $\alpha_c(B)=0$. One then obtains
\begin{align}
	m(0^\pm,B)=16^{-1} g'(0,B)+16^{-1} [m(0^\pm,B')\partial_\alpha \alpha'(0,B)+\nabla f(0^\pm,B')\cdot \partial_\alpha B'(0,B)].
\end{align}
Here, $g'(\alpha,B)$, $\partial_\alpha \alpha'(\alpha,B)$ and $\partial_\alpha B'(\alpha,B)$ have limits since these quantities are analytic including at $\alpha=0$. On the other hand \emph{we had to assume} the existence of $m(0^\pm,B)$ and of $\nabla f(0^\pm,B')$ (by flip symmetry the last two one-sided limits coincide).

By taking the difference, several terms cancel and one obtains a simpler equation for $\Delta m(B):=m(0^+,B)-m(0^-,B)$ (cf.~\cite{NienhuisNauenberg}, Eq.~(2))
\begin{align}
	\Delta m(B)=A_0\, \Delta m(B'),\qquad A_0 = 16^{-1} \partial_\alpha \alpha' (0,B).
	\label{NN2}
\end{align}
One then iterates this equation along the RG trajectory $B_0=B, B_{j+1} =B'(0,B_j)$, and obtains (cf.~\cite{NienhuisNauenberg}, Eq.~(3))
\begin{align}
	\Delta m(B)=
	\Bigl(\prod_{j=0}^\infty  A_j \Bigr) \Delta m(0),\qquad A_j = 16^{-1} \partial_\alpha \alpha' (0,B_j).
	\label{NN3}
\end{align}
Here $\Delta m(0)$ is the zero-temperature discontinuity, $\Delta m(0)=2$ from the exact expression. By $A_j=1+O(\|B_j\|^2)$ and $B_j\to0$ exponentially fast, the infinite product $\prod_{j=0}^\infty  A_j $ converges, and one recovers $\Delta m(B)=2+O(\|B\|^2)$.

To summarize, the Nienhuis-Nauenberg argument proceeds by taking the limits $\alpha \to 0^\pm$ of \eqref{NNmag} and subtracting, thus deriving equation \eqref{NN2} for the quantity $\Delta m(B)$ one is directly interested in. Here one considers only one RG step and not the full RG trajectory. In \eqref{NN3} one considers the full trajectory, but a very simple one, belonging to the coexistence surface $\alpha=0$. Compared to our argument, the simplification is due to not having to discuss trajectories departing from the coexistence surface.

The main problem with the Nienhuis-Nauenberg argument is to justify the key assumption that the one-sided limits $m(0^\pm,B)$ do exist. We don't know how to do this without considering general RG trajectories not necessarily belonging to the coexistence surface, splitting them into segments and analyzing how each segment's contribution behaves in the limit, which is what we have done in our proof. In other words, we don't know how to implement their argument rigorously, without destroying its simplicity. In addition, their argument does not obviously extend to $B$ not respecting the spin flip symmetry.

\section{Final remarks and open problems}
\label{conclusions}

In this paper we developed a theory of tensor RG at low temperatures applicable to the 2D Ising model and its small perturbations. The Ising model at exactly zero temperature has two states of positive and negative magnetization, described in the tensor network language by tensors $A^{(q)}$, $q=\pm$, each having a single nonzero element $(A^{(q)})_{qqqq}=1$. It's easy to see that their direct sum with equal weights $\frac12 A^{(+)} \oplus \frac12 A^{(-)}$ is a fixed point of tensor RG. It's also easy to find a tensor RG transformation which acts on linear combinations with unequal weights $\alpha A^{(+)} \oplus (1-\alpha)A^{(-)}$, $\alpha\in[0,1]$. This corresponds to working at nonzero magnetic field but still at exactly zero temperature. In this paper we accomplished a less trivial task of finding a tensor RG which acts on tensors of a more general form
\beq
\alpha A^{(+)} + (1-\alpha)A^{(-)} + B,
\eeq
where the perturbation $B$ has small Hilbert-Schmidt norm. We then used this map to establish the main features of the phase diagram of the Ising model and its small perturbations at low temperature: the existence of a first-order phase transition as a function of the magnetic field, with magnetization discontinuous at the phase coexistence surface. 

We emphasized several times the connection of our results to the discontinuity fixed point idea from theoretical physics \cite{NienhuisNauenberg}. Our work is the second rigorous realization of this idea, after the work of Gawędzki, Koteck\'y and Kupiainen \cite{Gawedzki1987}. One moral of \cite{Gawedzki1987} is that one is supposed to do RG at low temperatures in the contour representations, while one may do RG at high temperatures in the spin representation. It may come as a surprise that no such change of representation is required in our approach. Moreover, our construction of tensor RG at low temperatures turns out to be almost the same as our previous construction of tensor RG at high temperature \cite{paper1}. 

We see several directions for further work:
\begin{itemize}
	\item
	In this paper we discussed extensive observables such as the free energy and the magnetization. It would be interesting to also construct infinite volume limits of correlation functions in the high-T and low-T regimes. In the high-T regime one should show that the
	finite volume limit does not depend on the boundary conditions, while it does in the low-T regime. Correlation functions at high-T were already briefly discussed in \cite{paper1}, App.~B.
	
	\item Theorem \ref{LTmap} can be extended to a situation with $n\ge 3$ phases. There exists an analytic RG map acting on tensors of the form
	\beq
	\alpha_1 A^{(1)}+\ldots+\alpha_n A^{(n)}+ B\,,\qquad \sum_i \alpha_i=1,
	\eeq
	which maps 
	\begin{align}
	&\alpha_i \to \alpha_i^{16}/\Bigl(\sum \alpha_j^{16}\Bigr) +O(B^2),\\
	& B\to B',\qquad \|B'\|<l \|B\|\quad(l<1),
	\end{align}
    with a normalization factor $\nu = \sum \alpha_j^{16}+O(B^2)$. The proof of Theorem \ref{LTmap} should generalize verbatim.
	
	\item The results of Section \ref{LTprop} should also generalize to a situation with $n\ge 3$ phases, but the dynamical system aspect of the story will be more complicated. For $n\ge 3$, the stable manifold of the fixed point $\frac1n A^{(1)}+\ldots+\frac1n A^{(n)}$ is the coexistence manifold of all $n$ phases. It has codimension $n-1$ in the space of $(\alpha_i,B)$, $\sum \alpha_i=1$. For each $S\subset \{1,\ldots,n\}$, $2\le |S|\le n-1$, we also expect the coexistence manifold of phases from $S$, which is an invariant manifold of codimension $|S|-1$ having the above stable manifold as its boundary. To work this out fully, including the (dis)continuity properties of the free energy, is an interesting open problem.
	
	\item In this work, as in \cite{paper1}, we worked on the 2D square lattice. It would be interesting to realize rigorous tensor RG on 2D lattices different from the square lattice, in the neighborhood of the high-T and low-T fixed points. It would also be interesting to extend rigorous tensor RG methods to $d\ge 3$ dimensions.
	
	\item It would be interesting to apply tensor RG methods to the Ising model with disordered couplings (bond or field disorder). The RG in the contour representation was used to prove the existence of an ordered phase for the 3D Ising model at low temperature and in weak magnetic disorder \cite{Bricmont1988} (after \cite{Imbrie1985} showed ordering at zero temperature). It would be interesting to see if tensor RG could lead to alternative proofs of these results.

\item
In this paper we focused on lattice models with discrete spin space. It is worth pointing out that lattice models with continuous spins also allow transformation to tensor network form. Take for example the 2D $O(N)$-model on the square lattice, where we have spin $n_i\in S^{N-1}$ on each lattice site, with the Hamiltonian 
\beq
H=\beta \sum_{\langle ij\rangle} n_i\cdot n_j \,.
\eeq
As in Section \ref{nnising}, we transform the partition of this model into a tensor network made of tensors
\beq
A(n_1,n_2,n_3,n_4)=e^{\beta (n_1\cdot n_2+ n_2\cdot n_3+n_3\cdot n_4+n_4\cdot n_1)}.
\eeq
A tensor is now a function from $(S^{N-1})^4$ to $\mathbb{R}$, while tensor contraction means setting arguments to the same value and integrating over the sphere. This looks different from tensors with discrete indices we had in this paper, but in fact it is the same. Abstractly, a tensor can be thought of as a multilinear map from $V\times V \times V\times V$ into $\mathbb{R}$. When $V=\ell^2(\mathbb{N})$ we get tensors with discrete indices, and when $V=L^2(S^{N-1})$ we get tensors which are functions as above. Since $\ell^2(\mathbb{N})$ and $L^2(S^{N-1})$ are isomorphic, these constructions are equivalent. In practice, we can transform the tensor from continuous to discrete form via Fourier expansion in any orthonormal basis of $L^2(S^{N-1})$. 

The RG analysis of lattice models with continuous group symmetries tends to be complicated in the spin representation \cite{Balaban1996}. It would be very interesting to see if tensor RG methods simplify the analysis. For example, can we use tensor RG to show that $O(N)$ models order at low temperatures in $d\ge 3$ dimensions?

\item
Another interesting problem concerns the analysis near the Gaussian fixed point. One classic result is the nonperturbative IR stability of the free massless scalar field in 4D with respect to the $\phi^4$ perturbation \cite{Gawedzki1985}. Can this be recovered/simplified via tensor RG?

\item
The $q$-state Potts model is known to have, as a function of the temperature, a first-order order-disorder phase transition for large $q$ \cite{Kotecky1982}. Can this result be recovered via tensor RG? This would require showing that the lattice Potts model belongs, for large $q$, to the small neighborhood of the discontinuity fixed point describing coexistence of $q+1$ phases within which the RG map can be controlled.

\end{itemize}

The construction of a nontrivial tensor RG fixed point describing the critical point of the Ising model remains our priority.

\begin{acks}[Acknowledgments]
SR cordially thanks Senya Shlosman for sharing his knowledge of the first-order phase transitions and the Pirogov-Sinai theory. SR thanks the organizers of the school "Universality in mathematical physics: random geometries, field theories and hydrodynamics" (ENS Lyon, September 2022) where this material was presented, and in particular Christophe Garban for discussions about the Peierls argument, and Jérémie Bouttier for spotting a mistake (now corrected). We thank Tom Spencer for a question which led to footnotes \ref{whereapplies},\ref{whereapplies1}. We thank Nikolay Ebel for the careful reading of the preprint and for spotting several inaccuracies. SR is supported by the Simons Foundation grant 733758 (Simons Collaboration on the Nonperturbative Bootstrap).
\end{acks}

\appendix

\section{Analytic functions on Banach spaces}\label{abstract}

Let $X, Y$ be two complex Banach spaces and let $\mathcal{U}$ be an open
subset of $X$. A function $f : \mathcal{U} \rightarrow Y$ is called analytic
(or, equivalently, holomorphic) if it has a complex Fr{\'e}chet derivative at
every point, i.e. for each $x \in \mathcal{U}$ there exists a continuous
complex-linear map $D f (x) : X \rightarrow Y$ such that
\begin{equation}
	\lim_{y \rightarrow 0} \frac{\| f (x + y) - f (x) - D f (x) y \|}{\| y \|} =
	0 .
\end{equation}
We list a few basic properties of such abstract analytic functions, assumed in this paper, often tacitly. They
are similar in spirit to properties of ordinary analytic functions from
$\mathbb{C}$ to $\mathbb{C}$.
\begin{enumerate}
	\item (Analyticity of polynomials) Let $F : X \times \cdots \times X
	\rightarrow Y$ be a continuous multilinear map of degree $n$. Then the
	function $P (x) = F (x, \ldots, x)$ is analytic. It is called a homogeneous
	polynomial of degree $n$.
	
	\item (Analyticity of product) If $h : \mathcal{U} \rightarrow \mathbb{C}$
	is a second analytic function, then $h f : \mathcal{U} \rightarrow Y$ is
	analytic.
	
	\item (Analyticity of composition) Let $Z$ be a third complex Banach space,
	let $\mathcal{V}$ be an open subset of $Y$, and let $g : \mathcal{V}
	\rightarrow Z$ be a second analytic functions. If $f (\mathcal{U}) \subset
	\mathcal{V}$, then $g \circ f : \mathcal{U} \rightarrow Z$ is analytic.
	
	\item (Infinite smoothness) $f$ has derivatives of arbitrary order at any
	point. The $n$th derivative $D^n f (x)$ at a point $x \in \mathcal{U}$ is a
	continuous multilinear map from $X \times \cdots \times X$($n$ times) into
	$Y$.
	
	\item \label{Taylor-prop}(Taylor series) Denote by $\hat{D}^n f (x)$ the
	homogeneous polynomial associated with the $n$th derivative $D^n f (x)$. Let
	$B_\varrho (x)$ denote a ball of radius $\varrho$ centered at $x$. Assume that $f$ is
	bounded in a ball $B_\varrho (x) \subset \mathcal{U}$, i.e.
	\begin{equation}
		\| f (x + w) \| \leqslant M < \infty \qquad \forall w \in B_\varrho (0).
	\end{equation}
	Then $f$ has in $B_\varrho (x)$ a uniformly convergent Taylor series expansion:
	\begin{equation}
		f (x + w) = \sum_{n = 0}^{\infty} \frac{1}{n!} \hat{D}^n f (x) (w) \qquad
		\forall w \in B_\varrho (0).
	\end{equation}
	Individual terms in this expansion are bounded by:
	\begin{equation}
		\| \hat{D}^n f (x) (w) \| \leqslant n! \frac{M}{\varrho^n} \| w \|^n.
	\end{equation}
	In particular, the operator norm of the first derivative is bounded by:
	\begin{equation}
		\| D f (x) \|_{\mathcal{B} (X, Y)} \leqslant M/\varrho. \label{Taylor-eq}
	\end{equation}
	\item The sum of a uniformly convergent series of analytic functions is
	analytic.
\end{enumerate}
For more details about these facts see e.g. {\cite{hille1996functional}},
Chapters 13 and 26, and {\cite{Harris2003,analytic}}. In particular, Property
\ref{Taylor-prop} is proven by applying the Cauchy integral representation to
the ordinary analytic function $\ell (f (x + z w))$ of a complex $z$ in the
unit disk, where $\ell \in Y^{\ast}$. 

In this paper, $X$ can be a complex Hilbert space $\mathcal{H}$ of
tensors equipped with the Hilbert-Schmidt norm, or its closed
subspace $\HLT$ obtained by setting some tensor components to zero, or $\mathbb{C} \times \HLT$ (in Theorem \ref{LTmap}). $Y$ can be $\mathbb{C},\mathcal{H},\HLT$ or $\mathbb{C} \times \HLT$.

\section{Stable manifold theorem} \label{stable}

For the convenience of the reader we give a statement of the stable manifold
theorem following \cite{irwin1970stable}. Other references include
\cite{shub2013global,lanford1991}.

\begin{theorem}
  Let $E,F$ be Banach spaces. Let $G$ be the Banach space $E \times F$
  with the norm $\|(x,y)\|= \max \{ \|x\|,\|y\| \}$.
  Let $\lambda$ be a bounded linear map of $E$ into $E$ and $\mu$
  a bounded, invertible linear map of $F$ onto $F$.
  We assume there is a constant $a<1$ such that $\|\lambda\| \le a$ and
  $\|\mu^{-1}\| \le a$. 
  Let $f$ be a $C^r$ map ($r \ge 1$) from some neighborhood $U$ 
  of $0$ in $G$ into $G$ with $f(0)=0$ and $Df(0)=\lambda \otimes \mu$.
  Then there is a neighborhood $C$ of $0$ in $E$ and a neighborhood
  $D$ of $0$ in $F$ and a  unique map $h : C \rightarrow D$ such that
  $f(\mathop{graph}(h)) \subset \mathop{graph}(h)$. The map $h$ is $C^r$ and $Dh(0)=0$. 
  For $z \in U$, $f^n(z) \rightarrow 0$
  as $n \rightarrow \infty$ if and only if $f^k(z) \in \mathop{graph}(h)$
  for some $k \ge 1$.
\end{theorem}

\begin{remark}
The stable manifold is the set of $z \in U$ such that $f^n(z) \rightarrow 0$.
So the theorem says a point $z$ is on the stable manifold if and
only if $f^k(z)$ is eventually on the graph of $h$.
\end{remark}

In our application of the stable manifold theorem, $F$ is the one-dimensional
space corresponding to $\alpha$, and $E$ is the Hilbert space corresponding
to $B$. The point $(\alpha,B)=(1/2,0)$ is a fixed point of $R$, and 
Theorem \ref{DRthm} implies that at this fixed point
$\partial_\alpha \alpha'=16$, $\partial_B \alpha'=0$, $\partial_\alpha B=0$,
and $\partial_B B'$ has norm bounded by $l<1$.
The following proposition describes the flow 
under $R$ of points not on the stable manifold.

\begin{prop}
  \label{rgflow}
  Let $\O$ be defined as in Eq.~\eqref{deltachoice}.
  (So $\O$ consists of small neighborhoods
  about $(\alpha,B)=(0,0)$ and $(1,0)$.)
  Let $(\alpha_0,B_0)$ be such that $\|B_0\|<\epsilon$ and
  $\alpha_0 \neq \alpha_c(B_0)$. Let $(\alpha_j,B_j)$ be the RG trajectory
  starting at $(\alpha_0,B_0)$. So $(\alpha_j,B_j)=R(\alpha_{j-1},B_{j-1})$.
  Then eventually $(\alpha_j,B_j) \in \O$. If $\alpha_0>\alpha_c(B_0)$ then
  $\alpha_j>\alpha_c(B_j)$ for all $j$ for which the trajectory is defined.
  (And if $\alpha_0<\alpha_c(B_0)$ then  $\alpha_j < \alpha_c(B_j)$.)
\end{prop}

\begin{proof}
Since $r'(1/2)=16$, we can choose $\eta$ small enough that $r'(\alpha) \ge 15$
on $|\alpha-1/2| < \eta$.
Then given the bounds on the derivative of the RG map in Theorem \ref{DRthm},
it is straightforward to show that for sufficiently small $\|B_0\|$,
once $\alpha_j$ satisfies $|\alpha_j-1/2| \ge \eta$ then
the trajectory will flow into $\O$.

The non-trivial part of the proposition is controlling
the flow when the trajectory is near the stable manifold. 
So in the following we only consider the portion of the trajectory
with $|\alpha_j-1/2|<\eta$. 

We start by defining a sequence of points on the stable manifold that we
will use to get estimates on the RG trajectory $(\alpha_j,B_j)$. 
The RG map takes a point on the stable manifold to the stable manifold.
So we can write $R(\alpha_c(B_{j-1}),B_{j-1}))$ as $(\alpha_c(\oB_j),\oB_j)$.
for some $\oB_j$. So $(\alpha_c(\oB_j),\oB_j)$ is a sequence of points
on the stable manifold, but it is not a trajectory of the RG map. 
We consider the case that $\alpha_0>\alpha_c(B_0)$.
We will derive a lower bound on $\alpha_j-\alpha_c(B_j)$ that shows this
quantity is always positive and grows in such a way that eventually
$\alpha_j-1/2 \ge \eta$
We will make frequent use of the equations
\begin{eqnarray}
  R(\alpha_{j-1},B_{j-1}) &=& (\alpha_j,B_j) , \\
  R(\alpha_c(B_{j-1}),B_{j-1}) &=& (\alpha_c(\oB_j),\oB_j)  .
\end{eqnarray}
Since $\|B_0\| < \epsilon$, we have $\|B\|=O(\epsilon)$ for all the $B$ that appear
in the following.

First consider $\alpha_j-\alpha_c(\oB_j)$. We write it as the integral of the
derivative of the RG map and use the bounds on that derivative in Theorem
\ref{DRthm}.
\begin{eqnarray}
  \alpha_j - \alpha_c(\oB_j) &=& \int_{\alpha_c(B_{j-1})}^{\alpha_{j-1}}
  \partial_\alpha \alpha'(\alpha,B_{j-1}) \, d \alpha 
  = \int_{\alpha_c(B_{j-1})}^{\alpha_{j-1}}
  [r'(\alpha) + O(\epsilon))]  \, d \alpha \\
  &=& r(\alpha_{j-1}) - r(\alpha_c(B_{j-1}))
  + O(\epsilon) [\alpha_{j-1} - \alpha_c(B_{j-1})]. \nn
\end{eqnarray}

Next we consider $\alpha_c(\oB_j)-\alpha_c(B_j)$.
Our bounds on the derivative of $R$ and the definitions of $B_j$ and $\oB_j$
imply $\|\oB_j-B_j\|\le C |\alpha_{j-1}-\alpha_c(B_{j-1})|$ for some
constant $C$. 
The stable manifold theorem and our bounds on the derivative of $R$
imply $\|\partial_B \alpha_c \| = O(\epsilon)$.
So
\begin{equation}
  |\alpha_c(\oB_j)-\alpha_c(B_j)|
= O(\epsilon) |\alpha_{j-1}-\alpha_c(B_{j-1})|.
\end{equation}
Writing $\alpha_j-\alpha_c(B_j)$ as 
$(\alpha_j-\alpha_c(\oB_j)) + (\alpha_c(\oB_j)-\alpha_c(B_j))$, we have shown
\begin{equation}
  \alpha_j-\alpha_c(B_j) \ge r(\alpha_{j-1}) - r(\alpha_c(B_{j-1}))
  + O(\epsilon) \, [\alpha_{j-1} - \alpha_c(B_{j-1})].
\end{equation}

By the choice of $\eta$, for $|\alpha-1/2| < \eta$ we have $r'(\alpha) \ge 15$.
And we can assume that the $O(\epsilon)$ coefficient is no bigger than $1$.
So the above equation implies 
\begin{equation}
\label{geomgrowth}  \alpha_j-\alpha_c(B_j) \ge 14[\alpha_{j-1} - \alpha_c(B_{j-1})].
\end{equation}
By induction the above inequality implies that $\alpha_j>\alpha_c(B_j)$
if $\alpha_0>\alpha_c(B_0)$ as long as $|\alpha_{j-1} -1/2|<\eta$. Moreover,
$\alpha_j-\alpha_c(B_j)$ grows geometrically, and so eventually 
$\alpha_{j-1} -1/2 \ge \eta$.
\end{proof}

\bibliographystyle{utphys}
\bibliography{ising_lowT}

\end{document}